\documentclass[12pt]{ociamthesis}  

\usepackage[left=32mm,top=25mm,right=22mm,bottom=25mm]{geometry}

\usepackage{cite}
\usepackage{color}
\usepackage{amsmath}
\usepackage{bbm}
\usepackage{verbatim}   
\usepackage{subfigure}  
\usepackage{acronym}
\usepackage{amsfonts}
\usepackage{amssymb}
\usepackage{mathrsfs}
\usepackage{graphicx}
\usepackage{multirow}
\usepackage{slashed}
\usepackage{amsthm}

\usepackage{amsthm}

\usepackage{appendix}
\usepackage{chngcntr}
\usepackage{etoolbox}
\usepackage{lipsum}

\newcommand{\half}{\frac{1}{2}}

        \newcommand{\pb}{{\mathrm {pb}}}

\newcommand{\IC}{\mathbb{C}}

\def\bar#1{\overline{#1}}

\def\inv{^{\raise.15ex\hbox{${\scriptscriptstyle -}$}\kern-.05em 1}}
\def\lbar{{\lower.35ex\hbox{$\mathchar'26$}\mkern-10mu\lambda}} 
\def\e#1{{\rm e}^{^{\textstyle#1}}}

\def\to{\rightarrow}

\let\p=\partial
\let\<=\langle
\let\>=\rangle

\let\+=\uparrow

\def\a{\alpha'}

\def\d{\textrm{d}}
\def\End{\textrm{End}}
\def\Im{\textrm{Im}}
\def\ker{\textrm{ker}}
\def\J{\mathcal{J}}
\def\bp{\bar{\partial}}
\def\bD{\bar D}
\def\tr{\textrm{tr}}
\def\OO{\mathcal{O}}
\def\D{\mathcal{D}}
\def\CS{\mathcal{CS}}
\def\A{\mathcal{A}}
\def\Q{\mathcal{Q}}
\def\L{\mathcal{L}}
\def\H{\mathcal{H}}
\def\Z{\mathcal{Z}}
\def\X{\mathcal{X}}
\def\B{\mathcal{B}}
\def\K{\mathcal{K}}
\def\P{\mathcal{P}}
\def\G{\mathcal{G}}
\def\F{\mathcal{F}}
\def\V{\mathcal{V}}
\def\R{\mathcal{R}}
\def\H{\mathcal{H}}
\def\CS{{\cal CS}}
\def\M{\mathcal{M}}

\def\g{\mathbf{g}}

\newtheorem{Theorem}{Theorem}
\newtheorem{Conjecture}{Conjecture}
\newtheorem{Proposition}{Proposition}

\theoremstyle{definition}

\newcommand{\quot}[1]{``#1''}

\numberwithin{equation}{section}

\usepackage{slashed} 
\def\d{\textrm{d}}
\newcommand{\DD}{\mathrm{D}}
\newcommand{\I}{\mathrm{i}}
\def\tr{\textrm{tr}\,}
\def\e{e}
\def\p{\partial}

\newcommand{\SU}[1]{\ensuremath{\mathrm{SU}(#1)}}
\newcommand{\U}[1]{\ensuremath{\mathrm{U}(#1)}}
\newcommand{\Sp}[1]{\ensuremath{\mathrm{Sp}(#1)}}
\newcommand{\Gtwo}{\ensuremath{\mathrm{G}_{2}}}
\newcommand{\vol}{\mathcal{V}_0}

\usepackage{appendix}
\usepackage{chngcntr}
\usepackage{etoolbox}
\usepackage{lipsum}

\AtBeginEnvironment{subappendices}{
\bigskip
\section*{\LARGE{Appendix}}
\addcontentsline{toc}{section}{Appendices}
\counterwithin{figure}{section}
\counterwithin{table}{section}
\counterwithin{equation}{section}
}

\title{Moduli in General $SU(3)$-Structure\\[1ex]     
        Heterotic Compactifications}   

\author{Eirik Eik Svanes}             
\college{Balliol College}  

\degree{Doctor of Philosophy}     
\degreedate{Trinity 2014}         

\begin{document}


\baselineskip=18pt

\setcounter{secnumdepth}{3}
\setcounter{tocdepth}{3}

\maketitle                  

\begin{dedication}
\vspace*{5cm}\Large{\it  To my Family}
\end{dedication}        

\begin{abstract}
In this thesis, we study compactifications of ten-dimensional heterotic supergravity at $\OO(\a)$, focusing on the moduli of such compactifications. We begin by studying supersymmetric compactifications to four-dimensional maximally symmetric space, commonly referred to as the Strominger system. The compactifications are of the form $M_{10}=M_4\times X$, where $M_4$ is four-dimensional Minkowski space, and $X$ is a six-dimensional manifold of what we refer to as heterotic $SU(3)$-structure. We show that this system can be put in terms of a holomorphic operator $\bD$ on a bundle $\Q=T^*X\oplus\End(TX)\oplus\End(V)\oplus TX$, defined by a series of extensions. Here $V$ is the $E_8\times E_8$ gauge-bundle, and $TX$ is the tangent bundle of the compact space $X$. We proceed to compute the infinitesimal deformation space of this structure, given by $T\M=H^{(0,1)}(\Q)$, which constitutes the infinitesimal spectrum of the lower energy four-dimensional theory. In doing so, we find an over counting of moduli by $H^{(0,1)}(\End(TX))$, which can be reinterpreted as $\OO(\a)$ field redefinitions. In the next part of the thesis, we consider non-maximally symmetric compactifications of the form $M_{10}=M_3\times Y$, where $M_3$ is three-dimensional Minkowski space, and $Y$ is a seven-dimensional non-compact manifold with a $G_2$-structure. We write $X\rightarrow Y\rightarrow\mathbb{R}$, where $X$ is a six dimensional compact space of half-flat $SU(3)$-structure, non-trivially fibered over $\mathbb{R}$. These compactifications are known as domain wall compactifications. By focusing on coset compactifications, we show that the compact space $X$ can be endowed with non-trivial torsion, which can be used in a combination with $\a$-effects to stabilise all geometric moduli. The domain wall can further be lifted to a maximally symmetric AdS vacuum by inclusion of non-perturbative effects in a heterotic KKLT scenario. Finally, we consider domain wall compactifications where $X$ is a Calabi-Yau. We show that by considering such compactifications, one can evade the usual no-go theorems for flux in Calabi-Yau compactifications, allowing flux to be used as a tool in such compactifications, even when $X$ is K\"ahler. The ultimate success of these compactifications depends on the possibility of lifting such vacua to maximally symmetric ones by means of e.g. non-perturbative effects.
\end{abstract}          

\begin{acknowledgements}
I would like to begin by expressing my gratitude to both of my supervisors, Doctor Xenia de la Ossa and Professor Andr\'e Lukas. I am very grateful for all the support and encouragement you have given me, and for the extremely interesting and excitingly challenging projects you have introduced me to, both from a phenomenological and mathematical perspective. I am both lucky, and very happy to have had the opportunity to work with you both. Thank you for all you have done for me.

I would also like to thank my other collaborators Magdalena Larfors and Michael Klaput at Oxford. Both of whom have given me much support and encouragement during my time as a DPhil student. I am also grateful to Spiro Karigiannis, for his mathematical insights and continued support.

I would also like to thank Lara Anderson, James Gray, Philip Candelas, Rhys Davies, Yang Hui He, Burt Ovrut, Ruxandra Moraru, Ulrike Tillmann, for useful conversations on String Theory and other topics. I am especially grateful to Lara, James and Yang, for their tips, help and support during the final stages of my DPhil, and to Lara and James for their help and hospitality during my seminar trip to the United States.

I am also grateful to my fellow graduate students Edward Hardy, Saso Grozdanov, Setphen Angus, Cyril Matti, David Kraljic, Andrew Powell, Andrei Constantin, Challenger Mishra, Richard Lau, Kyle Allison, Jakub Sikorowski, Martin Fluder, Mathew Bullimore, Tim Adamo, and others for helpful support, friendship and encouragement. I also thank my friends at Balliol College MCR for their support. 

Finally, I would like to thank Oxford University and Balliol College for providing me with Scholarships, without which this thesis would have been impossible to complete.

\end{acknowledgements}   

\begin{romanpages}          
{\tableofcontents}            
\end{romanpages}            




\chapter{Introduction}


\section{Strings and Supergravity}
String theory has as its core assumption the idea that the fundamental constituents of nature are one-dimensional strings, rather than point particles. Originally introduced as a theory of the strong interaction \cite{Susskind:1970xm, Nielsen:1971ke}, it was later found to have far reaching consequences for the unification of gravity and the other fundamental forces of nature. Indeed, in string theory, it is the vibrations of the string that produce the low energy particle spectrum. Different vibration modes, or excitations, correspond to different particles. The graviton, which is the fundamental particle associated to the space-time metric, is then just another excitation of the string, giving a natural unification of the fundamental forces.

Looking very promising at first, there was of course bound to be complications associated to the theory. Firstly, in order to avoid tachyons, it is necessary to introduce world-sheet supersymmetry\footnote{This is not necessarily an issue, as it is really space-time (target space) supersymmetry which is broken at observable scales.}, giving rise to superstring theory. Secondly, the string being a one-dimensional object, it naturally lives in $(1+1)$ dimensions, commonly known as the world-sheet. The observed space-time is then an artefact of the vibrating string, referred to as the target space. In the case of superstring theory, in order to be a consistent theory quantum-mechanically, this space-time must be ten-dimensional . The low energy effective theory of string theory is thus a ten-dimensional supergravity. This supergravity arises as a double expansion of the world-sheet theory in $\a$, the string tension parameter, and $g_s$, the string coupling constant \cite{green1987superstring2, polchinski1998string1, polchinski1998string2}.

String theory then naturally includes many of the more exotic features of beyond Standard Model physics, such as supersymmetry \cite{ramond1971dual, gervais1971field, sohnius1985introducing}, extra dimensions \cite{kaluza1921unitatsproblem, klein1926quantentheorie, witten1981search}, grand unified theories (GUTs) \cite{Georgi:1974sy}, etc. There are five consistent low energy supergravities descending from the superstring. These are the type II theories, type IIA and type IIB, type I theory and the heterotic supergravities with gauge group $SO(32)$ or $E_8\times E_8$. The theories are related by a network of dualities \cite{schwarz1997lectures, forste1998duality}, most notably of these perhaps is Mirror Symmetry \cite{kikkawa1984casimir, sakai1986vacuum, candelas1991exactly}, a form of T-duality relating compactifications of type IIA on a manifold $X$ to compactifications of type IIB on a different manifold $\tilde X$ (it's mirror), so that the lower-dimensional physics is the same. 

The different string theories are conjectured to have a strong coupling completion in a framework known as M-theory \cite{witten1995string, hovrava1996eleven}, and in this sense the supergravities arise as different limits of this now eleven-dimensional theory. Looking very promising, there are however several open problems yet to be worked out in M-theory. Most notably is perhaps that of the world-volume theory of M5-branes and its relation to the Langlands program \cite{kapustin2006electric, gukov2006gauge, witten2009geometric}. A full understanding of the fundamental dynamics of M-theory is therefore far from complete.

We shall not discuss M-theory in any greater detail here, as the main focus of this thesis is the heterotic string, specifically the $E_8\times E_8$ heterotic supergravity. We shall mainly be concerned with $\a$-corrections to heterotic supergravity compactifications, non-standard compactifications of this theory, and moduli of heterotic compactifications. We will not be concerned with $g_s$-corrections in this thesis.

Heterotic string theory is very attractive in terms of standard-model model building. Indeed, the Standard Model naturally embeds in the gauge group $E_8$, providing fertile model building scenarios for beyond the standard model physics, as first noticed in \cite{Candelas:1985en}. Moreover, the other ``hidden sector" $E_8$ may be used to cancel additional anomalies, or help with moduli-stabilisation, as we shall see explicitly in chapter \ref{ch:HFNK}. Recall that supergravities derived from string theory naturally live in ten-dimensions. This might seem like a problem at first, but it turns out to also have advantages as we shall now explain. In order to do phenomenology with such theories, one usually assumes the ten-dimensional space-time $M_{10}$ to have the form of a fibration
\begin{equation}
\label{eq:fibr}
X_6\rightarrow M_{10}\xrightarrow{\pi} M_4\:,
\end{equation}
as its most general form. Here $X_6$ is a six-dimensional compact space, $M_4$ is four-dimensional space-time, and $\pi$ is the projection onto the non-compact base $M_4$. The compact space $X_6$ is assumed small and un-observable in accelerators, with an internal radius inversely proportional to the GUT scale or above. 

Often, the vacuum configuration of the space-time $M_4$ is taken to be maximally symmetric Minkowski space, and the fibration \eqref{eq:fibr} in this case reduces to a direct product
\begin{equation}
\label{eq:maxcomp}
M_{10}=X_6\times M_4\:.
\end{equation}
This is not the case however in general, and we shall discuss scenarios in Part~\ref{part:3d} of the thesis, where the space $X_6$ is fibered over one of the non-compact legs of $M_4$, often referred to as domain-wall solutions.

The compactifications described above are often referred to as Kaluza-Klein reductions \cite{duff1986kaluza}. One may wonder what the correct choice of compact space $X_6$ is. Firstly, it should solve the supergravity equations of motion derived from the ten-dimensional action.\footnote{We will not deal with quantum corrections to the ten-dimensional theory in this thesis.} Secondly, for phenomenological reasons but also in order to have some mathematical control over the solutions, we also wish to preserve four-dimensional supersymmetry. Supersymmetry, if it exists, provides a natural solution to the hierarchy problem of particle physics \cite{dimopoulos1981softly, dimopoulos1981supersymmetry}. The theories we consider in this thesis live at the GUT scale, upon compactification, where supersymmetry is unbroken. To break supersymmetry at this scale, would mean reintroducing the hierarchy problem, which is unattractive from a phenomenological point of view. For the most part in this thesis we will therefore not deal with supersymmetry breaking models, unless explicitly stated.

In heterotic theory we also have the $E_8\times E_8$ gauge group, which fibers non-trivially over the compact space $X_6$. It is precisely this gauge group that makes the theory so phenomenologically attractive. Indeed, the topology of this gauge group compactification should determine the net number of generations, the field content and the Yukawa couplings of the low energy theory. How this works has long been known to lowest order in $\a$, where $X_6$ is a Calabi-Yau\cite{strominger1985yukawa, candelas1991moduli, Anderson:2009ge}, and it is one purpose of this thesis to extend some of these results to higher orders. In particular, in chapter \ref{ch:SS} we work out the infinitesimal four-dimensional spectrum at $\OO(\a)$.

The heterotic string is thus a fertile ground for Standard Model building, and this is a direction that has been explored to a large degree in recent years. Compactifications to a maximally symmetric space-time of the type \eqref{eq:maxcomp} have been of particular interest \cite{Donagi:2004ia, Bouchard:2005ag, Anderson:2008uw, Anderson:2009mh, Anderson:2011ns, Braun:2011ni, Anderson:2012yf}, which to lowest order in $\a$ are compactifications on Calabi-Yau manifolds. Such compactifications always endure the problem of moduli stabilisation. Mathematically, the moduli space $\M$ of a compactification corresponds to the allowed deformations of the geoemetry that preserve the equations of motion, and supersymmetry in particular. These moduli then give rise to fields in the low energy theory, corresponding to flat directions in the four-dimensional potential. These fields are not observed in accelerators, and they must therefore be given a mass in order to lift them from the low energy spectrum. Progress in this direction has been made in recent years \cite{Anderson:2009nt, Anderson:2010mh}, but a maximally symmetric compactification with all moduli stabilised appears difficult to achieve \cite{Anderson:2011cza}.  

In this thesis, we want to remedy the problem of moduli stabilisation by allowing for compactifications of type \eqref{eq:fibr}. This allows the compact space $X_6$ to be of more exotic type, in particular it can be non-K\"ahler. Focusing on a particular type of compact space known as cosets, with a combination of $\a$-effects, we find that all geometric moduli can be stabilised perturbatively. Further, with inclusion of non-perturbative effects, we find that all moduli can be stabilised. Additionally, we will also consider Calabi-Yau compactifications to these more generic space-times, showing that they allow for a more flexibility in moduli stabilisation. In particular, by allowing for fluxes.


We now turn to the outline of the thesis, giving further motivation for each chapter. Before we begin, we also note that Heterotic supergravity has its UV completion in terms of a two-dimensional world-sheet theory, or sigma-model \cite{Hull1986357, Sen1986289, Hull1986187}. However, as this thesis is concerned with the ten-dimensional perspective we will not go into further details on this here. It should however be noted that one has to consider this framework in order to properly include loop corrections, i.e. $g_s$-corrections, to the structures presented in this thesis, and work in this direction is underway.

\section{Outline and Motivation}
As described above, we will be concerned with heterotic string compactifications in this thesis. In particular, compactifications for which the space-time vacuum configuration takes the form \eqref{eq:fibr}, which we rewrite as
\begin{equation*}
M_{10}=M_k\times X_{10-k}\:,
\end{equation*}
where $M_k$ is $k$-dimensional Minkowski space, and $X_{10-k}$ is a $(10-k)$-dimensional (possibly non-compact) manifold with non-trivial structure. In this thesis, we will specialise to the cases $k=\{3,4\}$, but more general cases can be considered. 

\subsection{Moduli in Maximally Symmetric Minkowski Compactifications}
We start in Part~\ref{part:4d} of the thesis by setting $k=4$, and considering supersymmetric solutions where now $X_6$ is compact. A lot is known about the lower energy effective four-dimensional theory at zeroth order in $\a$, where $X_6$ is a Calabi-Yau. Complications, however, arise in the higher order theory. In particular, the internal geometry need no longer be K\"ahler, and is torsional in general. A lot of the tools coming from K\"ahler and Calabi-Yau geometry are therefore lost. Moreover, the inclusion of the heterotic Bianchi identity, needed for anomaly cancellation, also complicates matters. A rather trivial identity at zeroth order where it states that the Neveu-Schwarz (NS) flux $H$ is closed, it becomes vastly more complicated at first order where it couples the flux with the gauge and gravitational sectors of the theory.

The form of the general supersymmetric solutions at $\OO(\a)$ where first written down in the 80's, by Strominger, Hull, de\,Wit\,{\it et\,al} and L\"ust \cite{Strominger:1986uh, Hull:1986kz, deWit:1986xg, Lust:1986ix}, and has since been known as the Strominger system. These are solutions where the base $X_6$ has a geometry of what we come to call a {\it heterotic $SU(3)$-structure}, which are complex conformally balanced manifolds with vanishing first Chern class. One particular feature of these solutions is that the torsion of the internal space gets identified with the flux. We will return to this in Part~\ref{part:3d} when discussing moduli stabilisation.

Although a lot of progress has been made in recent years \cite{Witten:1986kg, Dasgupta:1999ss, 1999math......1090A, Ivanov:2000fg, 2001CQGra..18.1089I, 2001math......2142F, Becker:2002sx, Gauntlett:2003cy, Becker:2003yv, Becker:2003sh, LopesCardoso:2003sp, LopesCardoso:2003af, Becker:2006xp, Becker:2009df,  Andreas:2010cv, 2012CMaPh.315..153A, 2013arXiv1304.4294G},\footnote{There has also been a lot of research into this topic from the two-dimensional world-sheet point of view, see e.g. \cite{Adams:2006kb, Sharpe:2008rd, Kreuzer:2010ph, McOrist:2010ae, Beccaria:2010yp, NibbelinkGroot:2010wm, McOrist:2011bn,  Blaszczyk:2011ib, Quigley:2011pv, Adams:2012sh, Nibbelink:2012wb, Melnikov:2012hk, Quigley:2012gq, Melnikov:2012nm}.} and particular examples of non-K\"ahler solutions have been found \cite{Becker:2006et, Fu:2006vj}, there is still a lot unknown about the effective four-dimensional supergravity. From a phenomenological point of view, a goal of heterotic supergravity is to relate it to some four-dimensional $N=1$ supergravity, with superpotential $W$ and K\"ahler potential $K$. In order to get to this, a Scherk-Schwarz type dimensional reduction needs to be performed \cite{scherk1979get}. For the heterotic string, partial reductions have been performed in the past, both for K\"ahler \cite{witten1985dimensional, candelas1991moduli} and non-K\"ahler cases \cite{Lust:1985be, Kapetanakis19924, Gurrieri:2004dt, KashaniPoor:2007tr, chatzistavrakidis2009reducing, zoupanos2, Lukas:2010mf}. These usually exclude the bundles and Bianchi identity, focusing on the geometric sector of the theory. A complete dimensional reduction and knowledge of the theory is therefore still lacking, even in the K\"ahler case. 

Worse still, before a dimensional reduction can be performed, knowledge is needed of what the moduli space of the compactification is. This has long been known for the geometric base $X_6$ in the case of Calabi-Yau compactifications \cite{candelas1991moduli}, and for special cases including the bundles and Bianchi identity \cite{witten1985dimensional, strominger1985yukawa, Berglund:1995yu}, but has so far been lacking for the $\a$-corrected Strominger system. Due to the non-K\"ahlerness of the compact space $X_6$, conventional methods seem hard to apply, and the non-trivial Bianchi identity, involving all parameters of the theory, also severely complicates such an approach. Partial attempts have been made to get a clue of the form of the moduli space, often excluding sectors such as the gauge bundle and the Bianchi Identity \cite{Becker:2005nb, Lapan2006, Adams:2009zg, Melnikov:2011ez}. However, this often leads to nonsensical infinite-dimensional results, as not all the conditions for a compactification are included. Knowledge of the full $\a$-corrected moduli space, even the infinitesimal moduli space, has therefore remained an open problem for the past 25 years, even in Calabi-Yau compactifications.

\subsubsection*{Moduli of the Strominger System}
We shall make some progress in this direction by deriving the infinitesimal moduli of the theory, or equivalently the tangent space $T\M$ of the moduli space. We relate this to certain cohomologies of holomorphic bundles over $X_6$. This then constitutes the infinitesimal spectrum of the lower energy four-dimensional theory, and this is a first step in the direction of deriving this theory. We perform this calculation in chapter \ref{ch:SS}. Of course, in order to fully know what the four-dimensional theory is, one also needs knowledge of the superpotential K\"ahler potential, Yukawa couplings, etc. In order to find these, one needs to do the dimensional reduction of the the ten-dimensional theory \cite{McOrist2014, Candelas2014}, which is beyond the scope of this thesis. 

To be more specific, we will see that the Strominger system can be put in terms of a holomorphic structure $\bD$ on a bundle $\Q$ over the complex base $X_6$. What is important for this structure is that {\it it is holomorphic, i.e. $\bD^2=0$, if and only if the relevant Bianchi Identities are satisfied.} Indeed, as we shall see, the heterotic Bianchi identity is naturally included in $\bD$. Moreover, the non-K\"ahler hermitian form is also given as part of this holomorphic structure, circumventing the issues concerning the non-K\"ahlerness of $X_6$. The spectrum is then computed as a cohomology
\begin{equation*}
T\M=H^{(0,1)}_{\bD}(\Q)\:.
\end{equation*}
The bundle $\Q$ is defined by a series of extensions, which in turn means that $H^{(0,1)}_{\bD}(\Q)$ can be computed by the machinery of long exact sequences in cohomology. We therefore find that $H^{(0,1)}_{\bD}(\Q)$ is given as a subset of the usual cohomologies $H^{(0,1)}(TX)$, $H^{(0,1)}(\End(V))$, $H^{(0,1)}(T^*X)$ and $H^{(0,1)}(\End(TX))$, counting complex structure moduli, bundle moduli, hermitian moduli, and connection moduli respectively. 

We note that putting the system in terms of a holomorphic might prove useful in further determining what the full moduli space $\M$ and the corresponding four-dimensional theory is. For one, obstructions to the infinitesimal deformations should be computed by the second cohomology $H^{(0,2)}_{\bD}(\Q)$, and these are expected to give rise to a non-trivial superpotential in the four-dimensional theory  \cite{Berglund:1995yu, Anderson:2011ty}. Secondly, K\"ahler metrics on moduli spaces of complex bundles have long been of interest to mathematicians, though usually with $X_6$ K\"ahler, see e.g. \cite{kobayashi1987differential, kim1987moduli, schumacher1993moduli, huybrechts2010geometry}. These results can potentially be extended to the Strominger system, viewed as a holomorphic structure $\bD$. However, in order to take full advantage of the mathematical framework, a lot of results would need to be extended to the non-K\"ahler case. This is left for future work.

\subsubsection*{Connections and Field Redefinitions}
In deriving the spectrum, we also encounter an ambiguity relating to a choice of connection $\nabla$ on the tangent bundle $TX$. More specifically, we encounter ``moduli", $H^{(0,1)}(\End(TX))$,  relating to deformations of this connection. This connection is a priori is a function of the other fields, and there should therefore not be such moduli related to this connection. This apparent discrepancy deserves more attention, and we devote chapter \ref{ch:connredef} to an explanation of these moduli. Specifically, we find that they correspond to field redefinitions. More precisely, they parametrize {\it how $\nabla$ depends on the other fields}. We also consider the higher order theory in this chapter, showing that the structure we find at $\OO(\a)$ survives to $\OO(\a^2)$ as well without any major modifications. 
\\\\
Part \ref{part:4d} of the thesis is based on the following papers:

\begin{itemize}

\item X. de la Ossa and E.E. Svanes, \it Holomorphic Bundles and the Moduli Space of $N=1$ Supersymmetric Heterotic Compactifications\rm, (2014), arXiv:1402.1725, Published in: JHEP, 1410, 123.

\item X. de la Ossa and E.E. Svanes, \it Connections, Field Redefinitions and Heterotic Supergravity\rm, (2014), arXiv:1409.3347.

\item X. de la Ossa and E.E. Svanes, \it Generalised Hermitian Yang-Mills Connections and the Strominger System\rm, (2014), in preparation.

\end{itemize}

\subsection{Non-Maximally Symmetric Compactifications and Moduli Stabilisation}
Having discussed $\a$-corrected maximally symmetric compactifications, we turn in Part~\ref{part:3d} to discuss moduli stabilisation and non-maximally symmetric compactifications. We note that the focus of this part of the thesis is \it moduli stabilisation\rm\:, more specifically stabilisation of geometric moduli corresponding to the compact space $X_6$. We will therefore ignore bundle moduli in this part of the thesis, apart from topological consistency checks. To include the bundle moduli, one would have to do an analysis similar to Part~\ref{part:4d} of the thesis. Moreover, the bundle would have to be included in the dimensional reduction. It is beyond the scope of this thesis to do so. 

As noted above, one of the main challenges of heterotic string theory is that of moduli stabilisation. In contrast to type II theory, where Ramound-Ramound (RR) fluxes are available, the heterotic string only has NS flux. This can be used in a similar manner to type II theories to stabilise flux through a Gukov-Vafa-Witten type superpotential \cite{Gukov:1999ya}. However, as we will explicitly see in chapter \ref{ch:SS} (see also chapter \ref{ch:CY}, Theorem \ref{tm:nogo}), a maximally symmetric space-time forces $H$ to be related to the torsion of the internal space, which leads us back to the Strominger system, where the geometry is necessarily non-K\"ahler. Although in principle not a problem, and we make modest progress with this case in Part~\ref{part:4d} of this thesis, we would also like to study scenarios where the internal geometry can be K\"ahler, without sacrificing flux. This forces us to consider non-maximally symmetric compactifications.

We also want to allow for more exotic types of geometries, where the internal space has torsion, which can be used in a similar manner as flux to stabilise moduli. As we shall see, this also leads us to consider compactifications to non-maximally symmetric space-times. For concreteness, we focus on the case $k=3$, where then $X_7$ has a $G_2$-structure with one non-compact leg, i.e. domain walls. This leads the compact part of $X_7$, $X_6$, to have the geometry of {\it half-flat} manifolds, to be defined in chapter \ref{ch:nonmaxcomp}.

\subsubsection*{Torsional Compactifications and Moduli Stabilisation}
In chapter \ref{ch:HFNK} we will study compactifications on a particular type of half-flat manifolds. Namely that of coset geometries. These are geometries of the type $G/H$, where $G$ is a Lie-group, and $H$ is a subgroup of $G$. Studying cosets have several advantages. In particular, their geometries are usually quite constrained, which means that there are less moduli to stabilise. In particular, in the examples we consider, there are no complex structure moduli. Secondly, their geometry has a nice description in terms of $G$-invariant forms. With this, it turns out that we can write a lot of expressions, and do a lot of computations, explicitly, which is not the case for Calabi-Yau's. The upshot is that we don't need to resort to hard theorems like e.g. the Donaldson-Uhlenbeck-Yau theorem \cite{0529.53018, MR997570}. In this thesis we focus on the coset $SU(3)/U(1)^2$, as this turns out to be most phenomenologically relevant. Other cosets can however also be considered, see e.g. \cite{Klaput:2011mz, CyrilThesis, Klaput:2012vv}.

We will consider the gravitational sector of the theory to $\OO(\a)$, i.e. the geometry of the internal space $X_6$. We will see that, through the heterotic Bianchi identity, a non-trivial flux is induced at this order, which in turn makes it possible to stabilise all geometric moduli perturbatively. By including non-perturbative effects, it is also possible to stabilise the axio-dilaton and arrive at a maximally symmetric vacuum in a heterotic KKLT type scenario \cite{Kachru:2003aw}. It should be noted that domain wall compactifications of the type we consider \cite{Gauntlett:2001ur, Gauntlett:2002sc, Gran:2005wf, Papadopoulos:2009br}, in particular coset compactifications\cite{zoupanos2, cardoso2003non, micu2004heterotic, chatzistavrakidis2009reducing, Gurrieri:2004dt, Gurrieri:2007jg, chatzistavrakidis2009dimensional, Lukas:2010mf, chatzistavrakidis2012nearly, Gray:2012md, 1998math.ph...7026O, castellani2001g, Kapetanakis19924, 0264-9381-5-1-011, Roberto19901, Castellani1984394, Lust:1985be, KashaniPoor:2007tr}, have appeared in the literature before, and our work is merely a continuation of an ongoing~story.

\subsubsection*{Calabi-Yau Compactifications with Flux}
In chapter \ref{ch:CY} we return to the case of K\"ahler geometries, more specifically Calabi-Yau's. 
As noted above, the presence of NS flux in maximally symmetric heterotic compactifications leads to internal manifolds which are complex but non-K\"ahler. This departure from Calabi-Yau manifolds, while not a problem in principle, constitutes a serious practical disadvantage. Compared to the significant body of knowledge on Calabi-Yau manifolds, not much is known about the required non-K\"ahler spaces. We attemt to make progress in this direction in Part~\ref{part:4d} of the thesis, but the construction of realistic particle physics models based on such spaces is still a long way off. It appears then, at present, that flux is of little practical use in the context of realistic heterotic model building.

In chapter \ref{ch:CY}, we wish to show how this conclusion can be avoided and to show that heterotic Calabi-Yau compactifications can indeed be consistent with the presence of NS flux. This requires dropping one of the assumptions which led to Strominger's result, namely that of maximal symmetry of the four-dimensional space-time. Instead, we will assume that four-dimensional space-time again has the structure of a domain wall, with a $2+1$-dimensional maximally symmetric world-volume and one transverse direction. As we will show, any Calabi-Yau manifold and any harmonic NS flux on this manifold can be combined with such a domain wall to form a full ten-dimensional solution of the heterotic string, at least to lowest order in $\alpha'$. This vacuum could however potentially be lifted by the use of non-perturbative effects, as in chapter \ref{ch:HFNK}. We leave this to future work. 
\\\\
Part \ref{part:3d} is based on the following papers:

\begin{itemize}

\item M. Klaput, A. Lukas, C. Matti, E. E. Svanes, \it Moduli Stabilising in Heterotic Nearly-KŠhler Compactifications\rm, (2012), arXiv:1210.5933, Published in: JHEP, 1301, 015.

\item M. Klaput, A. Lukas, E. E. Svanes, \it Heterotic Calabi-Yau Compactifications with Flux\rm, (2013), arXiv:1305.0594, Published in: JHEP 1309, 034.

\end{itemize}

\subsubsection*{Discussions and Conclusion}
We review the thesis and conclude in Part~\ref{part:concl}, summarising the most important results. We also have a discussion section at the end of each chapter, where we outline future directions and work in progress. At the end of each chapter, we also include appendices where technical details are laid out at the readers convenience.  






\section{First Order Heterotic Supergravity}

\label{sec:FirstOrder}

We now take a moment to review heterotic supergravity at first order in $\a$, as first given in \cite{Bergshoeff:1988nn, Bergshoeff1989439}. We write down the bosonic action and the corresponding supersymmetry transformations. Herby, we set up some notation which will be important for the remainder of the thesis. We also comment briefly on an ambiguity for a connection choice that appears in the action, leaving an extensive review of this issue to Chapter \ref{ch:connredef}.

\subsection{Action, Field Content, and Supersymmetry}
Let us begin by recalling the bosonic part of the action at this order \cite{Hull:1987pc, Bergshoeff1989439, Bergshoeff:1988nn}\footnote{For brevity, we omit to write the corresponding formulas for the fermionic sector for the most part of this thesis.}
\begin{align}
S=&\int_{M_{10}}e^{-2\phi}\Big[*\mathcal{R}+4\vert\d\phi\vert^2-\frac{1}{2}\vert H\vert^2-\frac{\alpha'}{4}(\tr\vert F\vert^2-\tr\vert R\vert^2)\Big]+\OO(\a^2)\:.
\label{eq:action}
\end{align}
$\mathcal{R}$ is the Einstein-Hilbert curvature term of the metric $g$, and the Hodge-dual $*$ is also defined with respect to this metric. We use the notation
\begin{equation*}
\vert C\vert^2=C\wedge *C\:,
\end{equation*}
for $C\in\Omega^p(M_{10})$, where
\begin{equation*}
C=\frac{1}{p!}C_{n_1\cdots n_p}\d x^{n_1\cdots n_p}\:.
\end{equation*}
Furthermore $F$ and $R$ are curvatures given by
\begin{equation*}
F=\d A+A\wedge A\:,\;\;\;\;R=\d\Theta+\Theta\wedge\Theta\:,
\end{equation*}
where $A\in\Omega^1(\End(V))$ is the connection of an $E_8\times E_8$ gauge bundle, and $\Theta\in\Omega^1(\End(TX))$ is the connection one-form of a tangent bundle connection $\nabla$. $B$ is a two-form field, with field strength given by the NS three-form flux,
\begin{equation}
\label{eq:anomalycancellation}
H=\d B+\CS\:,
\end{equation}
where
\begin{equation}
\label{eq:CS}
\CS=\frac{\a}{4}(\omega_{CS}^A-\omega^\nabla_{CS})\:.
\end{equation}

The $\omega_{CS}$'s are Chern-Simons three-forms of the gauge-connection $A$, and the tangent bundle connection $\nabla$,
\begin{align*}
\omega_{CS}^A&=\tr(A\wedge\d A+\frac{2}{3}A\wedge A\wedge A)\\
\omega_{CS}^\nabla&=\tr(\Theta\wedge\d\Theta+\frac{2}{3}\Theta\wedge\Theta\wedge\Theta)\:.
\end{align*}
This field strength is gauge invariant, provided the $B$-field transforms as \cite{Green1984117}
\begin{equation}
\label{eq:gaugeB}
\delta B=-\frac{\a}{4}(\tr \d A\epsilon-\tr \d\Theta\eta)\:,
\end{equation}
under gauge transformations $\d_A\epsilon$ of $A$, and $\d_\Theta\eta$ of $\Theta$, where $\epsilon\in\Omega^0(\End(V))$ and $\eta\in\Omega^0(\End(TX))$ respectively. These gauge transformations are in addition to the usual gauge transformation 
\begin{equation}
\label{eq:gaugeB0}
B\rightarrow B+\d\lambda
\end{equation}
of the $B$-field. Here
\begin{equation}
\label{eq:covderivA}
\d_A=\d+A
\end{equation}
is the covariant derivative on the bundle, with a similar expression for $\d_\Theta$. In particular, on endomorphism valued forms $\alpha\in\Omega^*(\End(V))$, this acts as
\begin{equation*}
\d_A\alpha=\d\alpha+[A,\alpha]\:,
\end{equation*}
where $[\:,\:]$ is the commutator on even forms, and anti-commutator on odd forms. Taking the exterior derivative of \eqref{eq:anomalycancellation}, we obtain the Bianchi identity
\begin{equation}
\label{eq:bianchi}
\d H=\frac{\a}{4}(\tr F\wedge F-\tr R\wedge R)\:.
\end{equation}


The fermonic fields of the theory are the gravitiono $\psi_I$, the dilatino $\lambda$ and the gaugino $\chi$. $N=1$ supersymmetry then imposes the following supersymmetry variations of these fields
\begin{align}
\label{eq:O1spinorsusy1}
\delta\psi_M &=\nabla^+_M\epsilon=\Big(\nabla^{\hbox{\tiny{LC}}}_M+\frac{1}{8}\H_M\Big)\epsilon+\OO(\a^2)\\
\label{eq:O1spinorsusy2}
\delta\lambda&=\Big(\slashed\nabla^{\hbox{\tiny{LC}}}\phi+\frac{1}{12}\H\Big)\epsilon+\OO(\a^2)\\
\label{eq:O1spinorsusy3}
\delta\chi&=-\frac{1}{2}F_{MN}\Gamma^{MN}\epsilon+\OO(\a),
\end{align}
where $\epsilon$ is a ten-dimensional Majorana-Weyl spinor parametrising supersymmetry, $\nabla^{\hbox{\tiny{LC}}}$ denotes the Levi-Civita connection with respect to the metric, and we have defined $\H_M=H_{MNP}\Gamma^{NP}$ and $\H=H_{MNP}\Gamma^{MNP}$. The $\Gamma^M$ are ten-dimensional Dirac gamma-matrices, and we use large roman letters $\{M,N,..\}$ to denote ten-dimensional indices.

Note that the transformation for the gauge field has a reduction in the order of $\a$. This is because the gauge field always appears with an extra factor of $\a$ in the action. Hence the gauge sector decouples in the limit $\a\rightarrow0$. In order to have a supersymmetry invariant action at $\OO(\a)$, we therefore only need to specify the gaugino transformation modulo $\OO(\a)$-terms. Supersymmetry for a given solution requires that the variations \eqref{eq:O1spinorsusy1}-\eqref{eq:O1spinorsusy3} are set to zero.

What the connection $\nabla$ on $TX$ is, is subtle. Firstly it cannot be a dynamical field, as there are no modes in the corresponding string theory corresponding to this. Hence, $\nabla$ must depend on the other fields of the theory in some particular way. This dependence is is forced upon us once the full supergravity action, whose bosonic part is \eqref{eq:action}, and supersymmetry transformations are specified. Indeed, the condition of a supersymmetry invariant supergravity action under the supersymmetry transformations  \eqref{eq:O1spinorsusy1}-\eqref{eq:O1spinorsusy3}, reduces this choice of connection to the particular choice of the Hull connection $\nabla^-$\cite{Hull1986357, Bergshoeff:1988nn, Bergshoeff1989439}
\begin{equation}
\label{eq:Hullconn}
{{\nabla^-}_{MN}}^{P}={{\nabla^{\hbox{\tiny{LC}}}}_{MN}}^{P}-\frac{1}{2}{H_{MN}}^{P}\:,
\end{equation}
where $\nabla^{\hbox{\tiny{LC}}}$ is the Levi-Civita connection. It should be noted that by deforming the supersymmetry transformations appropriately, this choice of connection can be relaxed. We will see precisely how this works in chapter \ref{ch:connredef}.

\subsection{Equations of Motion}
The action \eqref{eq:action} gives rise to the following set of equations of motion\footnote{In deriving these equations, one relies on a lemma by Bergshoeff and de Roo \cite{Bergshoeff1989439}, stating that the variation of the action with respect to $\nabla^-$ is of $\OO(\a^2)$, see also \cite{Andriot:2011iw}.}
\begin{align}
\label{eq:eom1}
\R-4(\nabla\phi)^2+\nabla^2\phi-\frac{1}{2}\vert H\vert^2-\frac{\a}{4}\Big(\tr\vert F\vert^2-\tr\vert R^{-2}\vert\Big)&=\OO(\a^2)\\
\label{eq:eom2}
\R_{MN}+2\nabla_M\nabla_N\phi-\frac{1}{4}H_{MPQ}{H_N}^{PQ}&\notag\\
-\frac{\a}{4}\Big(\tr\: F_{MP}{F_N}^P-\tr\: {R^-}_{MP}{R^-_N}^P\Big)&=\OO(\a^2)\\
\label{eq:eom3}
\nabla^N\Big(e^{-2\phi}H_{MNP}\Big)&=\OO(\a^2)\\
\label{eq:eom4}
e^{2\phi}\d_A(e^{-2\phi}*F)-F\wedge*H&=0\:+\OO(\a)\:.
\end{align}
Note again the reduction of orders for the gauge field equation of motion. 
The equations of motion may also be derived form the sigma-model world-sheet perspective, by demanding that the theory remains conformal, i.e. the beta functions vanish. 

We note for completeness that the supersymmetry conditions \eqref{eq:O1spinorsusy1}-\eqref{eq:O1spinorsusy3} implies the equations of motion, if and only if the connection $\nabla$ satisfies the condition \cite{Ivanov:2009rh, Martelli:2010jx}
\begin{equation}
\label{eq:instanton10d}
R_{MN}\Gamma^{MN}\epsilon=\OO(\a)\:,
\end{equation}
which we refer to as the instanton condition. Of course, in order to have a consistent theory with compatible solutions between the equations of motion and supersymmetry, we must check that the Hull connection does satisfy \eqref{eq:instanton10d} to the correct order. This is done in Appendix \ref{app:Hull}.

\part{Moduli in Minkowski Compactifications}
\label{part:4d}

\chapter{Moduli of the Strominger System}
\label{ch:SS}
We begin this part of the thesis by considering maximally symmetric four-dimensional compactifications of the heterotic string presented in the last chapter. We consider general supersymmetric compactifications at $\OO(\a)$, where the internal space need no longer be K\"ahler. The conditions for supersymmetric solutions of this theory where already worked out in the 80's by Strominger and Hull~\cite{Strominger:1986uh, Hull:1986kz}. However, working out what the four-dimensional supergravity looks like has been an elusive problem. Indeed, it is not even known what the massless spectrum of the theory is.\footnote{Progress in this direction has been made in recent years, see e.g. \cite{Lapan2006, Becker:2006xp, Melnikov:2011ez}.}

In this chapter, we attempt to make some modest progress in this direction. In particular, we will work out the massless spectrum of the theory by means of cohomologies of holomorphic bundles over the compact space $X$. We will see that the true spectrum, or moduli, are a subset of the usual cohomologies that appear in Calabi-Yau compactifications, as might be expected. We will only consider infinitesimal deformations in this thesis, leaving the study of obstructions to future work. Neither will we perform the full dimensional reduction, which is needed to find the full lower energy $N=1$ four-dimensional theory, equipped with K\"ahler potential and superpotential. We discuss some of these issues in the discussion section of the chapter, in relation to future work. The work in this chapter is based on \cite{delaOssa:2014cia}.


Before we begin, we note that the theory is very similar to the $\OO(\a^2)$-theory, which we present in the next chapter. We also remark that similar work to that presented in this chapter appeared concurrently with this work in \cite{Anderson:2014xha}, where the problem was addressed from a slightly different point of view.



\section{Introduction}
\label{sec:introSS}
We begin in section \ref{sec:hetcomp} with a review of  the geometry of the Strominger system and set up the notation.\footnote{For more details and extensive reviews, see also \cite{Ivanov:2000fg, 2001CQGra..18.1089I, Gillard:2003jh, Gauntlett:2003cy, Cardoso20035}} The heterotic compactification we are interested in consists of a pair $(X,V)$ where $X$ is a six-dimensional Riemannian spin manifold, equipped with what we call a \it heterotic\rm\:$SU(3)$-structure. In particular, $X$ is complex and conformally balanced.

We also have a vector bundle $V$ on $X$, with a holomorphic connection that satisfies the slope zero Yang-Mills condition,
\begin{equation*}
\omega\lrcorner F=0\:.
\end{equation*}
Together, these conditions (holomorphic and Yang-Mills of slope zero) imply that the connection is an {\it instanton}. Additionally, there is a holomorphic connection $\nabla$ on the tangent bundle with curvature $R $ which also satisfies this Yang-Mills condition. As discussed at length in the next chapter, this instanton connection is needed to ensure that supersymmetric solutions which satisfy the anomaly cancelation condition, also solve the equations of motion. This connection depends on the other fields of the theory. Exactly what this dependence is comes down to how one defines the various fields involved. Here, we will take the connection to be unspecified, and treat it more as a dynamical field\footnote{Recall that it should be the Hull connection \eqref{eq:Hullconn} for the usual  supersymmetry transformations \eqref{eq:O1spinorsusy1}-\eqref{eq:O1spinorsusy3}.}. As we will see, this is needed to be able to implement the Bianchi identity into the deformation theory.  The price we pay is that in doing so we get extra ``moduli" associated to this connection, which will have to be given the correct physical interpretation. We discuss this in detail in the next chapter.


\subsection{Holomorphic Structures and Moduli.}
We begin subsection \ref{sec:infsu3} by discussing the deformations of the relevant $SU(3)$-structure of $X$.  Such a deformation of the $SU(3)$-structure corresponds to simultaneous deformations of the complex structure, together with those of the hermitian structure, such that the heterotic $SU(3)$-structure is preserved. The deformations of the complex structure are easily described as the complex structure does not depend on the metric or the hermitian structure on $X$.  The analysis is similar to some extent to that for Calabi--Yau manifolds. Deformations of the hermitian structure satisfying the conformally balanced condition is more difficult as the hermitian structure also deforms with the complex structure. One of the problems is that the moduli space of deformations of the hermitian structure seems to be infinite-dimensional \cite{Becker:2006xp}.  Moreover, the conformally balanced condition is not stable under deformations of the complex structure~\cite{MR1107661, MR1137099, 0735.32009, 0793.53068}, in contrast with the stability of the K\"ahler condition. It turns out that by including the equations derived from the deformation of the anomaly cancelation condition, we find a finite-dimensional space for these parameters.  

We investigate the moduli of the Strominger system using the mathematical tools available in  deformation theory of holomorphic structures. We use the machinery developed by Atiyah \cite{MR0086359}.\footnote{This was also  used in \cite{Anderson:2010mh} where the combined bundle and complex structure moduli where studied in the Calabi-Yau case.}  We construct a holomorphic structure $\bD$ on an extension bundle $\Q$ which is an extension by the holomorphic cotangent bundle $T^*X$ of a bundle $E$ given by the short exact sequence
\begin{equation}
\label{eq:sesIntro}
0 \rightarrow T^*X \rightarrow \Q \rightarrow E \rightarrow 0\:,
\end{equation}
with an extension class  
\begin{equation*}
[\H]\in \textrm{Ext}^1(E,T^*X) = H^{(0,1)}(E^*\otimes T^*X)\:,
\end{equation*}
which precisely enforces the Bianchi identity. Here
\begin{equation*}
0\rightarrow\End(V)\oplus\End(TX)\rightarrow E\rightarrow TX\rightarrow0
\end{equation*}
is again defined by an extension sequence. We compute infinitesimal deformations of the holomorphic structure by computing long exact sequences in cohomology associated to the short exact sequences above.  We proceed in a stepwise manner.

In subsection \ref{subsec:defsV} we study in detail the deformations  of the holomorphic structure of the bundle $V$ on $X$. We generalise the work in~\cite{Anderson:2010mh, Anderson:2011ty} for the case in which $X$ is a Calabi-Yau manifold to the case of the more general non-K\"ahler manifolds at hand. We also extend these results in section \ref{subsec:defsTX} to include the deformations of the instanton connection $\nabla$ on $TX$. We then obtain simultaneous deformations of the holomorphic structures on the bundles and of the complex structure of $X$. 

The full moduli of the Strominger system is given in \ref{sec:defD} where we state the main result of the chapter. As mentioned above, we define $\Q$ by extending $E$ by the holomorphic cotangent bundle $T^*X$, given by \eqref{eq:sesIntro},
and define a holomorphic structure $\bD$  on $\Q$
\begin{equation*}
\bD:\Omega^{(0,q)}(\Q)\rightarrow\Omega^{(0,q+1)}(\Q)\:.
\end{equation*}
The operator $\bD$ has a rather lengthy definition, which we leave to section \ref{subsec:Anomaly}. What is important is that it includes the Bianchi identity. Moreover $\bD^2 = 0$ if and only if all the Bianchi identities involved are satisfied. We also note that the definition of $\bar D$ is such that $\Q$ is self-dual as a holomorphic bundle.

We shall see that most of the Strominger system can then be put in terms of $\bD$, except the Yang-Mills conditions. As we will see in section \ref{sec:HYM}, these conditions can however be taken care of by imposing that $\bD$ be an instanton. Naively, this might seem to impose extra conditions on the geometry. However, for a conformally balanced space $X$, to the given order in $\alpha'$, we will see that this is not the case. Rather, the condition imposes a constraint on which extension classes we choose. To be more precise, we find that we must choose the extension classes harmonic with respect to some Laplacian. This is shown in appendix \ref{app:HYMconfbal}. We discuss further implications of this in the discussion section \ref{sec:discussSS}.

The infinitesimal deformations of this holomorphic structure correspond to the elements in the cohomology group $H^{(0,1)}_{\bD}(\Q)$, that is, the tangent space $T\M$ to the moduli space $\cal M$ is given by 
\begin{equation*}
T\M\cong H^{(0,1)}_{\bD}(\Q)\cong  \Big[H^{(0,1)}_{\bp}(T^*X)\Big/\textrm{Im}(\H)\Big]\oplus\ker(\H)\:,
\end{equation*}
and we show that this is the infinitesimal moduli space of the heterotic compactifications. The subgroup $\ker(\H)$ is contained in the moduli space of deformations of $E$, that is the simultaneous variations of complex structure on X and holomorphic structures on the bundles, and it is in fact the kernel of a map $\H$ that corresponds to the analogue of the Atiyah class  for the  short exact extension sequence defining $\Q$. This is nothing but the obvious fact that the anomaly cancellation condition poses a non-trivial extra constraint on the moduli.  We also argue that
the (complexified) hermitian parameters belong to the group
\begin{equation*}
H^{(0,1)}_{\bp}(T^*X)\:,
\end{equation*}
like in the Calabi--Yau case.  These should be modded out by 
\[\textrm{Im}(\H)\cong\{\tr(F\,\alpha)\:\vert\:\alpha\in H^0(\End(V))\}\subset H^{(1,1)}(X),\] 
which enforces the Yang-Mills condition on $V$. 

As discussed above, we appear to find in our setup extra moduli, which correspond to deformations of the holomorphic structure of the tangent bundle given by $\nabla$ but which leave $X$ fixed.  It seems that we need this extra structure to be able to enforce the anomaly cancelation condition. Note that this structure, in which the connection $\nabla$ behaves as another dynamical field, is also natural when one considers the heterotic theory to higher orders in the $\alpha'$ expansion, as we discuss in the next chapter. There are however reasons as to why we do not want these extra fields in the low energy field theory, namely, that the connection $\nabla$ is not independent of the geometry of $X$.  We will discuss these points in more detail in Chapter \ref{ch:connredef}. We also note that the mathematical structure we have in this chapter is very similar to that of \cite{2013arXiv1304.4294G, 2013arXiv1308.5159B} where a generalised geometry for the heterotic supergravity is discussed. See also \cite{Coimbra:2014qaa}, where a generalized geometry with the Hull connection $\nabla^-$ was constructed, by a reduction of the generalized structure group. Additionally, these structures have also made an appearence in the context of double field theory \cite{Hohm:2011ex, Hohm:2013jaa, Bedoya:2014pma, Hohm:2014eba, Hohm:2014xsa}.


In the discussion section, we also discuss the corresponding four-dimensional theory, and future directions relating to this. We consider the conjectured Gukov-Vafa-Witten superpotential, and how the above mentioned kernel structure should correspond to F-terms of such a potential. Also, higher order obstructions should give rise to a superpotential in the low energy theory. There is also a question with regards to what the K\"ahler potential of the four-dimensional theory is. Knowledge of this would require a full dimensional reduction of the theory \cite{McOrist2014, Candelas2014}, which we will not perform here. However, putting the system in terms of a holomorphic structure may serve to help in this direction, as moduli spaces of holomorphic structures have been studied extensively in the mathematics literature before, see e.g. \cite{kobayashi1987differential, kim1987moduli, schumacher1993moduli, huybrechts2010geometry}.


We also include an additional appendix \ref{app:Primitive}, where we show that deformations of the Yang-Mills conditions impose no further constraint, provided the bundle is stable with traceless endomorphisms. This is a bit technical, and is therefore left to an appendix. Its generalisation to the poly-stable case is however straight forward, and is discussed in section \ref{sec:polystabH} in the context of deformations of $\bD$ and the Strominger system.

\section{$SU(3)$-Structures and the Strominger System}
\label{sec:hetcomp}

We now review the results of Strominger, Hull and de\,Wit\,{\it et\,al} \cite{Strominger:1986uh, Hull:1986kz, deWit:1986xg}. The requirements that the four-dimensional space--time is maximally symmetric, that $N=1$ supersymmetry is preserved in four-dimensions, and that  an equation canceling anomalies is satisfied, pose strong constraints on the possible geometries that are allowed as solutions of the equations of motion.   

A compactification to four-dimensions is obtained by considering a ten-dimensional space-time which is a local product 
\begin{equation}
\label{eq:4dcomp}
M_{10}=M_4\times X\:,
\end{equation}
of a maximally symmetric four-dimensional flat space-time $M_4$, and a six-dimensional manifold $X$.  We use latin indices $m,n$, etc, to denote six-dimensional indices on the tangent space $TX$ of $X$.  Also, a consequence of imposing the constraints of  $N=1$ supersymmetry in four-dimensional space-time is that $M_4$ must be Minkowski. So $X$ is a real six-dimensional manifold with metric $g$, and on $X$ there is a vector bundle $V$ with curvature $F$ which takes values in $\End(V)$.  We now discuss the constraints on the geometry of $(X,V)$. 

\subsection{Constraints on the Geometry of $X$.}
\label{subsec:Xgeom}

Under the compactification \eqref{eq:4dcomp}, the ten-dimensional spinor $\epsilon$ decomposes as
\begin{equation}
\label{eq:spinordecomp}
\epsilon=\kappa\otimes\eta\:,
\end{equation}
where $\eta$ is a nowhere vanishing globally defined complex spinor on $X$, and $\kappa$ is the four-dimensional remainder. This means that $X$ must be a spin manifold and that the structure group of  $X$ is reduced to a subgroup $SU(3)\subset \text{Spin}(6)$.  An $SU(3)$-structure on $X$ \cite{0444.53032, 1024.53018, LopesCardoso:2002hd}  is defined by a triple $(X,\omega,\Psi)$, where $\omega$ is a non-degenerate globally well defined real 2-form, and $\Psi$ is a 
no-where vanishing globally well defined complex 3-form. The forms $\Psi$ and $\omega$ satisfy
\begin{equation}
 \omega\wedge\Psi = 0~.\label{eq:compatibility}
 \end{equation}
In fact, there is an $SU(3)$-structure on $X$ determined entirely by the spinor $\eta$.  The two non-degenerate forms, $\omega$ and $\Psi$, can be constructed as bilinears of  $\eta$
\begin{equation*}
\begin{split}
\omega_{mn} &= -i \, \eta^\dagger \,\gamma_{mn}\, \eta\\
\Psi_{mnp} &= ~~\eta^T\,\gamma_{mnp}\,\eta~,
\end{split}
\end{equation*}
where $\gamma_m$ are the Dirac gamma-matrices, that satisfy the Clifford algebra in six-dimensions
\[ \{\gamma_m,\gamma_n\} = 2\, g_{mn}~,\]
and $\gamma_{m_1m_2\cdots m_p}$ denotes the totally antisymmetric product of $p$ gamma matrices
\[\gamma_{m_1m_2\cdots m_p} = \gamma_{[m_1}\gamma_{m_2}\cdots\gamma_{m_p]}~.\] 
Using Fierz rearrangement, one can prove that these satisfy  \eqref{eq:compatibility}.  One can also prove that there is a unique (up to a constant) invariant volume form on $X$ which satisfies
the compatibility condition
\begin{equation}
\d{\rm vol}_X =  \frac{1}{6}\, \omega\wedge\omega\wedge\omega =
\frac{i}{||\Psi||^2}\, \Psi\wedge\bar\Psi~,\label{eq:volume}
\end{equation} 
where
\[ ||\Psi||^2 = \frac{1}{3!}\, \Psi^{mnp}\, \Psi_{mnp}~.\]

The (real part of the) complex 3-form $\Psi$ determines a unique almost complex structure $J$\cite{MR1871001} such that $\Psi$ is a $(3,0)$-form
with respect to $J$.  In fact, 
\begin{equation}
\label{eq:complexstr}
{J_m}^n=\frac{{I_m}^n}{\sqrt{-\frac{1}{6}\tr I^2}},
\end{equation}
where the tangent bundle endomorphism $I$ is given by
\begin{equation}
{I_m}^n=(\textrm{Re}\Psi)_{mpq}(\textrm{Re}\Psi)_{rst}\,\epsilon^{npqrst}.\label{eq:CSnorm}
\end{equation}
With the normalization in \eqref{eq:complexstr},  it is not too difficult to prove that $J^2=-{\bf 1}$.  Note also that a change of scale $\Psi\to\lambda\Psi~,\lambda\in\IC^*$, defines the same complex structure $J$. 
With respect to $J$, the real two form $\omega$ is type $(1,1)$ due to \eqref{eq:compatibility}.  Moreover, $\omega$ is an almost hermitian form
\begin{equation}
\omega(X,Y) =  g(JX, Y)~,\quad \forall X, Y\in TX~.
\end{equation}

Therefore 6-dimensional manifolds with an $SU(3)$-structure are almost hermitian manifolds with trivial canonical bundle. We are not done however, as the preservation of  supersymmetry  also imposes differential conditions on $\omega$ and $\Psi$. 

Preservation of supersymmetry requires that on $X$ there must exist a metric connection $\nabla^+$ with skew-symmetric torsion $T=H$, where $H$ is the 3-form flux. Hence, the connection reads
\begin{equation*}
 \nabla^+_{mn}{}^p = {\nabla^{LC}}_{mn}{}^{\, p} + \frac{1}{2}\, H_{mn}{}^p~,
 \end{equation*}
The vanishing of the supersymmetric variation of the gravitino requires that the spinor $\eta$ must be covariantly constant with respect to this connection
\[ \nabla^+_m\,\eta = \nabla^{LC}_m\, \eta + \frac{1}{8}\, H_{mnp}\, \gamma^{np}\, \eta = 0~.\]
This in turn means that the forms $\omega$ and $\Psi$ are covariantly constant
\[ \nabla^+\Psi = 0~,\qquad \nabla^+\omega= 0~.\]
The almost complex structure determined by $\Psi$ must also be covariantly constant
\[ \nabla^+ J = 0~,\]
that is $\nabla^+$ is a hermitian connection. 

On an almost complex manifold there is unique metric connection for which $J$ is parallel, which has totally antisymmetric torsion. This is precisely the connection $\nabla^+$.  This connection is called the {\it Bismut connection}\cite{0666.58042} in the mathematics literature, and its torsion is given by
\begin{equation}
T= H =\d_c\omega= \J^{-1}\d\J\:,
\end{equation}
where
\begin{equation*}
\J=\sum_{p,q}i^{p-q}\Pi^{p,q}\:,
\end{equation*}
where $\Pi^{p,q}$ are the projectors onto $(p,q)$-forms. It should be noted that $\J$ is uniquely determined by $J$ and vice versa.

Equations for the exterior derivative of $\omega$ and $\Psi$ can be obtained from the fact that both $\omega$ and $\Psi$ are covariantly constant.  One can decompose the exterior derivative of $\omega$ and $\Psi$ into irreducible representations of $SU(3)$.
For a general $SU(3)$-structure, we have 
\begin{align}
\label{eq:tordom}
\d\omega=-\frac{12}{\vert\vert\Psi\vert\vert^2}\,\textrm{Im}(W_0\bar\Psi)+W_1^\omega\wedge\omega+W_3\\
\label{eq:tordPsi}
\d\Psi=W_0\ \omega\wedge\omega+W_2\wedge\omega+\bar W_1^\Psi\wedge\Psi\:.
\end{align}
were $(W_0, W _1^\omega,W_1^\Psi,W_2,W_3)$ are the five {\it torsion classes}\cite{0444.53032, 1024.53018, LopesCardoso:2002hd, Gauntlett:2003cy}.
Here, $W_0$ is a complex function, $W_2$ is a primitive $(1,1)$-form, $W_3$ is  a real primitive 3-form of type $(1,2)+(2,1)$, $W_1^\omega$ is a real one-form, and $W_1^\Psi$ is a $(1,0)$-form. 
The one forms $W_1^\omega$ and $W_1^\Psi$ are known as the Lee-forms of $\omega$ and $\Psi$ respectively, and they are given by
\begin{align*}
W_1^\omega&=\frac{1}{2}\, \omega\lrcorner\d\omega\\
W_1^\Psi&=- \frac{1}{\vert\vert\Psi\vert\vert^2}\,\Psi\lrcorner\d\bar\Psi.
\end{align*}
The contraction operator $\lrcorner$ is defined as 
\[ \alpha\lrcorner \beta = \frac{1}{k! p!}\, \alpha^{m_1\cdots m_k}\, \beta_{m_1\cdots m_k n_1\cdots n_p}\,
\d x^{n_1}\wedge\cdots\wedge \d x^{n_p}
= (-1)^{p(d-p-k)}\,* (\alpha\wedge *\beta)~,\]
where $\alpha$ is a $k$-form, $\beta$ is a $k+p$-form, and $d$ is the dimension of the manifold, which in our case is $d=6$.

Using both the gravitino and dilatino Killing spinor equations, one can show that $W_0$ and $W_2$ vanish, which is equivalent to the vanishing of the Nijenhuis tensor, and thus the integrability of the complex structure $J$ (note that these are the only torsion classes that scale under $\Psi\to\lambda\Psi$). Hence $X$ {\it must be a complex manifold}. It follows that the complexified tangent bundle can be decomposed into holomorphic and anti-holomorphic coordinates $\{z^a,\bar z^{\bar b}\}$ respectively, and similar for the cotangent bundle. We denote the holomorphic tangent bundles as $TX$ and $T^*X$ for ease of notation.

Also, for the Bismut connection the Lee-forms are related by\cite{LopesCardoso:2002hd}
\[   {\rm Re}(W_1^\Psi) =  W_1^\omega~.\]
Therefore, the exterior derivatives of $\omega$ and $\Psi$ are
\begin{align*}
\d\omega&= {\rm Re}(W_1^\Psi)\wedge\omega+W_3\:\\
\d\Psi&=\bar W_1^\Psi\wedge\Psi\:.
\end{align*}
Note that $\bar\partial\, \bar W_1^\Psi = 0$ as can be seen by taking the exterior derivative of the second equation. 

The vanishing of the supersymmetric variation of the dilatino gives a further constraint: the Lee-form of $\omega$ must be exact with
\[ W_1^\omega = \d\phi~.\]
Therefore, the equation for $\d\omega$ gives
\begin{equation}
\d( e^{-2\phi}\, \omega\wedge\omega) = 0~,\label{eq:confbal}
\end{equation}
that is, the manifold $X$ is required to be {\it conformally balanced}. Furthermore, the equation for $\d\Psi$ means that $X$ must have a {\it holomorphically trivial canonical bundle}
\begin{equation}
 \d\Omega=\d(  e^{-2\phi}\, \Psi) = 0~,\label{eq:hol}
 \end{equation}
where $\Omega= e^{-2\phi}\, \Psi$. In this chapter, a complex conformally balanced manifold $X$ with a holomorphically trivial canonical bundle will be called a manifold with a {\it heterotic structure}.

Note also that on a complex manifold, the exterior derivative decomposes in the Dolbeault operators $\d=\p+\bp$, where $\p^2=\bp^2=0$. The flux/torsion $H$ can then be written as
\begin{equation}
H=\d_c\omega=i(\partial - \bar\partial)\,\omega~.\label{eq:Bismut}
\end{equation}
Recall next the Weil relation \cite{weil1958introduction}, which states that for a primitive $k$-form $B_k$
\begin{equation}
\label{eq:Weil}
*\frac{1}{r!}\omega^r\wedge B_k=\frac{i^{k(k+1)}}{(n-k-r)!}\omega^{n-k-r}\J(B_k)\:,
\end{equation}
where $n=3$ is the complex dimension of $X$, and $*$ is the Hodge-dual with respect to $g$. Using \eqref{eq:Weil}, we see that we may also write the flux as
\begin{equation}
\label{eq:flux2}
H=*e^{2\phi}\d(e^{-2\phi}\omega)\:.
\end{equation}
 
\subsection{Constraints on the Vector Bundle $V$.}
\label{subsec:Vgeom}
The vanishing of the supersymmetric variation of the gravitino imposes conditions on the bundle $V$.  More precisely, the curvature of the Yang-Mills connection satisfies
\begin{align}
F\wedge\Psi  &= 0~,\label{eq:holF} \\
\omega\lrcorner F &= 0\label{eq:instV}~. 
\end{align} 
The first condition is equivalent to $F^{(0,2)}= 0$, that is, $V$ {\it must be a holomorphic bundle}. In particular, the bundle admits a decomposition into holomorphic and anti-holomorphic coordinates. We will be sloppy, and refer to both the full bundle and the holomorphic part of the bundle as $\End(V)$. It is hopefully clear which is meant from the context.

The second equation states that the curvature $F$ must be {\it primitive} with repect to $\omega$.  Both conditions together mean that  $V$ {\it must admit a hermitian Yang-Mills connection}. Because the right hand side of equation \eqref{eq:instV} is zero, we say that the connection on  $V$ is an {\it instanton}. When $X$ is K\"ahler, the existence of such a connection on $V$ is guaranteed by the work of Donaldson\cite{0529.53018} and Uhlenbeck and Yau\cite{MR861491, MR997570}. Buchdahl \cite{MR939923} (for the case of complex surfaces) and,  Li and Yau\cite{MR915839} (for higher dimensional complex manifolds) generalised the Donaldson, Ulenbeck-Yau Theorem to non-K\"ahler manifolds with a {\it Gauduchon metric}.  A Gauduchon metric $\hat g$ on a hermitian $n$ dimensional manifold with corresponding hermitian form $\widehat\omega$ is a metric that satisfies
\begin{equation*}
\partial\bar\partial\,\widehat\omega{}^{n-1} = 0\:.
\end{equation*}
For $n=3$, and a manifold $X$ which has an heterotic $SU(3)$-structure, this means that $\widehat\omega = e^{-\phi}\omega$ is Gauduchon.

\begin{Theorem}[Buchdahl, Li-Yau]
\label{tm:LiYau}
Let $X$ be a compact hermitian manifold with a Gauduchon metric $\hat g$ and corresponding hermitian form $\widehat\omega$.  A poly-stable (with respect to the class $[\widehat\omega^2]$) holomorphic vector bundle $V$ over $X$ admits a unique hermitian Yang-Mills connection.
\end{Theorem}
\rightline{$\square$}
\noindent The stability refers to the slope $\mu(V)$ of $V$ which is now defined as
\[\mu(V) = \frac{1}{{\text rk}(V)}\, \int_X c_1(V)\wedge\widehat\omega^2~,\]
and it states that for all sub-sheaves $E$ of $V$ it must be true that
\[ \mu(E)<\mu(V)~.\]

It should be noted that in proving Theorem \ref{tm:LiYau}, Li and Yau relied upon a choice of {\it local holomorphic trivialisation},
\begin{equation}
\label{eq:holtriv}
\d_A=\bp+\p+a\:,
\end{equation}
which is possible for a holomorphic vector bundle \cite{1055.14001}. That is, for a holomorphic vector bundle we can set $\bp_A=\pb$ upon a suitable local gauge transformation. This connection is then assumed to be the {\it hermitian connection}, or Chern connection of some hermitian structure $h_V$ on $V$, that is $a=h_V^{-1}\p h_V$.  The corresponding primitive curvature is $F=\bp a$. Li and Yau then proved the existence of such a hermitian Yang-Mills structure, if and only if the holomorphic bundle $V$ is poly-stable. If such a structure exists, it is unique with respect to the holomorphic structure on $V$.


Therefore, when $X$ has a heterotic $SU(3)$-structure, we require the bundle $V$ to be a poly-stable holomorphic bundle  (with respect to the class $[e^{-2\phi}\,\omega^2]$), which thus guarantees the existence of a hermitian Yang-Mills connection on $V$. For heterotic string compactifications, the relevant vector bundles have $c_1(V)=0$ and so {\it the slope vanishes} $\mu(V) =0$. In fact, for the gauge bundle, the gauge group is a subgroup of $E_8{\times} E_8$.  Also, as we will see in the next section, we require that $TX$ be stable too, where $\mu(TX)=0$ because $X$ has vanishing first Chern class. 

\subsection{Constraints from the Anomaly Cancelation and Equations of Motion.}
\label{subsec:anomaly}

Apart from the constraints from supersymmetry, the pair $(X,V)$ must also satisfy an anomaly cancelation condition 
\begin{equation}
H =i(\p-\bp)\omega=\d B + {\cal C S}~.\label{eq:anomaly}
\end{equation}
Recall that the right hand side of the anomaly cancelation condition \eqref{eq:anomaly} is a definition of $H$ as the gauge invariant field strength of the $B$ field. The Bianchi identity for this anomaly cancelation condition is 
\begin{equation}
\d H =-2i\p\bp\omega= \frac{\alpha'}{4}\,\left(\tr (F\wedge F) - \tr (R \wedge R ) \right) + {\cal W}~,\label{eq:BIanomaly}
\end{equation}
where $R $ is the curvature on $X$ with respect to a connection $\nabla$ on $TX$ which we discuss below. Here $F$ and $R$ are the curvatures on the full bundles, not just the holomorphic part. Again we are sloppy and use $F$ and $R$ interchangeably referring to the curvatures on $\End(V)$ and $\End(TX)$ as the full bundles or their holomorphic part..

The term $\cal W$ is a non-perturbative correction which is a closed 4-form on $X$ in a cohomology class which corresponds to the Poincar\'e dual of the class of an (effective) holomorphic curve $\cal{C}$ which is wrapped by a five-brane. 
A topological condition derives from equation \eqref{eq:BIanomaly}  
\begin{equation}
 0 = - \, {\cal P}_1(V) + {\cal P}_1(TX) + [\cal{C}]~,\label{eq:topanomaly}
 \end{equation}
 where ${\cal P}_1(E)$ represents the first Pontryagin class of a bundle $E$. We will ignore the non-pertubative correction $\cal W$.
 
Any solution $(X,V)$ of the supersymmetry conditions, which also satisfies the anomaly cancelation condition, automatically satisfies the equations of motion if and only if the connection $\nabla$ satisfies
\begin{align} 
R \wedge\Psi &= 0~,\label{eq:holR}\\
 \omega\lrcorner R  &= 0~.\label{eq:instR}
 \end{align}
That is, the connection $\nabla$ of the curvature $R $ of the tangent bundle $TX$ must is an $SU(3)$ instanton\cite{Hull:1986kz, Ivanov:2009rh, Martelli:2010jx}. We give an explicit proof of this in appendix \ref{app:proof}. By the Theorem of Li and Yau above, such a connection exists only if we require $(TX, \nabla)$ to be a stable holomorphic bundle on $X$. The corresponding connection is given by the holomorphic trivialisation
\begin{equation}
\label{eq:holtrivnabla}
\nabla=\bp+\p+\theta\:,
\end{equation}
where $\theta=h^{-1}\p h\in\Omega^{(1,0)}(\End(TX))$, and $h$ is the unique hermitian Yang-Mills structure on $TX$. Note that as $h$ is a hermitian structure on $TX$, it actually defines a hermitian metric on $X$.

\subsection{The usual field choice, $\nabla=\nabla^-$.}
\label{eq:usual}
As noted above, the connection $\nabla$ appearing in the Bianchi Identity should be hermitian Yang-Mills, but we did not further specify it. However as we also noted, it should depend on the fields of the theory in some particular way. As we will see in the next chapter, this comes down to how one defines the fields. For the usual field choice, it is given by the Hull connection $\nabla^-$. However, given that this connection only appears at first order in $\alpha'$, to first order in $\alpha'$ we may consistently choose $\nabla$ to be $\nabla^-+\OO(\a)$ with this field choice.

A simple application of Stokes theorem, first given in \cite{Gauntlett:2002sc}, now shows that the flux is indeed of first order for compact $X$,
\begin{equation*}
\vert\vert e^{-\phi}H\vert\vert^2=\int_Xe^{-2\phi}H\wedge*H=-\int_XH\wedge\d(e^{-2\phi}\omega)=0+\OO(\a)\:,
\end{equation*}
by an integration by parts. Here we have used equation \eqref{eq:flux2} and the Binachi Identity. It follows that the flux is indeed of $\OO(\a)$. Hence we may as well use the Levi-Civita connection, or the Chern connection, within this field choice. The hermitian Yang-Mills condition then simply becomes the condition that the zeroth order geometry is Ricci flat, and hence Calabi-Yau.

Note also that as the zeroth order geometry is CY, the tangent bundle $TX_0$ of the zeroth order geometry is stable, where we denote the zeroth order geometry by $X_0$ for short. In particular, we have 
\begin{equation*}
H^0(TX_0)=0\:.
\end{equation*}
However, as the complex structure is independent of rescaling, and the $\a$-expansion is really a large volume expansion, it follows that we may take the expansion of the complex structure to be trivial. That is $J=J_0$, the zeroth order complex structure. This further implies that the Dolbeault cohomology groups of the base remain the same in the expansion as well. In particular, we can assume that
\begin{equation*}
H^0(TX)=0\:.
\end{equation*}
We shall return to this point later in the chapter.

\section{Infinitesimal Deformations of Heterotic Structures}
\label{sec:moduli}

In the next three sections we study the space of {\it infinitesimal} deformations of a heterotic compactification $(X,V)$.  As described in detail in the following, this moduli space contains the following parameters:

\begin{itemize}
\item Deformations of the {\it complex structure $J$ on $X$} (which is determined by $\Psi$).  It is well known that infinitesimal deformations $\Delta$ of the complex structure which preserve the integrability of the complex structure  belong to $H_{\bar\partial}^{(0,1)}(TX)$.  These deformations $\Delta$ are constrained by requiring that $\Omega$ stays holomorphic, equation \eqref{eq:hol}, which in turn requires that deformations of the complex structure are in $H_\d^{(2,1)}(X)\subseteq H_{\bar\partial}^{(0,1)}(TX)$.  Moreover, they are also constrained by the requirement that the holomorphic conditions of the bundles $V$ ($F\wedge\Psi = 0$) and $TX$ ($R \wedge\Psi = 0$) be preserved. We also find a further constraint on $\Delta$ coming from the anomaly cancelation condition.  

\item Deformations of the bundle connection $A$, for a fixed complex structure $J$ and hermitian form $\omega$ on $X$.
These are the {\it bundle moduli of $V$} which belong to $H^{(0,1)}(\End(V))$.  Similarly, we have deformations of the holomorphic tangent bundle connection $\Theta$, the {\it tangent bundle moduli}, which belong to $H^{(0,1)}(\End(TX))$.  Note that we are considering the instanton connection as an unphysical field in the theory.  We find that this is needed for the appropriate implementation of the anomaly cancelation condition, but we do not consider these moduli as corresponding to physical fields in the effective four-dimensional field theory. We interpret these ``moduli" from a physics point of view in the next chapter. 
\item Deformations of the hermitian structure $\omega$ which preserve the conformally balanced condition which are constrained by the anomaly cancelation condition. As we shall see, these are also obstructed by the Yang-Mills condition in the case of poly-stable bundles.
\end{itemize}

We leave the study of obstructions to these deformations for future work. 

\section{Deformations of Heterotic $SU(3)$-Structures}
 \label{sec:infsu3}
 
Let $(X, \omega, \Psi)$ be a manifold with a heterotic $SU(3)$-structure.  In this subsection we discuss first order variations of this $SU(3)$-structure.
   
Consider a one parameter family of manifolds $(X_t, \omega_t, \Psi_t), ~t\in\IC$, with a heterotic $SU(3)$-structure where we set
$(X, \omega,\Psi) = (X_0, \omega_0, \Psi_0)$. 
A deformation of the heterotic $SU(3)$-structure parametrized by the parameter $t$ corresponds to simultaneous deformations of the complex structure determined by $\Psi$ together with those of the hermitian structure
determined by $\omega$, such that the heterotic $SU(3)$-structure is preserved.  Hence the variation with respect to $t$ of any mathematical quantity $\alpha$ (as for example a $p$-form, or the metric) is given by the chain rule as follows
\begin{equation*}
\partial_t\alpha = (\partial_t Z^A)\, \partial_A\alpha + 
(\partial_t Z^{\bar A})\, \partial_{\bar A}\alpha + (\partial_t y^i)\, \partial_i\alpha ~,
\end{equation*}
where we label the complex structure parameters by $Z^A$ and by $y^i$ the parameters of the hermitian structure.\footnote{We will need to extend this later to include variations of the bundles.}
 
Note that $\Psi$ is independent of the hermitian structure parameters, however $\omega$ does depend on the complex structure as it must be  a $(1,1)$-form with respect to any complex structure. Therefore the moduli space ${\cal M}_X$ of the manifold $X$ must have the structure of a fibration.  We discuss this structure in the following.

\subsection{Deformations of the Complex Structure of $X$}
\label{sec:varsJ}

We begin this subsection by reviewing standard results on variations of an {\it integrable} complex structure $J$ of a manifold $X$.   With respect to $J$, the exterior derivative $\bp$ which acts on forms on $X$, squares to zero, that is $\bp^2=0$.  This condition is equivalent to the vanishing of the Nijenhuis tensor.  Conversely, a derivative $\bp$ which squares to zero defines an integrable complex structure on $X$. In fact, it determines a holomorphic structure on~$X$.

Let $Z^A, A = 1,\ldots, N_{CS}$, be complex structure parameters and $\Delta_A^m$ be a variation of the complex structure
\[ \Delta_A  = \Delta_A{}_{\, n}{}^m\, \d x^n\otimes \partial_m = - \frac{i}{2}\, \partial_AJ~.\]
It is a standard result that $\Delta_A^m\in \Omega^{(0,1)}(TX)$. Moreover, the deformations $\Delta_A$ are valued in the {\it holomorphic tangent bundle}, so $\Delta_A^m=\Delta_A^a$. Further, preservation of the integrability of the complex structure under variations requires to first order that $\Delta_A^m$ defines an element of $H^{(0,1)}_{\bar\partial}(TX)$.  The integrability to first order is guaranteed (using the Maurer-Cartan equations) if $\bp\Delta_A = 0$, and $\bp$-exact forms $\Delta_A$ correspond to trivial changes of the complex structure (that is, changes corresponding to diffeomorphisms). 

Equivalently, as the form $\Psi$ on $X$ determines a unique integrable complex structure $J$ on a manifold with a heterotic structure $X$, one can study the  variations of $J$ in terms of the variations of $\Psi$. It will be more convenient however to discuss these deformations in terms of the holomorphic $(3,0)$ form $\Omega$. First order variations of $\Omega$ have the form\cite{MR0112154, 0128.16902} 
\begin{equation} 
\partial_A \Omega = \widetilde K_A\, \Omega + \chi_A~,\label{eq:prevarO}
\end{equation}
where the $\widetilde K_A\in\Omega^0(X)$, and $\chi_A$ is a $(2,1)$-form which can be written in terms of $\Delta_A$
\begin{equation}
\chi_A = \half\, \Omega_{mnp}\, \Delta_A{}^m\wedge\d x^n\wedge\d x^p~.\label{eq:chi}
\end{equation}
Actually, we can prove that one can take the $\widetilde K_A$ to be  constants. Indeed, the Hodge-decomposition of $\widetilde K_A\, \Omega$ in terms of $\p$ is
\begin{equation*}
\widetilde K_A\, \Omega=K_A\,\Omega+\p\beta\:,
\end{equation*}
where $\beta$ is a $(2,0)$-form, and $K_A$ is constant. The last term can be ignored because it corresponds to changes in $\Omega$ due to diffeomorphisms of $X$, that is, trivial deformations of the complex structure.  This can be seen by computing the  Lie-derivative of $\Omega$ along a vector $v\in TX$ which gives
\begin{equation}
{\cal L}_v\Omega = -\delta_{diff}\,\Omega = \d(v\lrcorner\Omega)~,\label{eq:Odiffeo}
\end{equation}
where we have used the fact that $\d\Omega= 0$. Taking the $(3,0)$-part of this equation, we obtain
\[ ({\cal L}_v\Omega)^{(3,0)} = \p(v\lrcorner\Omega)~.\]
 
We now vary the equation
\[\d \Omega = 0~,\]
and demand that the deformed manifold admits a holomorphic $(3,0)$-form.  We find
\begin{equation}
 \d\chi_A = 0~.
 \label{eq:dchi}
 \end{equation}
 Therefore  each $\chi_A$ defines a class in the de-Rham cohomology
 \[ \chi_A \in H_{\d}{}^{(2,1)}(X)~,\]
as $\d$-exact $(2,1)$-forms correspond to 
diffeomorphisms of $X$, as can be seen from equation \eqref{eq:Odiffeo}.

We remark that using the holomorphicity of $\Omega$, and equation \eqref{eq:chi}, it is straightforward to prove that $\chi_A\in H_{\bp}{}^{(2,1)}(X)$ is equivalent to $\Delta_A{}^m \in H^{(0,1)}_{\bar\partial}(TX)$.  In fact, $\Omega$ gives an isomorphism between these cohomology groups (just like in the case of Calabi--Yau manifolds)\cite{MR915841}
\[H_{\bp}{}^{(2,1)}(X)\cong H^{(0,1)}_{\bar\partial}(TX)~.\] 
However, on a non-K\"ahler manifold with a holomorphically trivial canonical bundle, it is not necessarily the case that
$H_{\d}{}^{(2,1)}(X) \cong H_{\bp}{}^{(2,1)}(X)$
and it is generally the case that 
\[ {\rm dim} H_{\d}{}^{(2,1)}(X) \le {\rm dim} H_{\bp}{}^{(2,1)}(X)~.\]
The way to see this is by observing that first order variations of  \eqref{eq:dchi} require, not only that $\chi_A$ is $\bp$-closed, but also that it is $\p$-closed. Therefore, in a given class of $[\chi_A]\in H_{\bp}{}^{(2,1)}(X)$, there must exist a representative which is  $\p$-closed and is not $\d$-exact.  This is not always the case, and there are many examples of non-K\"ahler manifolds for which this happens.  A simple example with a heterotic $SU(3)$-structure is the Iwasawa manifold\footnote{However, the Iwasawa manifold is not a good heterotic compactification for any bundle $V$ because its holomorphic tangent bundle is not stable.}.  It is not too hard to show that for this example\cite{MR0393580}
\[ {\rm dim} H_{\d}{}^{(2,1)}(X) = 4~,\quad{\rm and} \quad {\rm dim} H_{\bp}{}^{(2,1)}(X) = 6~.\]
The two extra elements in $H_{\bp}{}^{(2,1)}(X)$ are $\p$-closed, however they are $\d$-exact. 

There are also many examples of non-K\"ahler manifolds for which 
\begin{equation} H_{\d}{}^{(2,1)}(X)\cong H_{\bp}{}^{(2,1)}(X)~,
\label{eq:CScong}
\end{equation}
like for example manifolds which are {\it cohomologically K\"ahler}, that is, which satisfy the $\p\bp$-lemma, a property which is stable under complex structure deformations\cite{MR2451566,MR2449178}.  The Iwasawa manifold does not satisfy the $\p\bp$-lemma.

The condition that each $\chi\in H_{\bp}{}^{(2,1)}(X)$ also satisfies $\p\chi=0$ is used  in\cite{MR915841} as a first step to discuss the obstructions to infinitesimal deformations of the complex structure $J$ of Calabi--Yau manifolds, and it is stated in that proof that it goes through even if  the manifold is not  K\"ahler, as long as it satisfies the $\p\bp$-lemma.  In our work, the requirement that $\p\chi=0$ appears at first order in deformation theory when discussing the deformations of the complex structure in terms of the variations of $\Omega$ and {\it it represents a necessary condition for integrability to first order}.  Issues including  the integrability of the deformed complex structure of $X$ and a generalisation of the work of Tian and Todorov~\cite{MR915841,MR1027500}, is beyond the scope of this thesis, but will be discussed in a forthcoming publication \cite{DKS2014}.  For the rest of this chapter, we work with the variations of the complex structure in terms of $\Delta$, but we should keep in mind that some of these elements may be obstructed as discussed.

\subsection{Deformations of the Hermitian Structure on $X$}
\label{subsubsec:defHS} 
We now consider deformations of the hermitian form $\omega$. Recall that for Calabi-Yau compactifications, the conditions of the deformations of the K\"ahler form is for it to remain closed. This forces the infinitesimal deformations to be valued in $H^{(1,1)}(X)$ modulo diffeomorphisms, and there is therefore a finite set of allowed deformations. In our case, the closeness condition is replaced by the anomaly cancellation condition \eqref{eq:anomaly}. We will discuss later how to include this condition in the deformation story. Recall however that we also have the conformlly balanced condition \eqref{eq:confbal}. As it turns out, this condition does put restrictions on the allowed hermitian deformations. Indeed, it determines the $\bp$-exact part of the deformation of $\omega$. Let us see how this goes.

Let
\[\hat\rho = \half\, \widehat\omega\wedge\widehat\omega~,\] 
where $\widehat\omega = e^{-\phi}\omega$ is the hermitian form corresponding to the Gauduchon metric. The conformally balanced condition is then equivalent to
\[\d \hat\rho = 0~,\]
and so $\hat\rho\in H_\d^4(X)$. Any variation of $\hat\rho$ must preserve this condition, that is
\[ \d (\partial_t\hat\rho) = 0.\]
Consider the action of  a diffeomorphism of $X$ on $\hat\rho$
\begin{equation}
\label{eq:gaugehatrho}
\L_{diff} \hat\rho = \d (v\lrcorner\hat\rho) = - \d( e^{-\phi} \, J(v)\wedge\widehat\omega)=-\d(\tilde v\wedge\hat\omega)\:,
\end{equation}
where we have set $\tilde v = e^{-\phi} \, J(v)$. Therefore, variations of $\hat\rho$ which preserve the conformally balanced condition correspond to $\d$-closed four forms modulo $\d$-exact forms which have the form $\d\beta$, where $\beta$ is a {\it non-primitive} three-form.  So this space is not necessarily finite dimensional as was first pointed out in\cite{Becker:2006xp}. As we will see,  when taking into account the anomaly cancelation condition, we obtain a finite dimensional parameter space. For the remainder of this section, we set up some notation and make some further remarks on the deformations of hermitian structure of $X$.

Consider a variation, $\p_t\omega$, of $\omega$.  We can decompose this variation in terms of the Lefshetz decomposition
\begin{equation}
 \p_t\omega = \lambda_t\omega + h_t~,\label{eq:varomega}
 \end{equation}
where $\lambda_t$ is a function on $X$ and $h_t$ is a primitive two form ($\omega\lrcorner h_t = 0$).

It is not too difficult to show that the $(0,2)$-part of the variation, $h_t^{(0,2)}$, depends only on the variations of the complex structure $\Delta_A$.  To prove this we vary the compatibility condition \eqref{eq:compatibility} 
\[ \Omega\wedge\omega = 0~,\]
which expresses the fact that with respect to the complex structure $J$ determined by the $(3,0)$-form $\Psi = e^{2\phi}\, \Omega$, the hermitian form $\omega$ is a $(1,1)$-form.  
Varying this equation we find
\[0= \partial_t\omega\wedge\Omega + \omega\wedge\partial_t\Omega = h_t\wedge\Omega+ \omega\wedge\chi_t~,\]
where
\[ \Delta_t = (\partial_t Z^A)\,\Delta_A~, 
\quad{\rm and}\quad \chi_t = (\partial_t Z^A)\,\chi_A~.\]
Contracting with $\bar\Omega$ we obtain
\begin{equation*}
  h_t^{(0,2)} = (\partial_t Z^A)\,h_A^{(0,2)}
 \end{equation*}
 where
 \begin{equation}
 h_A^{(0,2)} = \Delta_A{}^m\wedge \omega_{mn}\, \d x^n~,
 \label{eq:varpureomega}
 \end{equation}
and where we have used equation \eqref{eq:chi}.  Therefore, the $(0,2)$-part of the variation of $\omega$ is entirely  determined by 
the allowed variations of the complex structure of $X$, and there are no new moduli associated to $h_A^{(0,2)}$.

We would like to remark that it has been known for over 20 years in mathematics that the conformally balanced condition is not stable under deformations of the complex structure \cite{MR1107661, MR1137099, 0735.32009, 0793.53068}.  This is in sharp contrast with the Theorems of Kodaira and Spencer for the stability of the K\"ahler condition under deformations of the complex structure \cite{MR0112154, 0128.16902}. 

Returning now to the variation of the conformally balanced condition for the hermitian structure we have the following proposition.

\begin{Proposition}\label{prop:CBcond}
Let $\widehat\lrcorner$ and $\hat *$ be the contraction operator and the Hodge dual operator with respect to the Gauduchon metric respectively. As long as we assume that the tangent bundle is stable and has zero slope, the variation of the conformally balanced condition for the hermitian structure
\begin{equation}
\d \p_t\hat\rho = 0~,\label{eq:dvarhatrho}
\end{equation}
determines the $\bp$-exact part of the Hodge decomposition of the $(1,1)$-form $(\p_t\omega)^{(1,1)}$
in terms of deformations of the complex structure and dilaton, leaving a $\bp^{\tilde\dagger}$-closed part undetermined. Here $\bp^{\tilde\dagger}$ is the adjoint with respect to the rescaled metric $\tilde g=e^{-2\phi}g$.
\end{Proposition}

\begin{proof}
From equation \eqref{eq:dvarhatrho} we get
\begin{equation}
\label{eq:dvarhatrho2} 
 \d^{\widehat\dagger} (2\hat \lambda_t \,\widehat\omega - \hat h_t^{(1,1)} + \hat h_t^{(0,2)} + \hat h_t^{(2,0)}) = 0~,
 \end{equation} 
where $d^{\widehat\dagger}$ is the adjoint of the exterior derivative $\d$ with respect to the Gauduchon metric, and where we have used equations \eqref{eq:varomega}, and $\hat h_t = e^{-\phi}\, h_t$ and $\hat\lambda_t = \lambda_t - \p_t\phi$. Here
\begin{equation*}
\p_t\widehat\omega=\hat\lambda_t\widehat\omega+\hat h_t\:,\;\;\;\;\p_t\hat\rho=2\hat\lambda_t\,\hat\rho+\widehat\omega\wedge\hat h_t\:.
\end{equation*}
Consider the $(1,0)$-part of equation \eqref{eq:dvarhatrho2}
\begin{equation}
\p^{\hat\dagger} \hat h_t{}^{(2,0)} = \bp^{\hat\dagger}( - 2\hat\lambda_t\, \widehat\omega + \hat h_t{}^{(1,1)})~.
\label{eq:eqforh}
\end{equation}
Recall from section \ref{eq:usual} that we may assume that
\[ H_{\bp}^{(2,0)}(X)\cong H_{\bp}^0(TX) = 0~,\]
where the first isomorphism is due the $\Omega$ isomorphism (that is,
for every element in $s^m\in H_{\bp}^0(TX)$ we have an element in $H_{\bp}^{(2,0)}(X)$ given by $s^m\,\Omega_{mnp}$).
The Hodge decomposition of $\hat h_t^{(2,0)}$ in terms of the Laplacian $\widehat\Delta_{\bp}\,$ requires that
\begin{equation*}
\hat h_t^{(2,0)}=\bp^{\hat\dagger}\Lambda_t^{(2,1)},
\end{equation*}
for some $(2,1)$-form $\Lambda$. Recall that $\hat h_t^{(2,0)}$ is completely determined by the complex structure deformations. It follows that $\Lambda^{(2,1)}$ is also given in terms of complex structure variations.  Equation \eqref{eq:eqforh} can now be written as
\begin{equation}
\label{eq:bpdagger}
\bp^{\hat\dagger} (\hat h_t{}^{(1,1)} - 2\hat\lambda_t\, \widehat\omega ) 
= - \bp^{\hat\dagger} ( 
 \p^{\hat\dagger}\Lambda_t^{(2,1)})~,
\end{equation}
which means that the left hand side is entirely determined by variations of the complex structure.  

Next, recall the diffeomorphism action on $\hat\rho$, \eqref{eq:gaugehatrho}. We may decompose $\L_v\hat\rho$ as
\begin{equation}
\label{eq:hodgegaugerho}
\L_{diff}\hat\rho=k\hat\rho+\widehat\omega\wedge P\:,
\end{equation}
for some function $k$ and primitive two-form $P$. Take the wedge product of \eqref{eq:hodgegaugerho} with $\widehat\omega$ to get
\begin{equation*}
\widehat\omega\wedge\L_{diff}\hat\rho=k\,\widehat\omega\wedge\hat\rho=-2\d(\tilde v\wedge\hat\rho)\:.
\end{equation*}
It follows that $k\in\Im(\d^{\hat\dagger})$. By a suitable diffeomorphism, we can hence remove any non-constant part of $\hat\lambda_t$. Equation \eqref{eq:bpdagger} can then be rewritten as
\begin{equation*}
\bp^{\hat\dagger}(\p_t\widehat\omega) =  - \bp^{\hat\dagger} (\p^{\hat\dagger}\Lambda_t^{(2,1)})\:.
\end{equation*}
This can further be rewritten as
\begin{equation*}
\bp^{\tilde\dagger}(\p_t\omega)=\bp^{\tilde\dagger}(\omega\p_t\phi+\p^{\tilde\dagger}\Lambda_t^{(2,1)})\:,
\end{equation*}
where we have rescaled the metric $\tilde g=e^{-\phi}\hat g=e^{-2\phi}g$. Using the Hodge decomposition, we find that this equation determines, as claimed, the $\bp$-exact part of the $(1,1)$-form $\p_t\omega^{(1,1)}$
in terms of deformations of the complex structure and dilaton, and leaves the $\bp^{\tilde\dagger}$-closed part undetermined. 
\end{proof}

By enforcing the anomaly cancelation condition, we will able to fix these parameters further. In fact, we will argue later in section \ref{sec:defD} that, by enforcing the anomaly cancelation condition, the moduli space of the (complexified) hermitian form is finite dimensional and related to the cohomology group
\begin{equation*}
H^{(0,1)}_{\bp}(T^*X)\:.
\end{equation*} 

\section{Deformations of the Holomorphic Bundles}
We now wish to include the holomorphic connections on the bundles into the story. That is, the gauge connection $A$ and tangent bundle connection $\nabla$. Working with holomorphic tangent bundles, we will often find it more convenient to use holomorphic and anti-holomorphic indices $\{a,\bar b\}$ corresponding to the holomorphic and anti-holomorphic coordinates $\{z^a,\bar z^{\bar b}\}$. We start with the gauge connection, i.e. we consider deformations of the holomorphic structure of $V$.

\subsection{Deformations of the Holomorphic Structure on $V$}
\label{subsec:defsV}

The study of deformations of the holomorphic bundles has a long history in mathematics. Of particular relevance is  the work of Atiyah\cite{MR0086359} which considers the parameter space of simultaneous deformations of the complex structure on a manifold $X$ together with those of the holomorphic structure on $V$.  This work has already been applied to the case in which $X$ is a Calabi-Yau manifold\cite{Anderson:2010mh, Anderson:2011ty}, and in this section we extend it to the more general case of a manifold with a heterotic $SU(3)$-structure. We will do this in detail, even though not much is different for this part of the parameter space, as it is the structure that we encounter here that generalises when we include the more complicated anomaly cancelation condition.

Consider now a one parameter family of heterotic compactifications $(X_t, V_t)~t\in\IC$
where we set  $(X_0, V_0) = (X, V)$. We study simultaneous deformations of the complex structure $J$ and the holomorphic structure on $V$.  Hence the variation with respect to $t$ of any mathematical quantity $\beta$ (which may have values in $V$ or $\End V$) is given by the chain rule as follows
\begin{equation*}
\partial_t\beta = (\partial_t Z^A)\, \partial_A\beta+ 
(\partial_t Z^{\bar A})\, \partial_{\bar A}\beta + (\partial_t y^i)\, \partial_i\beta
+ (\partial_t \lambda^\alpha)\, \partial_\alpha\beta
+ (\partial_t \lambda^{\bar\alpha})\, \partial_{\bar\alpha}\beta
\end{equation*}
where we label the bundle moduli by $\lambda^\alpha$.

Let $F$ be the curvature of the bundle connection $A$, where
\begin{equation}
F = d_A^2=\d A + A\wedge A\in\Omega^2(\End(V))\:,
\label{eq:curvV}
\end{equation}
and where $A\in\Omega^1(\End(V))$ is the gauge potential. A holomorphic structure on $V$ is determined by the derivative $\bp_A$ which is defined as the $(0,1)$-part of the operator $\d_A$. It is easy to prove that $\bp_A^2 = 0$ iff $F^{(0,2)} = 0$. Upon a choice of holomorphic trivialisation \eqref{eq:holtriv}, we may set $A=a$, and so $\bp_A=\bp$. We will use $\bp_A$ or $\bp$ interchangeably when dealing with the holomorphic bundle, depending on the circumstance. 


Consider now what happens to the holomorphicity of the bundle $V$ under deformations of the complex structure of $X$. Varying equation \eqref{eq:holF}, it is easy to see that
\begin{equation}
 (\partial_B F)^{(0,2)} = \Delta_B{}^a\wedge F_{a\bar b}\, \d z^{\bar b}~,\label{eq:varpureFone}
 \end{equation}
On the other hand, varying \eqref{eq:curvV} we find that
\begin{equation}
 (\partial_B F)^{(0,2)} = \bp\,\alpha_B~,\label{eq:varpureFtwo}
 \end{equation}
 where $\alpha_B$ is the non-$\bp$-closed $(0,1)$-part of the variation of $A$. Putting together equations \eqref{eq:varpureFone} and \eqref{eq:varpureFtwo} we find
 \begin{equation}
  \bp\,\alpha_B = \Delta_B{}^a\wedge F_{a\bar b}\, \d x^{\bar b}~.\label{eq:atiyahF}
 \end{equation}
 This equation represents a constraint on the possible variations $\Delta_B$ of the complex structure $J$ on $X$.
 
Consider the map
\begin{equation}
\label{eq:mapF}
{\F}\;:\;\Omega^{(0,q)}(TX)\longrightarrow \Omega^{(0,q+1)}(\End(V))
\end{equation}
given by
\begin{equation}
{\F}\big(\Delta\big)= (-1)^{q}\, \Delta{}^a\wedge F_{a\bar b}\, \d x^{\bar b}~,
\qquad \Delta\in\Omega^{(0,q)}(TX)~.
\label{eq:mapFdef}
\end{equation}
We have the following Theorem: 
\begin{Proposition}\label{tm:one}
\begin{equation}
\bp\left({\F}\big(\Delta\big)\right) = - {\F}\left(\bp\Delta\right)~,
\quad\forall\, \Delta\in H_{\bp}^{(0,q)}(TX)~,
\label{eq:Fcohom}
\end{equation}
and therefore the map ${\F}$, referred to as the {\it Atiyah map for $F$}, is a map between cohomologies\footnote{We will only consider the Atiyah maps acting on $(0,q)$-forms in this thesis. The generalization to $(p,q)$-forms is straight forward.}
\begin{equation}
\label{eq:mapFcohom}
{\F}\;:\;H^{(0,q)}(TX)\longrightarrow H^{(0,q+1)}({\rm End}(V))~.
\end{equation}
\end{Proposition}
\begin{proof}
This follows from the fact that we have a holomorphic connection, so $F$ is a $(1,1)$-form, and the Bianchi identity,
\[\bp F = 0~.\]
\end{proof}

\noindent By the Bianchi identity, this map defines an element in the cohomology
\begin{equation}
\label{eq:AtiyahFcohom}
[\F]\in H^{(0,1)}(\End(V)\otimes T^*X)\:.
\end{equation}
It is worth noting that elements of the same class $[\F]$ give rise to holomorphically equivalent bundles. It is also worth remarking that this map is well defined as a map between cohomologies, due to the fact that under gauge transformations the curvature $F$ is covariant.\footnote{A gauge transformation does however spoil the holomorphic trivialisation in general.}

In terms of the map $\F$, the constraint \eqref{eq:atiyahF} on the variations of the complex structure $\Delta _A\in H^{(0,1)}(TX)$ can now be written as
\begin{equation}
 \bp\,\alpha_A = - {\F}\big(\Delta_A\big)~.\label{eq:atiyahFbis}
 \end{equation}
So  ${\F}\big(\Delta_A\big)$ must be exact in $H^{(0,2)}(\End(V))$, in other words
\[ \Delta _A\in \ker({\F})\subseteq H_{\bp}^{(0,1)}(TX)~.\]
The tangent space $T\M_1$ of the moduli space of combined deformations of the complex structure and bundle deformations, is therefore given by
\begin{equation}
T\M_1= H^{(0,1)}(\End(V))\oplus \ker({\F})~ ,\label{eq:M1}
\end{equation}
where $H^{(0,1)}(\End(V))$ is the space of bundle moduli. These correspond to $\bp$-closed forms that can be added to $\alpha_A$ without changing \eqref{eq:atiyahFbis}.

These results can be restated in a way that will be suitable for generalisations later when we include the other constraints on the heterotic compactification $(X,V)$. Define a bundle ${\cal Q}_1$ which is the extension of $TX$ by $\End(V)$, given by the short exact sequence
  \begin{equation}
   0\rightarrow \End(V)\xrightarrow{\iota_1} \Q_1 \xrightarrow{\pi_1} TX \rightarrow 0~,   \label{eq:ses1}
   \end{equation}
or $\Q_1=TX\oplus\End(V)$, with some extension class $\F\in\Omega^{(0,1)}(T^*X\otimes\End(V))$. That is on $\Q_1$ we define the $(0,1)$-connection
\begin{equation*}
\bp_1=\bar\nabla+\F\:,
\end{equation*}
where $\bar\nabla$ is the connection on the individual bundles. In matrix form, this connection takes the form
 \begin{equation}
 \label{eq:bp1block}
\bp_1\;=\; \left[ \begin{array}{cc}
\bar\nabla_{\End(V)} & \;{\F} \\
0  & \bar\nabla_{TX}  \end{array} \right]\:,
\end{equation}
when acting on $\End(V)\oplus TX$. 
Note that this is the most general $(0,1)$-connection on the extension bundle $\Q_1$ that we can write, i.e. the most general connection respecting the extension sequence \eqref{eq:ses1}. It should however be noted that upon a local $\Q_1$-valued coordinate transformation, any connection on $\Q_1$ can locally be put in an upper-right block-diagonal form, as in \eqref{eq:bp1block}. There are however obstructions to doing so globally, and there are issues in how to interpret the transformed bundles. We intend to investigate these issues in an upcoming publication \cite{delaOssaSoon}.




Next, note that $\bp_1$ holomorphic is equivalent to
\begin{equation*}
\bp_1^2=0\;\;\;\;\Leftrightarrow\;\;\;\;\bar\nabla^2=0\:,\;\;\;\;\bar\nabla\F=0\:.
\end{equation*}
It follows that the individual bundles are holomorphic. Hence $\bar\nabla_{\End(V)}$ and $\bar\nabla_{TX}$ determine holomorphic structures on $V$ and $TX$ respectively. With a choice of holomorphic trivialisation \eqref{eq:holtrivnabla}, we may set $\bar\nabla=\bp$. Note then also that as $\bp\F=0$, we have $\F$ is the curvature of some connection $a$ on $V$, that is $\F=\bp a$ by the Poincare lemma.

The infinitesimal moduli space of the holomorphic structure $\bp_1$ on the extension bundle $\Q_1$, which is given by
\[T\M_1 = H^{(0,1)}_{\bp_1}(\Q_1)~,\]
can be computed by a long exact sequence in cohomology of the sequence \ref{eq:ses1},
\begin{equation} 
\begin{split}
0 &\rightarrow H^{(0,1)}(\End(V)) \xrightarrow{\iota_1'} H^{(0,1)}(\Q_1) \xrightarrow{\pi_1'} H^{(0,1)}(TX) \\
&\xrightarrow{\F} H^{(0,2)}(\End(V)) \rightarrow H^{(0,2)}(\Q_1) \rightarrow \ldots\:.
\label{eq:les1}
\end{split}
\end{equation}
Here the Atiyah map $\F$ is the connecting homomorphism, whose definition is the usual one
\begin{equation}
[\iota_1^{-1}\circ\bp_1\circ \pi_1^{-1}(x)]=[{\F}(\Delta)]\:,
\end{equation}
where we have used the definition of $\bp_1$ above. This definition is valid by Proposition \ref{tm:one}.

In the computation of the long exact sequence \eqref{eq:les1}, we have used
\[H^0(TX) = 0~,\]
which is a valid assumption as we saw in section \ref{eq:usual}. Thus, we also have
\begin{equation}
 H^0(\Q_1) \cong H^0(\End(V))~.\label{eq:nosecQ1}
 \end{equation} 
Recall that for a stable bundle $V$   
\[ {\rm dim} H^0(\End(V))\le 1~.\] 
There are non-trivial sections whenever the Lie algebra $\End(V)$ has non-vanishing trace. Then,  for a poly-stable bundle
\begin{equation*}
V=\oplus_{i=1}^nV_i~,
\end{equation*}
we have 
\[\textrm{dim}(H^0(\End(V)))=\tilde n-1~\]
where $\tilde n$ is the number of bundle factors which have  endomorphisms non-vanishing trace, and we subtract one as the overall trace should vanish.

Finally, we find, by exactness of the sequence \eqref{eq:les1}, that
\[T\M_1= H_{\bp_1}^{(0,1)}(\Q_1) = {\rm Im}(\iota_1') \oplus {\rm Im}(\pi_1') \cong
H^{(0,1)}(\End(V))\oplus \ker({\F}) ~,\]
in agreement with equation \eqref{eq:M1}.

\subsection{Deformations of the Holomorphic Structure on $TX$.}
\label{subsec:defsTX}

We now extend our results to include deformations of the holomorphicity condition \eqref{eq:holR} of the tangent bundle $TX$ under deformations of the complex structure of $X$. Basically, we repeat the analysis above.  Let $R $ be the curvature of the instanton connection $\nabla$ 
\begin{equation}
R  = \d_\Theta^2=\d \Theta  + \Theta \wedge \Theta\in\Omega^2(\End(TX)) ~,\label{eq:curvR}
\end{equation}
where $\Theta \in\Omega^1(\End(TX))$ is the connection one-form of $\nabla$. Note again that as ${R}^{(0,2)}=0$, i.e. $(TX,\nabla)$ is holomorphic, we may set set $\d_{\Theta }^{(0,1)}=\bp_{\Theta }=\bp$ by a choice of local holomorphic trivialisation. Again, we use $\bp_{\Theta }$ and $\bp$ interchanginbly, depending on the circumstance. 


We now add the bundle $\End(TX)$ to $\End(V)$, to get 
\begin{equation*}
\g=\End(TX)\oplus\End(V)\:,
\end{equation*}
and in a similar fashion as above we define a holomorphic extension bundle $E=\g\oplus TX$, given by the short exact sequence
 \begin{equation}
  0\rightarrow \g\xrightarrow{\iota_2} E \xrightarrow{\pi_2} TX \rightarrow 0~,
  \label{eq:ses2}
  \end{equation}
with extension class
\begin{equation}
\label{eq:AtiyahRcohom}
[\B]=[\R]+[\F]\in H^{(0,1)}(T^*X\otimes\g)\:,
\end{equation}
given by the curvatures $R$ and $F$. I.e. there is a holomorphic structure on $E$ defined by the exterior derivative $\bar\partial_2$ 
 \begin{equation*}
\bp_2=\bar\nabla+\B=\bar\nabla+\F+\R 
\end{equation*}
which acts on $\Omega^{(0,q)}(E)$ and squares to zero by the Bianchi identity of $R $. Here $\bar\nabla$ is a connection on the individual bundles involved, where we exclude connections between endomorphism bundles. Again, in matrix form, $\bp_2$ is given by 
 \begin{equation*}
\bp_2\;=\;\left[ \begin{array}{cc}
\bar\nabla_{\g} & \;\B\\
 0 & \bar\nabla_{TX} \end{array} \right]
\:,
\end{equation*}
on $E=\g\oplus TX$. Note again that 
$\bp_2$ is the most general holomorphic connection on the extension bundle $E$, excluding connections between the endomorphism bundles.\footnote{Again, any connection on $E$ can locally be put in this form, upon a local $E$-valued coordinate transformations. As noted, there might be obstructions to doing so globally \cite{delaOssaSoon}.}
We proceed to compute the infinitesimal deformations.

The infinitesimal moduli space of the holomorphic structure $\bp_2$ on the extension bundle $E$, which is given by
\[T\M_2 = H^{(0,1)}_{\bp_2}(E)~,\]
can be computed by a long exact sequence in cohomology as in the previous section
\begin{align*}
0 &\xrightarrow{\B} H^{(0,1)}(\g) \xrightarrow{\iota_2'} H^{(0,1)}(E) \xrightarrow{\pi_2'} H^{(0,1)}(TX) \\
&\xrightarrow{\B} H^{(0,2)}(\g) \rightarrow H^{(0,2)}(E) \rightarrow \ldots\:,
\end{align*}
where the Atiyah map $\R+\F$ is the connecting homomorphism, and we have again set $H^0(TX)=0$. The zeroth order cohomology, or ``sections" of $E$, then read
\begin{equation}
 H^0(E) \cong H^0(TX)\oplus H^0(\g)\cong H^0(\End(V))~.\label{eq:nosecQ2} 
 \end{equation}
The last equality follows from the stability of $TX$, and the fact that the endomorphisms in $spin(6)$ are traceless, which implies that for traceless endomorphisms of $TX$ there are no holomorphic sections with values in $\End(TX)$
\[ H^0({\rm End} (TX)) = 0~.\]
Then we find that the infinitesimal moduli space of the extension $E$ is
\[T\M_2= H_{\bp_2}^{(0,1)}(E) =
H^{(0,1)}(\End(TX))\oplus
H^{(0,1)}(\End(V))\oplus (\ker({\F})\cap\ker({\R} )) 
~.\]
We remark again that the deformations in $H^{(0,1)}(\End(TX))$ should not correspond to any physical fields, but are needed for the implementation of the Bianchi identity, which we come to next.

\section{Including the Bianchi identity}
\label{subsec:Anomaly}

We now wish to include the Bianchi identity 
\begin{equation}
\label{eq:BIanomaly3}
\d H = -2 i \partial\bar\partial\omega = \frac{\a}{4}\,\left(\tr\,F\wedge F - \tr\,R \wedge R\right).
\end{equation}
into the story. We will see that the proceeding structure of extensions is useful in this regard. That is, we will construct a holomorphic extension bundle $\Q$ of $E$ such that $\Q$ has a holomorphic structure, and which allows for the implementation of the Bianchi identity \eqref{eq:BIanomaly3}. Moreover, using deformation theory of holomorphic bundles, we will show that this construction results in a description of the moduli space of heterotic compactifications $(X,V)$.  

Before we move on, we make the following observation that the map $\B$ also acts naturally as maps
\begin{equation*}
\B \::\:H^{(0,q)}(\g)\rightarrow H^{(0,q+1)}(T^*X)\:.
\end{equation*}
It therefore seems natural to extend the bundle $E$ to
\begin{equation*}
\Q=T^*X\oplus\g\oplus TX
\end{equation*}
as a smooth bundle, and consider holomorphic connections on $\Q$. Written as an extension bundle, we have
\begin{equation}
0\rightarrow T^*X\xrightarrow{\iota}\Q\xrightarrow{\pi}E\rightarrow0\:,\label{eq:sesH}
\end{equation}
with extension class which we call $\H$, that we shall return to below.


The most general $(0,1)$-connection we can write on $\Q$, as an extension bundle given by \eqref{eq:sesH}, is\footnote{A general connection on $\Q$ can again locally be put in this form, upon a local $\Q$-valued coordinate transformation. This can again be obstructed globally \cite{delaOssaSoon}.}
\begin{equation*}
\bD=\bar\nabla+\B+\hat{H}\:,
\end{equation*}
where $\B\in\Omega^{(0,1)}(T^*X\otimes\g)$, $\hat{H}\in\Omega^{(0,1)}(T^*X\otimes T^*X)$. Here $\bar\nabla$ is the connection on the individual bundles. Again, in matrix form this looks like
 \begin{equation}
 \label{eq:barDmatrix}
\bD\;=\; \left[ \begin{array}{ccc}
\bar\nabla_{T^*X} & \B_2 & \hat{H} \\
0 & \bar\nabla_{\g} & \B_1 \\
0  & 0 & \bar\nabla_{TX}  \end{array} \right]
\end{equation}
where both $\B_i\in\Omega^{(0,1)}(T^*X\otimes\g)$. It should be noted that in our case $\B_1=\B_2=\B$. This is equivalent to that $\Q$ is self-dual as a holomorphic bundle. Note also that this condition is not spoiled under infinitesimal deformations, as this is a condition on the {\it extension classes} which do not change under infinitesimal deformations. 


Requiring $\bD$ to be holomorphic, i.e. $\bD^2=0$, we first get
\begin{equation*}
\bar\nabla^2=0\:,
\end{equation*}
so we may again set $\bar\nabla=\bp$ by a choice of holomorphic trivialisation. Furthermore, we need $\bp\B=0$, and so 
\begin{equation*}
[\B]\in H^{(0,1)}(T^*X\otimes\g)=H^{(0,1)}(T^*X\otimes\End(V))+H^{(0,1)}(T^*X\otimes\End(TX))\:.
\end{equation*}
We thus have
\begin{equation*}
\B=c_1\F+c_2\R\:,
\end{equation*}
i.e. $\B$ is the given by curvatures of connections on the bundles, where we have included some overall scalings $c_i$, needed in order to relate the structure to the heterotic Bianchi identity. Note that these scalings are unimportant as far as the deformation problems of holomorphic structures in the previous sections go. 

Finally, we also get the condition
\begin{equation}
\label{eq:pont0}
\bp\hat{H}+\B\wedge\B=0\:.
\end{equation}
Note that this condition forces $\hat H\in \Omega^{(0,1)}(\Lambda^2T^*X)$.\footnote{In principle, it is possible to add a $\bp$-closed part to $\hat H$ which is non-skew. We will not consider this in this thesis.} By choosing the constants $c_i$ appropriately, this can be related to the heterotic Bianchi identity. Indeed, by setting $c_1=\frac{\sqrt{\a}}{2}$ and $c_2=i\frac{\sqrt{\a}}{2}$, \eqref{eq:pont0} can be rewritten as
\begin{equation}
\label{eq:redefB}
\frac{\a}{4}\Big(\tr\:F\wedge F-\tr\:R \wedge R \Big)=-2\bp\hat H\:,
\end{equation}
where
\begin{equation*}
\hat H=\frac{1}{2}\hat H_{\bar c ab}\d z^{\bar c ab}\:.
\end{equation*}
We proceed to relate this to the heterotic Bianchi identity \eqref{eq:BIanomaly3}.

Note that up until now, the holomorphic structures constructed have not included the metric $g$. Strictly speaking, there is no need of a metric in order to construct a holomorphic structure on a bundle. We shall continue with this philosophy, by letting the metric, or hermitian structure $\omega$, be defined {\it as part of the holomorphic structure} $\bar D$. It should further be noted that the heterotic Bianchi Identity  \eqref{eq:BIanomaly} is a real equation. So far in the deformation stary, we have been working in the holomorphic parts of the bundles where everything is complex. In order to get to the heterotic Bianchi Identity, we need to take the real part of \eqref{eq:redefB} to get
\begin{equation}
\label{eq:redefB2}
-2\bp\hat H-2\p\bar{\hat H}=\frac{\a}{4}(\tr\:F^2+\tr\:\bar F^2-\tr\:R^2-\tr\:\bar R^2)\:,
\end{equation}
where ${\bar F_{\bar i}}^{\bar j}={{\p\bar A}_{\bar i}}^{\bar j}$, ${\bar R_{\bar\alpha}}^{\bar\beta}={\p\bar\Theta_{\bar\alpha}}^{\bar\beta}$, and $\bar A$ and $\bar\Theta$ are connections on the anti-holomorphic bundles $\End(\bar V)$ and $\End(\overline{TX})$ respectively. Choosing $\hat  H=-\frac{i}{2}\p\gamma+\bp\textrm{-closed}$, where $\gamma\in\Omega^{(1,1)}(X)$ and where ${\rm Re}(\gamma)=\omega$ is the heterotic hermitian form, we arrive at the heterotic Bianchi Identity \eqref{eq:BIanomaly}.


Having shown that the the holomorphic conditions of the Strominger system  and the Bianchi Identity embeds into the holomorphic structure $\bar D$, one may wonder if any choice of holomorphic structure $\bar D$ corresponds to some Strominger system? We now argue that this is indeed the case, provided we are in the large volume limit for $X$, where the $\a$-expansion is valid. 

Indeed, as is well known, the zeroth order in $\a$ solution of heterotic supergravity requires the compact space $X$ to be a Calabi-Yau. In the large volume limit, one assumes that $\a$-corrections are small. In particular, we assume that the corrections are purely geometric in nature. It follows that there exists a closed K\"ahler form $\tilde\omega$, making $(X,\tilde\omega)$ a Calabi-Yau, and so that
\begin{equation*}
\omega=\tilde\omega+\a\omega_1=\tilde\omega+\OO(\a)\:.
\end{equation*}
Furthermore, the existence of a K\"ahler form is enough to ensure that $X$ satisfys the $\p\bp$-lemma. It follows that in the large volume limit, we can assume that $X$ satisfies the $\p\bp$-lemma. As the right hand side of \eqref{eq:redefB2} is $\d$-closed, using the $\p\bp$-lemma, \eqref{eq:redefB2} can be rewritten as
\begin{equation*}
-2i\p\bp\omega_1=\frac{1}{4}\Big(\tr\:F\wedge F-\tr\:R\wedge R\Big)\:,
\end{equation*}
where the curvatures are now the curvatures on the full bundles. $\omega_1$ then denotes the first-order non-closed correction to the hermitian form. It should also be noted that modulo diffeomorphisms, we can also assume that the dilaton is constant, that is $\omega$ is balanced \cite{Gillard:2003jh, Anguelova:2010ed}.

Outside of the large volume limit the story is more subtle. Indeed, here we should really include higher orders in $\a$-effects as well, making the Strominger system void. There is a question as to what the correct heterotic ``supergravity" theory is in this case, but whatever the theory is, the Strominger system should correspond to supersymmetric vacuum solutions of such a theory in the large volume limit. Rewriting the vacuum as a condition on the connection $\bar D$ might give a clue to what the solution looks like outside of the large volume limit. Indeed, by replacing the Bianchi Identity with a holomorphic condition on $\bar D$, we have effectively removed $\a$ from the system. As there is no inherent size associated to a holomorphic structure, this condition also makes sense outside of the large volume limit. Such a condition therefore has a shot of describing the correct vacuum to all orders in $\a$. We discuss this more in the discussion section.

\subsection{The Extension Class $\H$}
Note that we may also write
\begin{equation}
\label{eq:barD}
\bD=\bp+\B+\hat H=\bp_2+\H\:,
\end{equation}
where $\bp_2=\bp+\B$ acts as $\bp$ on $TX$, and where $\B$ acts on $E$ by acting on $TX$-valued forms as before. Here we have defined the extension class $\H\in\Omega^{(0,1)}(E^*\otimes T^*X)$, extending $E$ by $T^*X$. That is,
\begin{equation}
\label{eq:classH}
\H=\hat{H}+\B=\hat{H}+c_1\F+c_2\R\:,
\end{equation}
where now $\B$ in this expression acts on endomorphism valued forms by the trace.
Writing $\bar D$ as \eqref{eq:barD} is better when we later come to compute the infinitesimal deformations of $\bD$ in section \ref{sec:defD}, by means of long exact sequences of cohomologies.
\begin{Proposition}
\label{prop:five}
\begin{equation}
\bp({\H}(x)) = - {\H}(\bp_2(x))~,\quad\forall x\in \Omega^{(0,q)}(E)~, 
\label{eq:Hcohom} 
\end{equation}
and therefore the map $\H$ is a map between cohomologies
\begin{equation}
{\H}\;:\;H_{\bp_2}^{(0,q)}(E)\longrightarrow H_{\bp}^{(0,q+1)}(T^*X).\label{eq:mapHcohom}
\end{equation}
\end{Proposition}

\begin{proof}

This is a direct consequence of the fact that $\bD^2=0$. Indeed, we have by \eqref{eq:barD}
\begin{equation*}
0=\bD^2=\bp_2^2+\bp\H+\H\bp_2=\bp\H+\H\bp_2\:.
\end{equation*}


\end{proof}




The Atiyah map $\H$ is well defined as a map between cohomologies.  To see this we need to prove that the class ${\H}(x)\in H_{\bp}^{(0,q+1)}(T^*X)$ 
is invariant under gauge transformations. As gauge transformations generically spoil the holomorphic trivialisation, we return to using the operators $\bp_A$ and $\bp_\Theta$. Consider an element $\F(\alpha)$ as part of $\H(x)$, where $\alpha\in\Omega^{(0,q)}(\End(V))$. 
Recall that under a gauge transformation
\[ {\cal A} \mapsto \Phi({\cal A} - \Phi^{-1}\bp \Phi)\Phi^{-1}~,\]
where $\Phi$ takes values in the Lie algebra of the structure group of the bundle $V$.  
This implies that
\[ \alpha_t \mapsto  \Phi(\alpha_t - \bp_{A}(\Phi^{-1}\p_t \Phi))\Phi^{-1}~.\]
It can then be shown that under a gauge transformation
\[ \alpha \mapsto \Phi(\alpha - \bp_{A}Y)\Phi^{-1}~,\]
for some $Y\in\Omega^{(0,q-1)}(\End(V))$.  Hence
\begin{equation*}
\begin{split}
 \F(\alpha)=\tr(f_a\wedge\alpha)&\mapsto
 \tr(\Phi f_a\Phi^{-1}\wedge\Phi(\alpha - \bp_{A}Y)\Phi^{-1}) \\ 
&\mapsto \tr(f_a\wedge\alpha) + \bp(\tr(f_a\wedge Y)) - \tr(\bp_{A}f_a\wedge Y)~,
 \end{split}
 \end{equation*}
 where $f_a=F_{a\bar b}\,\d z^{\bar b}$. As the last term vanishes due to the Bianchi identity for $F$, we find that under a gauge transformation ${\H}(x)$ changes only by a $\bp$-exact part, and therefore the class ${\H}(x)\in H_{\bp}^{(0,q+1)}(T^*X)$ is gauge invariant. The argument for the other term $\R(\kappa)$, $\kappa\in\Omega^{(0,q)}(\End(TX))$, in $\H(x)$ is similar.

We note that the Atiyah class $\H$ should also be an element of some cohomology. Indeed, 
by Proposition \ref{prop:five}, we have for $x\in \Omega^{(0,q)}(E)$,
\begin{equation*}
0=\bp(\H(x))+\H(\bp x+\B x)=(\bp\H)(x)+\H (\B x)=(\bp\H+\B\H)(x)=(\bp_2\H)(x)\:,
\end{equation*}
where now $\B$ acts as a connection on $\H$, which has values in $E^*$. It follows that 
\begin{equation*}
\bp_2\H=0\:,
\end{equation*}
and so
\begin{equation}
\label{eq:AtiyahHcohom}
[\H]\in H^{(0,1)}_{\bp_2}(E^*\otimes T^*X)\:,
\end{equation}
as was also shown in \cite{2013arXiv1308.5159B, Anderson:2014xha}. 

For completeness, we now compute this cohomology following \cite{Anderson:2014xha}. First dualise the extension sequence \eqref{eq:ses2}, and tensor this by $T^*X$ to get
\begin{equation}
\label{eq:sesExtH}
 0\rightarrow TX^*\otimes T^*X\xrightarrow{\iota_2} E^*\otimes T^*X \xrightarrow{\pi_2} \g\otimes T^*X \rightarrow 0\:.
\end{equation}
The Atiyah map is again given by $\B=c_2\R+c_1\F$. By a long exact sequence in cohomology of \eqref{eq:sesExtH}, we get that
\begin{equation*}
H^{(0,1)}_{\bp_2}(E^*\otimes T^*X)\cong \Big[H^{(0,1)}(TX^*\otimes T^*X)\Big/\Im(\B)\Big]\oplus\ker(\B) \:.
\end{equation*}
Note then that the $\B$-part of $\H=\B+\hat H$ is clearly in $\ker(\B)$ by the Bianchi identity \eqref{eq:pont0}. These are the allowed field strengths of the bundles. The part $H^{(0,1)}(TX^*\otimes T^*X)$ then corresponds to the $\bp$-closed element $\hat H_0$ of $\hat H$, which can be added without changing this Bianchi identity. To understand why we should mod out by $\Im(\B)$, consider
\begin{equation*}
\B(x)\in\Im(\B)\subseteq H^{(0,1)}(TX^*\otimes T^*X)\:.
\end{equation*}
where $x\in H^0(\g\otimes T^*X)$. 
As a map, this acts on $TX$-valued forms. Let $\delta\in\Omega^*(TX)$. It is easy to see that
\begin{equation*}
\B(x)(\delta)=\B(x_a\delta^a)\:.
\end{equation*}
The map $\B(x)$ can therefore be regarded as part of the $\B$-map of $\H$ instead. 
It follows that we should mod out $H^{(0,1)}(T^*X\otimes T^*X)$ by maps in $\Im(\B)$.


\subsection{The Strominger System as a Yang-Mills Connection}
\label{sec:HYM}

Having seen that most of the Strominger system can be put in terms of a holomorphic connection $\bD$ on the extension bundle $\Q$, we also want to implement the Yang-Mills conditions into the story. We do this by imposing the instanton condition on $\bD$. This might seem like it could impose extra conditions, but we will see that this is in fact not the case.

We may assume that $\bD$ is the $(0,1)$-part of some connection
\begin{equation*}
\D=\bD+D
\end{equation*}
on $\Q$. Here $D$ is of type $(1,0)$, and $D^2=0$. Recall that
\begin{equation*}
\bD=\bar\nabla+\B+\H\:,
\end{equation*}
where
\begin{equation*}
\B\::\:\Omega^{(0,q)}(TX)\rightarrow\Omega^{(0,q+1)}(\g)\:,\;\;\;\;\H\::\:\Omega^{(0,q)}(E)\rightarrow\Omega^{(0,q+1)}(T^*X)\:,
\end{equation*}
and similar for $D$. $\bar\nabla$ is the connection on the individual bundles.

For the gauge bundle $\End(V)$, the connection $A$ is required to satisfy the Yang-Mills condition. We shall also require that it is hermitian. We have not specified the connections on $TX$ and $T^*X$ parts of $\Q$ yet, other than that they agree with the holomorphic structures on these bundles and we assume they are hermitian connections. In order then for $\D$ to satisfy the Yang-Mills condition, it is easy to see from Appendix \ref{app:HYMconfbal} that we also need the connections on $TX$ and $T^*X$ to do so. Luckily, we know from section \ref{eq:usual} that the Chern-connection $\nabla^c$ corresponding to $\omega$ is Yang-Mills to zeroth order in $\a$,
\begin{equation}
\label{eq:YMZero}
\omega\lrcorner R^c=\OO(\a)\:.
\end{equation}
Performing an infinitesimal gauge transformation on the left-hand side of \eqref{eq:YMZero}, this reads as is shown in Appendix \ref{app:Primitive}
\begin{equation*}
\delta_\epsilon\left(\omega\lrcorner R^c\right)=\Delta_{\bp_{\nabla^c}}\epsilon\:,
\end{equation*}
where $\epsilon$ is proportional to the gauge transformation. If we let $\epsilon=\a\epsilon_1$ we get
\begin{equation*}
\delta_\epsilon\left(\omega\lrcorner R^c\right)=\a\Delta_{\bp_{\nabla^c_0}}\epsilon_1+\OO(\a^2)\:,
\end{equation*}
where the zero denotes the zeroth order geometry. As $\Delta_{\bp_{\nabla^c_0}}$ spans $\Omega^0(\End(TX_0))$ for the zeroth order geometry it follows that the right hand side of \eqref{eq:YMZero} can be chosen to be $\OO(\a^2)$ by an appropriate gauge transformation. In other words, we can take the connections on $TX$ and $T^*X$ to satisfy the Yang-Mills condition to the order we are working. Finally, recall also that the connection on the $\End(TX)$-part of $\Q$ is required to satisfy the Yang-Mills as well. As we will see in the next chapter and as noted above, a change of this connection corresponds to a field redefinition. We may as well choose our fields so that the connection on the $\End(TX)$ part of $\Q$ is Yang-Mills as well.


In appendix \ref{app:HYMconfbal} we show that, on conformally balanced manifolds, the condition of an extended connection to satisfy the Yang-Mills condition is equivalent to the connections on each individual bundles to do so, in addition to the extension class to be harmonic with respect to the Laplacian defined by the Gauduchon metric $\hat g$. As each extension class has a harmonic representative, this does not introduce further constraints. The connection $\D$ on $\Q$ is constructed in a stepwise manner by extensions, as shown above. Provided then that we choose each representative of extension classes harmonic with respect to the Gauduchon Laplacian, it follows that we can assume $\D$ to satisfy a instanton condition on its own. That is, in addition to being holomorphic, we have
\begin{equation}
\label{eq:Dinstanton}
\widehat\omega\wedge\widehat\omega\wedge F_\D=0\:,
\end{equation}
where $F_\D=\D^2$ is the curvature of $\D$. Furthermore, the condition of holomorphy, and the Yang-Mills condition can be derived from an instanton condition
\begin{equation*}
F_{\D mn}\gamma^{mn}\eta=0\:.
\end{equation*}
Interestingly then, the Strominger system can be repackaged in terms of an instanton connection on $\Q$. We discuss possible implications of this in the discussion section \ref{sec:discussSS}. 

The condition \eqref{eq:Dinstanton} appears to be a very natural BPS-condition for heterotic supergravity, potentially valid outside of the large volume limit. It should however be emphasized that in choosing $\D$ to satisfy the Yang-Mills condition we used the fact that the zeroth order geometry is Calabi-Yau, assuming that a large volume limit exists. It would be interesting to see if this is a sensible condition for examples where no zeroth order geometry exisits. E.g. the examples of \cite{Dasgupta:1999ss, Becker:2006et} or non-compact solutions \cite{Carlevaro:2009jx}. This is a subject of further study.


\subsection{Deforming $\bD$}
\label{sec:defD}

We now proceed to deform the holomorphic structure $\bD$, in order to compute the first order deformation space $T\M_{\bD}$. As we shall see, this space corresponds to the first order deformation space of the Strominger system. 

Deformations of the holomorphic structure determined by 
\begin{equation*}
\bD\;=\; \left[ \begin{array}{cc}
\bp & \;{\H}\\
0  & \bp  \end{array} \right]\::\;\;\;\;\Omega^{(0,q)}(\Q)\rightarrow\Omega^{(0,q+1)}(\Q)\:,
\end{equation*}
correspond to elements of $H^{(0,1)}_{\bD}(\Q)$. We will compute this cohomology by the usual means of a long exact sequence in cohomology. Above we have defined a short exact extension sequence \eqref{eq:sesH}, with extension class $\H$. This gives rise to a long exact sequence in cohomology
\begin{equation} 
\begin{split}
& 0 \rightarrow H^0(T^*X)\xrightarrow{\iota'} H^0(\Q) \xrightarrow{\pi'} H^0(E) \\
&\xrightarrow{{\H}_0} H^{(0,1)}(T^*X) \xrightarrow{\iota'} H^{(0,1)}(\Q) \xrightarrow{\pi'} H^{(0,1)}(E) \\
&\xrightarrow{{\H}_1} H^{(0,2)}(T^*X) \rightarrow H^{(0,2)}(\Q) \rightarrow \ldots\:,
\end{split}
\label{eq:les3}
\end{equation}
where the connecting homomorphism  is $\H$, and where we denote by ${\H}_q$ the map $\H$ when we need to make it clear that it is acting on $(0,q)$-forms.


We are now ready to write the infinitesimal moduli space of holomorphic structures of the extension $\Q$. By exactness of the sequence \eqref{eq:les3}, it follows that
\begin{equation}
\label{eq:defBI}
H_{\bD}^{(0,1)}(\Q)\cong\Im(i')\oplus\Im(\pi')\cong \Big[H^{(0,1)}(T^*X)\Big/\textrm{Im}({\H}_0)\Big]\oplus\ker({\H}_1)\:,
\end{equation}
is the tangent space to the moduli space of deformations of the holomorphic structure defined by $\bD$  on $\cal Q$.
The condition 
\begin{equation}
\label{eq:condkerH}
x\in\ker({\H}_1)\subseteq H_{\bp_2}^{(0,1)}(E)
\end{equation}
is required for the deformations $x$ of the holomorphic structure on $E$ to preserve the holomorphic structure $\bD$, including the heterotic Bianchi Identity, and the elements in the factor 
\[ {\cal M}_{HS} = \Big[H_{\bp}^1(T^*X)\Big/\textrm{Im}({\H}_0)\Big]\]
are the moduli of the (complexified) hermitian structure. In the following subsections we interpret in detail the elements in $H_{\bD}^{(0,1)}(\Q)$,  which by construction should be precisely the infinitesimal moduli space of the Strominger system. 

\subsection{Explicit Deformations}

Firstly, we would like to compare our results with those obtained by directly varying the anomaly cancelation condition.  Recall that
\begin{equation}
H =  i\, (\p- \bp)\, \omega = \J(d\omega) = \d B + {\cal C S}~.\label{eq:anomalyA}
\end{equation}
where  
\[ {\cal C S} = \frac{\alpha'}{4}\, (\omega_{CS}^A-\omega^\nabla_{CS})~.\]
The variations of equation \eqref{eq:anomalyA} are given by \cite{Anderson:2014xha, Candelas2014}
\begin{equation}
\begin{split}
\p_t H &= \J(\d(\p_t\omega)) + (\Delta_t + \Delta^*_t)^p\wedge H_{pmn}\, \d x^m\wedge\d x^n\\
&= \frac{\alpha'}{2}\, \left(\tr (\p_t A\wedge F) - \tr(\p_t\Theta \wedge R )\right) + \d\tilde\B_t~,\label{eq:varanomaly}
\end{split}
\end{equation}
where
\begin{equation}
\tilde\B_t = \p_tB - \frac{\alpha'}{4}\, \left(\tr (A\wedge \p_t A) - \tr(\Theta \wedge\p_t\Theta )\right)~,\label{eq:varB}
\end{equation}
and
\[ \Delta_t = (\p_t Z^A) \Delta_A~,\qquad \Delta^*_t = (\p_t \bar Z^{\bar A})\, \bar\Delta_{\bar A}~,\]
with $\bar\Delta_{\bar A}$ the complex conjugate of $\Delta_A$. Let
\begin{equation*}
{\cal D}_t = \tilde B_t + i \p_t\omega~.
\end{equation*}
Separating equation \eqref{eq:varanomaly} by type we find
\begin{equation}
\begin{split}
(0,3)\, {\rm part:}\qquad &\bp{\cal D}_t^{(0,2)} = 0\\
(1,2)\, {\rm part:}\qquad &\p{\cal D}_t^{(0,2)} +  \bp{\cal Z}_t^{(1,1)}  = 2{\H}(x_t)_a\wedge\d z^a\:,
\label{eq:eqsforZ}
\end{split}
\end{equation}
where we have further separated out a complex structure dependence of the $(1,1)$-part of the $B$-field, which appears when $B^{(2,0)}\neq0$,
\begin{equation*}
(\p_tB)^{(1,1)}=(\p_tB_{a\bar b})\,\d z^{a\bar b}+\Delta^a_t\wedge B_{ab}\,\d z^b\:.
\end{equation*}
$\bp B^{(0,2)}$ then gives the $\bp$-closed part of the $\hat H$-part of the map $\H$ defined in equation \eqref{eq:classH}.  Here we let
\begin{equation}
\label{eq:complexherm}
\Z_t^{(1,1)}=\p_t(B_{a\bar b})\,\d z^{a\bar b}- \frac{\alpha'}{4}\, \left(\tr (a\wedge \alpha_t) - \tr(\theta \wedge\kappa_t)\right)+i\p_t\omega^{(1,1)}=\B_t^{(1,1)}+i\p_t\omega^{(1,1)}\:,
\end{equation}
where $\kappa_t=\p_t\Theta^{(0,1)}$ and $\alpha_t=\p_tA^{(0,1)}$ are taken to be deformations of connections on the holomorphic part of the bundles.

The map now reads
\begin{align}
{\H}(x)_a =\hat H_{\bar c ab}\,\d z^{\bar c}\wedge \Delta{}^b+ \frac{\sqrt{\alpha'}}{2}\, \Big(\tr(f_a\wedge\tilde\alpha) +i\,\tr(r_a\wedge\tilde\kappa)\Big)\:,
\label{eq:Hdef}
\end{align}
for a generic $x=(\tilde\kappa,\tilde\alpha,\Delta)\in\Omega^{(0,q)}(E)$. Here $f_a=F_{a\bar b}\,\d z^{\bar b}$, $r_a=R_{a\bar b}\,\d z^{\bar b}$, and 
\begin{equation*}
\hat H = -i\p\omega-\hat H_0\:,
\end{equation*}
where $[\hat H_0]=[\bp B^{(2,0)}]\in H^{(2,1)}(X)\cong H^{(0,1)}(T^*X\otimes T^*X)$. In particular, for deformations of the holomorphic structures $(\kappa_t,\alpha_t,\Delta_t)$, we have
\begin{equation*}
x_t=(\tilde\kappa_t,\tilde\alpha_t,\Delta_t)=\Big(i\,\frac{\sqrt{\a}}{2}\kappa_t,\frac{\sqrt{\a}}{2}\alpha_t,\,\Delta_t\Big) \in H_{\bp_2}^{(0,1)}(E)\:,
\end{equation*} 
and the map reads
\begin{equation*}
{\H}(x_t)_a=\hat H_{\bar c ab}\,\d z^{\bar c}\wedge \Delta_t{}^b+ \frac{\alpha'}{4}\, \big(\tr(f_a\wedge\alpha_t) + \tr(r_a\wedge\kappa_t)\big)\:.
\end{equation*}
Note how the rescaling of the map $\F$ in $\bp_2$ also requires a rescaling of the corresponding elements $\alpha_t$ in $x_t$, as can be easily seen from \eqref{eq:atiyahF}. The same holds true for $\R$ and $\kappa_t$.  


Note how the condition \eqref{eq:condkerH} is naively more restrictive than equation \eqref{eq:eqsforZ}. This means that the representative of the class $\Z^{(0,2)}\in H^{(0,2)}(X)$  must be such that the $\bp$-closed form $\p\Z_t^{(0,2)}$ is $\bp$-exact if the deformed structure is to remain a holomorphic structure on Q. This subtlety regarding deformations of the anomaly cancellation condition versus deformations of the holomorphic structure $\bD$ deserves a bit more attention. First recall that $\bD$ is a holomorphic structure on $\Q$ if and only if the Bianchi identities hold. Deformations of $\bD$, which are the elements of $H^{(0,1)}_{\bD}(\Q)$, therefore includes deformations of the Bianchi identity.  These correspond to deformations of the anomaly cancellation modulo $\d$-exact terms. One might think that in our scheme the deformations of the anomaly cancellation condition are only defined modulo $\d$-closed terms. However, due to flux quantisation, which states that the harmonic part of $H$ is quantised, we find that closed infinitesimal deformations of the anomaly cancellation condition must be exact.

It follows that the elements of $H^{(0,1)}_{\bD}(\Q)$, i.e. deformations of the Strominger system, which of course includes the Bianchi identity, only define deformations of the anomaly cancellation modulo $\d$-exact terms. We can use this ambiguity to get rid of the $\p$-exact $(1,2)$-piece of the deformation of the anomaly cancellation condition. We might also get an extra $\bp$-exact piece which can be absorbed into $\bp\Z_t^{(1,1)}$. In this way the $\p$-exact piece is trivial from the point of view of deformations of $\bD$, and may equivalently be set to zero. Equations \eqref{eq:eqsforZ} then give a good interpretation of the elements in $H^{(0,1)}(T^*X)$ in the moduli space as the parameters for $\bp$-closed part of the complexified hermitian structure ${\cal Z}_t^{(1,1)}$ in equation \eqref{eq:complexherm}.\footnote{This is also obtained in\cite{Candelas2014} from the dimensional reduction of the 10 dimensional heterotic string theory.}



\subsection{An Interesting Subcase}
We would now like to discuss an interesting subcase where all deformations of the anomaly cancellation correspond to deformations of $\bD$, irrespective of adding $\d$-exact pieces. That is, deformations of the anomaly cancellation condition, and deformations of the Bianchi identity are equivalent. 

A mild assumption on the cohomology of $X$ would guarantee that this condition is satisfied. Suppose that
\begin{equation}
 H_{\bp}^{(0,1)}(X) = 0~.\label{eq:novectors}
 \end{equation}
This condition is very  interesting regarding deformations of the heterotic $SU(3)$-structure of the manifold $X$. 
It is not too hard to prove that this  is enough to guarantee that
\begin{equation*}
H_{\bp}^{(2,1)}(X)=H^{(2,1)}_{\d}(X)\:,
\end{equation*}
so that the allowed complex structure variations in this case are counted by the dimension of $H^{(0,1)}_{\bp}(TX)$, and not a subset of this (see section \ref{sec:varsJ} on deformations of the complex structure of $J$).  These matters are discussed further in\cite{DKS2014}.

Note that
\begin{equation*}
H^{(0,2)}_{\bp}(X)=H^{(3,1)}_{\bp}(X)=H^{(0,1)}_{\bp}(X)\:,
\end{equation*}
where the first equality follows from taking the Hodge-dual and complex conjugating. The second equality is clear by the existence of a nowhere vanishing holomorphic three-form~$\Omega$.

Therefore, by \eqref{eq:novectors}, 
\[ H_{\bp}^{(0,2)}(X)= 0~.\]

Returning now to the moduli space, this result means that ${\D}_t^{(0,2)}$ must be $\bp$-exact, by \eqref{eq:eqsforZ}. Hence, the requirement \eqref{eq:condkerH} and \eqref{eq:defBI} are equivalent and the {\it infinitesimal} moduli space is given by equation \eqref{eq:defMSS}.  The condition \eqref{eq:novectors} is therefore sufficient to ensure that {\it deformations of the anomaly cancelation condition} are equivalent to deformations of the holomorphic structure $\bD$ on $\cal Q$.

\subsection{The Quotient by $\textrm{Im}({\H}_0)$}
\label{sec:polystabH}
In the computation leading to  \eqref{eq:defBI}, we found that we need to take the quotient by 
\begin{equation*}
\textrm{Im}({\H}_0)\cong\{\tr({\H}_0({x}))\:\vert\:\ x\in H^0(E)\}\:.
\end{equation*}
Noting that (see equation \eqref{eq:Hdef})
\[ {\H}(x)_a\wedge\d z^a = \frac{\sqrt{\a}}{2}\, \tr(F\,\tilde\alpha)~,\quad \tilde\alpha\in H^0(\End(V))\]
we find that
\[ \textrm{Im}({\H}_0)\cong\{\tr(F\, \tilde\alpha)\:\vert\:\ \tilde\alpha\in H^0(\End(V))\}\subset H_{\bp}^{(1,1)}(X)\:.\]
which may be non-trivial whenever $H^0(\End(V))$ is non-trivial, that is, when the bundle $V=\oplus_iV_i$ is poly-stable with bundle factors $V_i$ for the which ${\rm End}(V_i)$ has non-vanishing traces. 
Let $V_i$ be such a stable bundle where $\End(V_i)$ has non-vanishing trace, and let
\begin{equation*}
\tilde\alpha_i\in H^0(\End(V_i))=\mathbb{C}\:,
\end{equation*}
where the $\mathbb{C}$ corresponds to the trace of the endomorphisms. These correspond to sections of $\End(V_i)$ by the Dolbeault Theorem. Without loss of generality, we may assume that this section takes the form $c_i I_i$, where $c_i$ is a constant, and $I_i$ is the identity isomorphism, which is part of the Lie-algebra for algebras of non-trivial trace. We may therefore assume that a generic section takes the form
\begin{equation}
\tilde\alpha=\sum_i c_i I_i\:,\label{eq:gensec}
\end{equation}
where the constants  $c_i$ are chosen so that $\tilde\alpha$ is traceless. It follows that the elements in $\textrm{Im}({\H}_0)$ are of the form\footnote{
Note that ${\rm Im}({\H}_0)=\{\sum_ic_i\tr F_i\}$ without any further constraints on the constants $c_i$. This is due to the fact that 
$\sum_ic_i\tr F_i = \sum_i(c_i+K)\tr F_i$  for any constant $K$, as  $\sum_i\tr F_i = 0$.}
\begin{equation*}
\label{eq:Dtermcond}
[h]=\sum_i c_i[\tr(F_i)]\:,
\end{equation*}
where the brackets refer to cohomology classes. 

Recall that the hermitian moduli are of the form
\begin{equation}
\label{eq:hermmod}
\Z_t^{(1,1)}=\B_t^{(1,1)}+i\p_t\omega^{(1,1)}\:.
\end{equation}
We claim that this constraint on the $\p_t\omega$-part of the moduli \eqref{eq:hermmod} is enforced by the Yang-Mills condition. In appendix \ref{app:Primitive} we show in Theorem \ref{tm:zero} that the Yang-Mills conditions pose no extra conditions on the moduli for {\it stable bundles}. If, on the other hand, the vector bundle is {\it poly-stable}, then these conditions may introduce constraints on the moduli. The constraint is exactly of the form above, and we take a moment to explain~why. 

Let $V_i$ be a stable bundle of nonzero trace. As $V=\oplus_i V_i$ is poly-stable of zero slope,
\[\mu(V_i) = \mu(V) = 0~, \]
we must have that the Yang-Mills condition for a bundle $V_i$ is,
\begin{equation*}
\omega\lrcorner F_i=0\:.
\end{equation*}
As noted before, it is only the trace part of the bundle that can impose non-trivial constraints from this condition. Taking the trace and using instead the Gauduchon metric $\hat\omega$ this condition becomes
\begin{equation}
\label{eq:YMi}
\hat\omega\hat\lrcorner\,\tr\:F_i=0\:.
\end{equation}
Varying equation \eqref{eq:YMi}, and performing a computation similar to that leading to equation \eqref{eq:varF11final}, we obtain that on a {\it conformally balanced manifold}
\begin{equation*}
\p_t\omega\,\tilde\lrcorner\,\tr F_i\:\in\:\textrm{Im}(\tilde\Delta_\p)+\textrm{Im}(\tilde\Delta_{\bp})\:,
\end{equation*}
where we recall that $\tilde g=e^{-2\phi}g$. Equivalently, this condition means that
\begin{equation*}
(\p_t\omega,\tr\:F_i) = 0\:,
\end{equation*}
where the integration is done with respect to $\tilde g$. Considering the Hodge decomposition of $\p_t\hat\omega$ with respect to the $\bp$ operator and $\tilde g$, it is easy to see that the $\bp^{\tilde\dagger}$-exact piece  drops out from the inner product.  Hence, only the $\bp$-closed part contributes, corresponding to in $H_{\bp}^{(1,1)}(X)$\footnote{By Proposition \ref{prop:CBcond}, the $\bp$-exact part is determined entirely by deformations of the complex structure and the dilaton.}.  These correspond to the hermitian moduli.  Then the vanishing of the inner product implies that we should mod out by forms proportional to $\tr\:F_i$ in the hermitian moduli, or more generally, by terms proportional to $\sum_ic_i\tr\:F_i$. 

Interestingly, by computing the first cohomology $H^1_{\bD}(\Q)$, which gives the tangent space $T\M$ of the moduli space of holomorphic structures on $\Q$ at $\bD$, we find that the Yang-Mills condition gets implemented for free. This is not surprising, as discussed in the next section, where we consider $T\M$  in more detail. As we will see, this is naturally included in the quotient by $\bD$-exact terms.

The constraint of modding out by $\Im(\H_0)$ is however given for the full hermitian modulus $\Z_t^{(1,1)}$, not just the $\p_t\omega^{(1,1)}$-part. The constraint on the $\B_t^{(1,1)}$-part of the moduli \eqref{eq:hermmod} is due to gauge transformations. Recall the definition of $\tilde\B_t$ through \eqref{eq:varB}. It follows from this, and the gauge transformations of the B-field \eqref{eq:gaugeB}, that trivial deformations $\tilde\B_t$, corresponding to gauge transformations, take the form
\begin{equation}
\label{eq:varBgauge}
\tilde\B_{t\,\textrm{gauge}}=-\frac{\a}{2}\big(\tr\,F\epsilon_t-\tr\,R\eta_t)+\frac{\a}{4}\d\big(\tr\,A\epsilon_t-\tr\,\Theta\eta_t\big)\:,
\end{equation}
where $\epsilon_t$ and $\eta_t$ correspond to infinitesimal gauge transformations of $A$ and $\Theta$ respectively. At this point, we are not interested in gauge transformations that change the holomorphic structures, corresponding to $\bp$-exact terms. That is, we set $\Delta_t=\bp\epsilon_t=\bp\eta_t=0$. It follows that $\eta_t=0$ by stability of $TX$. The same is true for $\epsilon_t$ if $V$ is stable. If $V=\oplus_iV_i$ is poly-stable, we may assume by \eqref{eq:gensec} that
\begin{equation*}
\epsilon_t=\sum_i c_i I_i\:,
\end{equation*}
and 
\begin{equation*}
\tilde\B_{t\,\textrm{gauge}}=\B^{(1,1)}_{t\,\textrm{gauge}}=-\frac{\a}{2}\sum_ic_i\tr\,F_i+\frac{\a}{4}\sum_ic_i\d\tr\,A_i=-\frac{\a}{4}\sum_ic_i\tr\,F_i\:,
\end{equation*}
where we have used that $\d\tr A_i=\tr F_i$ by symmetry of the trace. It follows that any term in $\Z_t^{(1,1)}$ which lies in $\Im(\H_0)$ should be considered trivial, and can thus be modded out.

\subsection{The Moduli Space of the Strominger System}

We now claim that the tangent space of the moduli space of the Strominger system, is given by $H_{\bD}^{(0,1)}(\Q)$ in equation \eqref{eq:defBI}. As we have seen, the extension bundle $\cal Q$, with extension class $\H$ in equations \eqref{eq:sesH} and \eqref{eq:classH}, with the holomorphic structure $\bD$ determined by the Bianchi identities, together with the requirement that $\bD$ is an instanton \eqref{eq:Dinstanton}, is equivalent to the Strominger system\footnote{We suppress the issue regarding the existence of a holomorphic three-form on $X$, discussed in section \ref{sec:varsJ}.}.


Consider the elements in the cohomology
\begin{equation}
\label{eq:defMSS}
H^{(0,1)}_{\bD}(\Q)\cong {\cal M}_{HS}\oplus\ker({\H}_1)\:,
\qquad  {\cal M}_{HS} = \Big[H^{(0,1)}(T^*X)\Big/\textrm{Im}({\H}_0)\Big]~,
\end{equation}
which we would like to interpret as the moduli of the Strominger system. The cohomology group  $H_{\bD}^{(0,1)}( \Q)$ is of course the tangent space to the moduli space of deformations of the holomorphic structure on
$\cal Q$ given by the differential operator $\bD$ in equation \eqref{eq:barD}. The key issue here is that by preserving the holomorphic structure on $\cal Q$ these moduli correspond to deformations which preserve the Bianchi identities.

We begin with the $\bD$-closed elements
\begin{equation}
 {\H}_1(x_t)_a = - \bp y_{t\, a}~,\qquad \bp_2 x_t = 0~,\label{eq:Dclosed}
\end{equation} 
for $x_t\in \Omega^{(0,1)}(E)$ and $y_t\in \Omega^{(0,1)}(T^*X)$. Clearly, the  left hand side of the first equation only involves $x_t\in H_{\bp_2}^{(0,1)}(E)$, that is, only involves variations of the holomorphic structure of $E$. Hence, the moduli in
\[ \ker(\H_1)\subseteq H_{\bp_2}^{(0,1)}(E)~,\]
represent those deformations of the holomorphic structure of $E$ which preserve the full holomorphic structure $\bD$ on $\Q$.
On the other hand, for a fixed holomorphic structure on $E$, that is for $x_t = 0$, we have that $\bp y_t = 0$ and so 
the moduli in 
\[{\cal M}_{HS}= \Big[H^{(0,1)}(T^*X)\Big/\textrm{Im}({\H}_0)\Big]~\]
correspond to the (complexified) hermitian moduli, that is
\begin{equation*}
y_{t\, a}\d z^a=\Z_t^{(1,1)}\:.
\end{equation*}

Consider now the $\bD$-exact forms. Let
\begin{equation*}
(y_t,x_t)
\in \Omega^{(0,1)}(\Q)~,\qquad
x_t = 
(\tilde\kappa_t,\tilde\alpha_t,\Delta_t)
\in \Omega^{(0,1)}(E)~,
\end{equation*}
and 
\begin{equation*}
(f_t,\xi_t)\in\Omega^0({\cal Q})~,\qquad
\xi_t  = (i\,\frac{\sqrt{\a}}{2}\eta_t,\frac{\sqrt{\a}}{2}\epsilon_t,\delta_t) \in \Omega^0(E)~.
\end{equation*}
The $\bD$-exact forms satisfy
\begin{equation}
\left( \begin{array}{c}
 y_t \\ x_t
\end{array}\right)
=
\left( \begin{array}{cc}
\bp f_t + {\H}_0(\xi_t) \\
\bp_2  \xi_t \end{array} \right).\label{eq:Dexact}
\end{equation}
The second equation in \eqref{eq:Dexact} are the trivial deformations of the holomorphic structure on $E$ corresponding to changes in $J$ due to diffeormophisms
\[\Delta_t = \bp \delta_t~,\]  
changes of the gauge fields from gauge transformations and trivial deformations of $J$
\[  \tilde\alpha_t = \frac{\sqrt{\a}}{2}\big(\bp_{A}\epsilon_t + {\F}(\delta_t)\big)~,\]
and a similar equation for the trivial deformations of the tangent bundle
\[\tilde\kappa_t = i\,\frac{\sqrt{\a}}{2}\big(\bp_\Theta\eta_t + \R(\delta_t)\big)~.\]

The first equation in \eqref{eq:Dexact} can be written as
\begin{equation}
(y_t)_a = \bp (f_t)_a + {\H}_0(\xi_t)_a = \bp (f_t)_a + \hat H_{\bar c ab}\d z^{\bar c}\delta^b+\frac{\a}{4}\Big( \tr(f_a\epsilon_t) - \tr(r_a\eta_t)\Big)~. \label{eq:trivialy}
\end{equation}
The first term in equation \eqref{eq:trivialy} enforces the triviality of $\bp$-exact terms, corresponding to trivial deformations of the holomorphic structure on $T^*X$, so that hermitian moduli take values in $H^{(0,1)}(T^*X)$. This is also related to the preservation of the conformally balanced condition. 
Indeed, that the $\bp$-exact part of $y_t$ is trivial enforces Proposition \ref{prop:CBcond}. In fact, Proposition \ref{prop:CBcond} means that, as long as $TX$ is stable, the preservation of the conformally balanced condition ($\d\p_t\hat\rho = 0$) determines the $\bp$-exact part of $\p_t\omega^{(1,1)}$.
We should hence mod out by $\bp$-exact terms in $\p_t\omega^{(1,1)}$, as these is already given by the preservation of the conformally balanced condition. Note also that the $(1,1)$-part of the gauge transformation of the B-field \eqref{eq:gaugeB0} also gives rise to $\bp$-exact $(1,1)$-forms, i.e. $\delta B^{(1,1)}_{gauge}=\bp\lambda^{(1,0)}+\p\lambda^{(0,1)}$, where the $\p$-exact piece corresponds to a trivial deformation of the $(1,0)$-type anti-holomorphic structure on $T^*X$. The remaining $\bp$-exact term corresponds to a trivial deformation of the holomorphic structure on $T^*X$.
We should therefore also mod out by $\bp$-exact terms when considering the $\B_t^{(1,1)}$-part of the hermitian moduli.


The last three terms come from trivial deformations of the holomorphic structure of $E$, that is $\bp_2\xi_t$. Keeping fixed the deformations of the holomorphic structure on $E$, that is, setting 
\[\bp_2\xi_t = 0~,\]
we see that the second and last term vanishes due to the stability of the tangent bundle $TX$.
The third term corresponds to the discussion in the previous section.  In fact, since a generic section of ${\rm End}(V)$ takes the form in equation \eqref{eq:gensec} we have that  this term is of the form
\[ \tr(\epsilon_t F) = \sum_i\,\tr( c_i F_i) = [h]~, \]
where $[h]$ represents a class in $H^{(1,1)}(X)$. In particular, as we argued in the previous section, this implements the Yang-Mills condition on the slope zero poly-stable bundle $V$.

\section{Discussion, and Future Directions}
\label{sec:discussSS}
We complete the chapter with a review and a discussion of the results. As we have seen, the Strominger system can be put in terms of a hermitian Yang-Mills operator $\bD$ on an extension bundle $\Q$ over $X$. The infinitesimal moduli space of the Strominger system is in term computed by the first cohomology 
\begin{equation*}
T\M=H^{(0,1)}(\Q)\:,
\end{equation*}
where $H^{(0,1)}(\Q)$ is given by a subset of the usual cohomologies $H^{(0,1)}(T^*X)$ (K\"ahler moduli), $H^{(0,1)}(TX)$ (complex structure moduli), $H^{(0,1)}(\End(V))$ (bundle moduli), and $H^{(0,1)}(\End(TX))$ which are ``moduli" related to deformations of $\nabla$. We stress that the elements of $H^{(0,1)}(\End(TX))$ do not correspond to physical fields. Rather, they are field redefinitions as we shall see in the next chapter. 

\subsection*{Relation to the Four-Dimensional Theory}
Now that we know what the infinitesimal moduli space of $N=1$ heterotic compactifications $T\M$ is, the next step is to relate it to a four-dimensional $N=1$ effective theory. Indeed, as discussed in the introduction, this is the part of the phenomenological goal of heterotic supergravity. Here, we make some brief remarks about this. In doing so, we also comment on future directions that are worth investigating. The story presented in this chapter is only a first step in figuring out what the correct low energy theory is, and there are plenty of unresolved issues to be worked out. Putting the system in terms of a holomorphic structure $\bD$ can hopefully serve to help with this.


When compactifying the heterotic string to maximally symmetric four-dimensional spacetime, the corresponding four-dimensional supergravity is an $N=1$ supergravity with K\"ahler potential $K(\X,\bar\X)$, and superpotential $W(\X)$. Here the low energy fields $\X$ correspond to the moduli-fields that we found in the previous section. A priori, the moduli fields $\X$ take the form
\begin{equation*}
\X=(y,x)\in H^{(0,1)}(\Q)\:,
\end{equation*}
where
\begin{equation*}
y\in H^{(0,1)}(T^*X)\Big/\Im(\H_0)\:,\;\;\;\;x=(\kappa,\alpha,\Delta)\in H^{(0,1)}(E)\:.
\end{equation*}

Moreover, the Bianchi Identities impose the following conditions
\begin{equation*}
\Delta\in\ker(\F)\cap\ker(\R)\:,\;\;\;\; x\in\ker(\H)\:.
\end{equation*}
These equations have interesting implications in terms of moduli stabilisation. Indeed, we expect them to be related to F-term conditions in the low energy theory. The superpotential of the low energy theory is conjectured to have the following Gukov-Vafa-Witten-type formula \cite{Gukov:1999ya, CyrilThesis}:
\begin{equation}
\label{eq:suppot}
W=\int_X\Big(H+i\d\omega\Big)\wedge\Omega=\int_X\Big(H_0+\frac{\a}{4}(\omega_{CS}^A-\omega_{CS}^\nabla)+i\d\omega\Big)\wedge\Omega\:,
\end{equation}
where $H_0=\d B$. This superpotential was used in \cite{Anderson:2010mh} to relate the conditions on the simultaneous deformations of the bundle $X$ and the complex structure on $X$ to F-terms in the lower energy four-dimensional theory. That is, for $\Delta\in H^1(TX)$, we should have $\Delta\in\ker(\F)$. This analysis generalises easily for the tangent bundle connection, and one would get F-terms for complex structure moduli for which $\Delta\;\slashed\in\;\ker(\R)$. In a similar fashion, it would be interesting to see if the superpotential \eqref{eq:suppot} can be used to generate F-terms for the moduli-fields $x$, and if these can be related to the requirement that $x\in\ker(\H)$. We leave this for future work. 

Recall also the condition we found for the hermitian moduli $\Z_i$, that we should mod out by forms of type $\tr\:F_i$ in the case of poly-stable bundles $V=\oplus_iV_i$. We can relate this to the Yang-Mills conditions, which can be shown to correspond to D-term conditions in the low energy theory, at least in the Calabi-Yau case \cite{Anderson:2009sw, Anderson:2011cza}. We expect this to be true in the non-K\"ahler case as well. In this way, one may relate the cohomology of the allowed infinitesimal deformations
\begin{equation*}
\X\in H^{(0,1)}(\Q)
\end{equation*}
to the low energy theory.

Furthermore, we would like to consider obstructions to higher order deformations of $\bD$. From a deformation theory point of view, these should be valued in $H^{(0,2)}_{\bD}(\Q)$. Returning to the case of a bundle $V\rightarrow X$, it turns out that the obstructions to the deformations of the bundle also gives rise to a non-vanishing superpotential in the low energy theory \cite{Berglund:1995yu, Anderson:2011ty}. It would be interesting to see if the superpotential \eqref{eq:suppot} generate similar terms when we calculate the obstructions for the generic deformations $\X$.

Finally, in order to write down the full low energy 4d supergravity, we also need to know what the K\"ahler potential $K$ is. To find this, one needs to do the dimensional reduction of the ten-dimensional theory down to four-dimensions, and read off the corresponding kinetic terms \cite{McOrist2014, Candelas2014}. This is well beyond the scope of this thesis. However, the fact that we have put the Strominger system in terms of a holomorphic structure $\bD$ on $\Q$ might give us a hint of what this metric could look like. Indeed, investigations into the existence of K\"ahler metrics on the moduli space of holomorphic bundles over complex manifolds has been carried out in the mathematical literature before. Particularly in the case when the compact manifold $X$ is K\"ahler \cite{kobayashi1987differential, kim1987moduli, schumacher1993moduli, huybrechts2010geometry}. 
Hopefully, the correct K\"ahler metric for the 4d theory will come from some suitable generalisation of these structures. What exactly this generalisation is remains to be seen.

\subsection*{The Yang-Mills Connection $\D$}
As we saw in section \ref{sec:HYM}, and appendix \ref{app:HYMconfbal}, we can without loss of generality assume that $\bD$ is an instanton. Indeed, this implements the instanton conditions for the connections $A$ and $\nabla$. Curiously then, it seems then that the heterotic supergravity, including anomalies, can perhaps be written as a Yang-Mills theory for some connection $\D$ on the extension bundle $\Q$. Indeed, the supersymmetry condition for such a connection would read
\begin{equation*}
F_{\D\,AB}\Gamma^{AB}\epsilon=0\:,
\end{equation*}
where $A,B,..$ denote flat indices, so that $\Gamma^M={e_A}^M\Gamma^A$, where $\{{e_A}^M\}$ is a ten-dimensional vielbein frame. The six-dimensional compact supersymmetric solutions of such a theory would then precisely correspond to the Strominger system, at least in the large volume limit. It also has the potential of describing the vacuum {\it outside the large volume limit}, as the instanton condition  can be defined without invoking $\a$. Indeed, holomorphic structures have no inherent size associated with them. Note however that {\it curvatures of the usual dynamical fields} are involved in the construction of $\D$. There is therefore a question as to what exactly the dynamical field of such a theory should be.

Taking such a viewpoint might however have interesting consequences, both from a four-dimensional point of view, but also from a world-sheet point of view. Indeed, due to lack of supersymmetry, the heterotic world-sheet theory has far more flexibility then its type II cousins, making it harder to compute quantum corrections, i.e. $g_s$-corrections to these theories. Putting the structure in terms of a single Yang-Mills connection might help in this direction. Indeed, world-sheet models for Yang-Mills connections have been studied in the literature before, see e.g. \cite{Park:1993ga, Park:1993fy, Hofman:1999dt, Hofman:2000yx}. It would be interesting to see if the present case can be put in this framework.


There is also a question remaining as to how general the story we have presented here is. Indeed, as we have seen, the Strominger system corresponds to a holomorphic structure $\bD$ on $\Q$, which is compatible with $\Q$ as an the extension bundle. This requires $\bD$ to take the upper block-diagonal form \eqref{eq:barDmatrix}. Clearly, not all connections on $\Q$ are of this form. Any matrix can however locally be put in this form, upon some local $\Q$-valued transformation. Note however that there might be global obstructions to doing so.
These transormations are furthermore expected to be related to a heterotic form of T-duality \cite{2013arXiv1308.5159B, Bedoya:2014pma, Coimbra:2014qaa}, and we intend to investigate this more in the future \cite{delaOssaSoon}.

\begin{subappendices}

\section{Instanton Connections on Extensions over Conformally Balanced Manifolds}
\label{app:HYMconfbal}
In this appendix we discuss the hermitian Yang-Mills conditions on connections on extended bundles. It turns out that on conformally balanced manifolds, this comes down to a choice of representative of the extension class, and hence poses no extra constraint on the geometry. We use this in section \ref{sec:HYM} in relation to the Yang-Mills conditions.

Let $X$ be a complex space with hermitian two-form $\widehat\omega$ which is balanced, that is
\begin{equation*}
\d(\frac{1}{2}\hat\omega\wedge\hat\omega)=\hat\omega\wedge\d\hat\omega=0\:.
\end{equation*}
The corresponding metric $\hat g$ is Gauduchon. Let $E_1$ and $E_2$ be vector bundles over a manifold $X$, with connections $\nabla_1=\d+A_1$, and $\nabla_2=\d+A_2$. Consider the split exact extension sequence
\begin{equation}
\label{eq:genericExt}
0\rightarrow E_1\rightarrow \Q\rightarrow E_2\rightarrow 0\:,
\end{equation}
where we define a connection on $\Q$ as
\begin{equation}
\D\;=\; \left[ \begin{array}{cc}
\nabla_2 & \,{\X} \\
0  & \,\nabla_1  \end{array} \right]\:.
\label{eq:D}
\end{equation}
This is the most generic form a connection on the extension bundle $\Q$ can take\footnote{It should again be noted that any connection can be written in this way, upon a $\Q$-valued coordinate redefinition.}. Note that when $E_2=TX$ and $E_1$ is some vector bundle with a Lie-bracket, e.g. $E_1=\End(W)$ for some vector bundle $W$, the sequence \eqref{eq:genericExt} is the Atiyah sequence of $\Q$ as a Lie algebroid.

Next, we decompose
\begin{equation*}
\D=\bD+D\:,
\end{equation*}
in terms of it's $(1,0)$-part $D$, and $(0,1)$-part $\bD$ with respect to the complex structure $J$ on $X$. We also decompose
\begin{align*}
\X&=\bar\alpha+\alpha\\
\nabla_1&=\p_1+\bp_1\\
\nabla_2&=\p_2+\bp_2\:,
\end{align*}
where $\alpha$ is type $(0,1)$. Let us then compute, for $(x,y)\in\Omega^{(0,q)}(\Q)$
\begin{equation*}
\bD^2\left[ \begin{array}{c} x \\ y\end{array} \right]
\;=\;F_{\D}^{(0,2)}\left[ \begin{array}{c} x \\ y\end{array} \right]\;=\; 
\left[ \begin{array}{c} \bp_2^2x +\bp_2\alpha(y)+\alpha(\bp_1y) \\ \bp_1^2y\end{array} \right]
\end{equation*}
where $F_{\D}$ is the curvature of $\D$. We see that in order for $\D$ to be holomorphic, we need
\begin{align}
\label{eq:holbunds}
\bp_1^2=\bp_2^2&=0\\
\label{eq:classa}
\bp_2\alpha+\alpha\bp_1 &=0\:.
\end{align}
The first equation implies that $\nabla_1$ and $\nabla_2$ are holomorphic connections, so we may set $\bp_1=\bp_2=\bp$ by a choice of holomorphic trivialization.
It then follows from equation \eqref{eq:classa} that
\begin{equation*}
\bp\alpha=0\:,
\end{equation*}
or
\begin{equation*}
[\alpha]\in\textrm{Ext}^1(E_1,E_2)=H^{(0,1)}(E_2\otimes E_1^*)\:.
\end{equation*}
Note that we also have
\begin{align}
\label{eq:holbunds}
\p_1^2=\p_2^2&=0\\
\label{eq:classa}
\p_V\bar\alpha&=0\:,
\end{align}
by requiring $F_\D$ to be type $(1,1)$. 

We also want to check the Yang-Mills conditions. That is, we will impose that the connections on both bundles $E_i$ are instantons. We will see that this is indeed equivalent to choosing $\bD$ an instanton, provided we choose the correct extension class $\alpha$. We then compute
\begin{equation}
\label{eq:ExtCurv}
\D^2\left[ \begin{array}{c} x \\ y\end{array} \right]\;=\; 
\left[ \begin{array}{c}
F_{\nabla_2}(x) + (\bp_{V}\bar\alpha)(y)+(\p_{V}\alpha)(y) \\ F_{\nabla_1}(y)  \end{array} \right]\:,
\end{equation}
where we have set $V=E_2\otimes E_1^*$ for ease of notation. It is clear that we need both $\nabla_1$ and $\nabla_2$ to be instantons. It follows from the Li-Yau theorem, Theorem \ref{tm:LiYau}, that both $E_i$ are poly-stable of slope zero.


We additionally assume that we have hermitian structures $h_i$ on the bundles $E_i$ with respect to which $\nabla_i$ are hermitian. We then have well-defined inner-products given by
\begin{equation*}
(\beta_1,\beta_2)_i=\int_X\beta_1^j\wedge\hat*\bar\beta_2^{\bar k}h_{i,j\bar k}\:,
\end{equation*}
for $\beta_m\in\Omega^{(0,q)}(E_i)$. Here the Hodge-dual is with respect to $\hat\omega$. These in turn induce inner products on $\Omega^{(0,q)}(V)$. We can use this to define a Laplacian
\begin{equation}
\label{eq:lapl}
\hat\Delta_V=\bp^{\hat\dagger}_V\bp_V+\bp_V\bp_V^{\hat\dagger}\:,
\end{equation}
where the adjoint is defined with respect to the inner-product. That is
\begin{equation*}
\bp_V^{\hat\dagger}=-\hat*\p_V\hat*\:,
\end{equation*}
The Laplacian \eqref{eq:lapl} is elliptic with a finite-dimensional kernel, and it induces the usual Hodge-decomposition of $V$-valued forms
\begin{equation*}
\Omega^{(0,q)}(V)=\ker(\hat\Delta_V)\oplus\Im(\bp_V)\oplus\Im(\bp_V^{\hat\dagger})\:,
\end{equation*}
where the forms in $\ker(\hat\Delta_V)$ are harmonic with respect to the Gauduchon metric $\hat g$. Note that by choosing the harmonic representative for $\alpha$, i.e.
\begin{equation*}
0=\bp_{V}^{\hat\dagger}\alpha=i\hat*\p_{V}(\widehat\omega^2\wedge\alpha)=i\hat*(\widehat\omega^2\wedge\p_{V}\alpha)\:,
\end{equation*}
and similarly choosing $\bar\alpha$ $\p_V$-harmonic, we see from \eqref{eq:ExtCurv} that we get
\begin{equation*}
\widehat\omega\wedge\widehat\omega\wedge F_\D=0\:.
\end{equation*}
I.e. the connection $\D$ is an instanton.

\section{Stability and Variations of the Primitivity Conditions for the Curvatures.}
\label{app:Primitive}

In this appendix, we discuss variations of the primitivity conditions for a curvature of a {\it stable bundle} $V$
\[ \omega\lrcorner F = 0 ~.\]
This condition should be preserved under a general deformation, in particular under the deformations of the bundle $E$, but also including deformations of the hermitian parameters. 
We show that on a conformally balanced manifold, a general variation of the primitivity condition of the curvature does not pose any constraints on the first order moduli space whenever the bundle is stable. 

Under a general variation the Yang-Mills equation becomes
\[0 = \p_t(\omega\lrcorner F)
= \half\,\p_t\left(\omega^{mn}F_{mn}\right)
= \half\, \left((\p_t\omega^{mn})  F_{mn} + \omega^{mn}\p_t F_{mn}\right)
~, \]
and therefore
\begin{equation}
 \omega\lrcorner \p_t F = -  \half\, \p_t(\omega^{mn})  F_{mn} =   (h_t^{(1,1)})\lrcorner F ~.
  \label{eq:varinstone}
 \end{equation}
 This equation means that $F$ acquires a non-primitive part under a general deformation
 \[(\partial_t F)^{(1,1)} = \frac{1}{3}\, \big((h_t^{(1,1)})\lrcorner F\big)\, \omega + f_t~,\]
 where $f_t$ is a primitive $(1,1)$-form, $\omega\lrcorner f_t = 0$. 
 Note that this non-primitive part of $\p_t F$ depends on the variations of the hermitian form
 and it is needed so that $F_t$ is primitive with respect to $\omega_t$.
 
On the other hand, considering instead a general variation of $F$ using equation \eqref{eq:curvV}. We find
\begin{equation}
(\partial_t F)^{(1,1)} =  \bp_A b_t + \p_A \alpha_t~,\label{eq:varinsttwo}
\end{equation}
where $\p_A=\p+a$ is the $(1,0)$-part of $\d_A$, $\alpha_t$ is the $(0,1)$-part of $\p_t A$ as before, and $b_t$ is the $(1,0)$-part of $\p_t A$. We are interested in the infinitesimal deformations of the holomorphic bundle counted by $\alpha_t$, and we therefore set $b_t=0$. $\p_A$ is given by, for $\beta\in \Omega^*(\End(V))$,
\begin{equation} 
\partial_A\beta = \p\beta + [a, \beta] ~,
\label{eq:pAdag}
\end{equation}
where the brackets act as commutators or anti-commutators, depending on if $\beta$ is even or odd respectively.
Again, it is easy to prove that this operator also squares to zero, $\p_A^2 = 0$ if and only if $F^{(2,0)} = 0$. 

Putting together equations \eqref{eq:varinstone} and \eqref{eq:varinsttwo} we obtain a relation
\begin{equation}
\label{eq:varF11}
 (h_t^{(1,1)})\lrcorner F = \omega\lrcorner\p_A \alpha_t~,
\end{equation}
which seems to represent a constraint on the moduli space of hermitian structures $h_t$.  However for stable bundles this is not the case.  

\begin{Theorem}\label{tm:zero}
On a {\it conformally balanced manifold}, with a stable holomorphic vector bundle $V$, such that the endomorphisms of $V$ are traceless, there are no gauge bundle parameters on the right hand side of equation \eqref{eq:varF11}, and there is always a gauge transformation so that \eqref{eq:varF11} is satisfied for any variation $h_t$ of the hermitian structure $\omega$.
\end{Theorem}

\begin{proof}
Let $\hat g_{mn} = e^{- \phi}\, g_{mn}$ be the Gauduchon metric and $\widehat\omega = e^{-\phi}\, \omega$ be the corresponding Gauduchon hermitian form.  Let $\widehat\lrcorner$ and $\hat *$ be the contraction operator and the Hodge dual operator with respect to the Gauduchon metric respectively.  Then
\begin{equation*}
(h_t^{(1,1)})\lrcorner F 
= *\left(\p_A\alpha_t\wedge *\omega\right)
= e^{2\phi}\, *\left(\p_A\alpha_t \wedge\hat\rho \right)~,
\end{equation*}
where 
\begin{equation*}
\hat\rho = e^{-2\phi}\, \rho~,\qquad\rho = *\omega= \half\, \omega\wedge\omega~.
\end{equation*}
Because on a conformally balanced manifold $\d\hat\rho = 0$, we have
\begin{align*}
(h_t^{(1,1)})\lrcorner F 
&= e^{2\phi}\,*\left(\bp_A(\p_A (\alpha_t \wedge \hat\rho)\right)
= e^{2\phi}\,* \left(\p_A\hat * (J(\alpha_t ))\right)\\
&= i\, e^{-\phi}\,\hat* \left(\p_A\, \hat * \, \alpha_t \right)~,
\end{align*}
where we have used the fact that $a_t$ is a $(0,1)$-form.  We have also used
\begin{equation}
\label{eq:confdual}
\hat *\beta = e^{(p-3)\phi}\, *\beta\:,
\end{equation}
which is true for any $p$-form $\beta$ in six-dimensions. We now note that the operator on the right hand side in the last equality is the adjoint, with respect to $\hat g$, of the differential operators $\bp_A$ given by
\begin{align*}
\bp_A^{\hat\dagger} &= - \hat*\p_A\hat*\:.
\end{align*}
Using this we now have
\begin{equation*}
(h_t^{(1,1)})\lrcorner F = 
i\, e^{-\phi}\,\bp_A^{\widehat\dagger}\,\alpha_t~,
\end{equation*}
where $\widehat\dagger$ means the adjoint of the operators taken with respect to the Gauduchon metric. Consider now the Hodge decomposition of $\alpha_t$ \[ \alpha_t = \bp_A\epsilon_t + \bp_A^{\hat\dagger}\, \eta_t + \alpha_t^{har}~,\]
where $\epsilon\in \Omega^0(\End(V))$, $\eta\in \Omega^{(0,2)}(\End V)$ and $\alpha_t^{har}$ is the $\bp_A$-harmonic part of $\alpha_t$ (using the Gauduchon metric).  
Then
\begin{align*}
\bp_A^{\widehat\dagger}\, \alpha_t  &=\bp_A^{\widehat\dagger}\bp_A\epsilon_t~,\\
\end{align*}
Then, equation \eqref{eq:varF11} becomes
\begin{equation}
i\, e^{\phi}\, (h_t^{(1,1)})\lrcorner F = 
-\, \bp_A^{\widehat\dagger}\bp_A\epsilon_t\:.
~.\label{eq:varF11prefinal}
\end{equation}
Any variation $\alpha_t$ of $A$ corresponding to a gauge transformation, and which is therefore trivial,
is of the form
\begin{equation}
\label{eq:complexgauge}
\alpha_t = \bp_A\epsilon_t~,
\end{equation}
for some $\epsilon_t\in\Omega^0(\End(V))$. 
Consider the Laplacian
\[\Delta_{\bp_A} = \bp_A^\dagger\, \bp_A
+ \bp_A\, \bp_A^\dagger~,
\]
and let $\widehat\Delta_{\bp_A}$ be the corresponding Laplacian with respect to the Gauduchon metric.
Then we can write equation \eqref{eq:varF11prefinal} as
\begin{equation}
e^{\phi}\, (h_t^{(1,1)})\lrcorner F = \hat h_t^{(1,1)}\,\widehat\lrcorner F = 
i\, \widehat\Delta_{\bp_A}\epsilon_t\:,
\label{eq:varF11final}
\end{equation}
where $\hat h_t = e^{-\phi}\, h_t$ and $\widehat\lrcorner$ is the contraction operator with respect to the Gauduchon metric.

This equation means that $\hat h_t^{(1,1)}\,\widehat\lrcorner F$, which belongs to the space $ \Omega^0(\End(V))$, is in the image of this Laplacian which is a self-adjoint elliptic operator.  Therefore, whenever the kernel of this Laplacian is trivial, the image of the Laplacian spans all of the space $ \Omega^0(\End(V))$ and equation \eqref{eq:varF11final} always has a solution for any  $\hat h_t^{(1,1)}$.  This is precisely the case for a {\it stable} bundle $V$ because $H^0(\End V) = 0$ for traceless endomorphisms. We conclude that, for stable vector bundles $V$ with traceless endomorphisms, equation \eqref{eq:varF11final} poses no constraints on the deformations of the hermitian moduli. 

\end{proof}

We see then that variations of the Yang-Mills equation imposes no constraints on the variations $h_t$ of the hermitian form $\omega$, 
provided the bundles are {\it stable} with traceless endomorphisms. It should be noted however that infinitesimal deformations may be obstructed at higher orders, and that stability or the Yang-Mills condition may be spoiled.



\end{subappendices}





\chapter{Connection Redundancy, and Higher Orders in $\a$}
\label{ch:connredef}

We will now shed more light on the ambiguity which appeared in the previous chapter, concerning moduli related to deformations of $TX$ as a holomorphic bundle defined by the holomorphic connection $\nabla$. We argued that these moduli could not be physical, but where needed for the mathematical structure presented in the last chapter to work out. This chapter concerns the physical meaning of the ``moduli" related to $\nabla$. We also discuss higher orders in $\a$-corrections to heterotic supergravity, and we argue that the structure we found at $\OO(\a)$ persists, at least to second order in $\a$, provided one chooses this connection appropriately. Heterotic supergravity has been considered at higher orders in $\a$ before \cite{Witten:1986kg, Bergshoeff:1988nn, Bergshoeff1989439, Gillard:2003jh, Anguelova:2010ed}.
Ambiguities concerning the connection $\nabla$ have also been discussed extensively in the literature before, both from a world-sheet perspective  \cite{Hull1986187, Sen1986289, 0264-9381-4-6-027, Melnikov:2012cv, Melnikov:2012nm}, where a change of this connection can be shown to correspond to a field redefinition, or from the supergravity point of view  \cite{Hull198651, Hull1986357, Becker:2009df, Melnikov:2014ywa}, where a change of connection choice can be shown to correspond to a change of regularisation scheme in the effective action. We will review and extend some of these results. This chapter is based on \cite{delaossa2014}.


\section{Introduction}
We begin in section \ref{sec:FirstOrder1} with a discussion of the connection choice $\nabla$ on the tangent bundle $TX$ needed for the supersymmetry equations and equations of motion to be compatible. This leads to the an instanton condition on $\nabla$ \cite{Ivanov:2009rh, Martelli:2010jx, delaOssa:2014cia}
\begin{equation*}
R_{mn}\Gamma^{mn}\eta=0\:,
\end{equation*}
where $R_{mn}$ is the curvature two-form of $\nabla$ and $\eta$ is the spinor parametrising supersymmetry on $X$. This condition has an associated moduli space of infinitesimal deformations 
\begin{equation}
\label{eq:Mnabla}
T\M_\nabla\cong H^{(0,1)}_{\bp_\Theta}(\End(TX))\:,
\end{equation}
which we considered in the last chapter. These moduli cannot be physical, and the main purpose of this chapter is to understand their appearance. 

It is known that in order to have a consistent supergravity at $\OO(\a)$, the connection must be the Hull connection \cite{Hull1986357, Bergshoeff1989439}. A deformation of this connection is equivalent to a field redefinition or a change of the regularisation scheme in the effective action \cite{Sen1986289, Hull198651}. We are interested in the space of allowed deformations, for which there are supersymmetric solutions to the supergravity equations of motion. We find that even though we need to deform the supersymmetry transformations accordingly, as was also pointed out in \cite{Melnikov:2014ywa}, the conditions for preservation of supersymmetry may be assumed to be the same. Moreover, the space of connections which allow for such supersymmetric solutions to exist is again given by \eqref{eq:Mnabla}. 

In section \ref{sec:higher} we discuss extensions of these results to second order in $\a$. We find that the choice of the Hull connection, which was required at at $\OO(\a)$, should in general be corrected at higher orders. Indeed, as we shall see, insisting on the Hull connection can put additional constraints on the higher order geometry. This was also noted in \cite{Ivanov:2009rh}, where the first order geometry was taken as exact, resulting in a Calabi-Yau geometry. 

At $\OO(\a^2)$, we can assume that compact supersymmetric solutions satisfy the Strominger system, provided the connection $\nabla$ satisfies the instanton condition. This condition looks surprisingly like a supersymmetry condition corresponding to the connection $\nabla$ as if it was a dynamical field. Indeed, it was precisely this fact that $(\nabla,\psi_{IJ})$, where $\psi_{IJ}$ is the supercovariant curvature, transforms as an $SO(9,1)$-Yang-Mills multiplet at $\OO(\a)$ which lead to the construction of the $\OO(\a)$-action in the first place \cite{Bergshoeff1989439}. As also noted in \cite{Bergshoeff1989439}, this is symmetric with the gauge sector of the theory, and it is natural to assume this symmetry to higher orders in $\a$. This also prompts us to make a conjecture for what the connection choice should be at higher orders in $\a$.

We review the results of the chapter and discuss future directions in the discussion, section \ref{sec:discConn}. We have left some technical details to the appendices.

\section{First Order Heterotic Supergravity}
\label{sec:FirstOrder1}

Let us begin by recalling the bosonic part of the action at this order \cite{Bergshoeff1989439}\footnote{Although this action is valid to second order in $\a$, we only need it to first order in this section.}
\begin{align}
\label{eq:actionRedef}
S=&\int_{M_{10}}e^{-2\phi}\Big[*\mathcal{R}-4\vert\d\phi\vert^2+\frac{1}{2}\vert H\vert^2+\frac{\alpha'}{4}(\tr\vert F\vert^2-\tr\vert R\vert^2)\Big]+\OO(\a^3),
\end{align}
where $R$ is the curvature of the connection $\nabla$ on $TX$. The choice of connection $\nabla$ on $TX$ is a subtle question. Firstly it cannot be a dynamical field, as there are no modes in the corresponding string theory corresponding to this. Hence, $\nabla$ must depend on the other fields of the theory in some particular way. This dependence is forced upon us once the supergravity action and supersymmetry transformations are specified.




Let us also recall the supersymmetry variations \cite{Strominger1986, Bergshoeff1989439},
\begin{align}
\label{eq:O1spinorsusyConn1}
\delta\psi_M &=\nabla^+_M\epsilon=\Big(\nabla^{\hbox{\tiny{LC}}}_M+\frac{1}{8}\H_M\Big)\epsilon+\OO(\a^2)\\
\label{eq:O1spinorsusyConn2}
\delta\lambda&=\Big(\slashed\nabla^{\hbox{\tiny{LC}}}\phi+\frac{1}{12}\H\Big)\epsilon+\OO(\a^2)\\
\label{eq:O1spinorsusyConn3}
\delta\chi&=-\frac{1}{2}F_{MN}\Gamma^{MN}\epsilon+\OO(\a),
\end{align}

If we want the supergravity action to be invariant under the supersymmetry transformations \eqref{eq:O1spinorsusyConn1}-\eqref{eq:O1spinorsusyConn3} at $\OO(\a)$, we need a particular choice of connection in the action, namely the Hull connection $\nabla^-$, given by \eqref{eq:Hullconn}. This connection is needed in order that $(\nabla^-,\psi_{IJ})$ transforms as an $SO(9,1)$ Yang-Mills multiplet to the given order, as explained in \cite{Bergshoeff1989439}. Here $\psi_{IJ}$ is the supercovariant curvature given by
\begin{equation}
\label{eq:scurv}
\psi_{IJ}=\nabla^+_I\psi_J-\nabla^+_J\psi_I\:.
\end{equation}
With this, the full first order heterotic action is invariant under the supersymmetry transformations \eqref{eq:O1spinorsusyConn1}-\eqref{eq:O1spinorsusyConn3}, together with the corresponding bosonic transformations which we omit to write down for brevity.

\subsection{Instanton Condition}
\label{sec:instcond}
Let us now proceed to compactify the theory to four-dimensional Minkowski space, following the last chapter. The reader is referred to Section \ref{sec:hetcomp} for details. As we shall see in section \ref{sec:supsol1}, Proposition \ref{prop:Strom}, we may assume that a supersymmetric solution is a solution of the Strominger system, even when the connection on $TX$ is not the Hull connection.
Furthermore, supersymmetry should be compatible with the bosonic equations of motion derived from \eqref{eq:actionRedef}. This leads to a condition on $\nabla$ known as the instanton condition \cite{Ivanov:2009rh, Martelli:2010jx}, which we also give a proof of in appendix \ref{app:proof}. It should be noted that for supersymmetric solutions, as we also show in Appendix \ref{app:Hull}, the Hull connection does satisfy the instanton condition to $\OO(\a)$ \cite{Martelli:2010jx}.

Requiering $\nabla$ to be an instanton implies that it satisfies the conditions
\begin{equation}
\label{eq:inst2}
R\wedge\Omega=0,\;\;\;\;R\wedge\omega\wedge\omega=0\:,
\end{equation}
which are similar to those for the field-strength $F$. 
The first condition in \eqref{eq:inst2} implies that $R^{(0,2)}=0$. Therefore there is a holomorphic structure $\bp_\Theta$ on $TX$. We denote $TX$ with this holomorphic structure as $(TX,\nabla)$. 


The second condition of \eqref{eq:inst2} says that the connection $\nabla$ is Yang-Mills, more precisely, $\nabla$ is an instanton. By Theorem \ref{tm:LiYau}, such a connection exists if and only if the holomorphic bundle $(TX,\nabla)$ is poly-stable. Moreover, the connection is the unique hermitian connection with respect to the corresponding hermitian structure on $TX$.

As shown in appendix \ref{app:Primitive}, the Yang-Mills condition is stable under infinitesimal deformations. The infinitesimal deformation space for this connection is therefore
\begin{equation}
\label{eq:modulitheta}
T\M_{\nabla}=H^{(0,1)}_{\bp_\Theta}(X,\End(TX))\:.
\end{equation}
More explicitly, in appendix \ref{app:Primitive} we showed that for each $[\kappa]\in T\M_{\nabla}$,\footnote{Here $[\kappa]$ denote equivalence classes of deformations modulo gauge transformations.} where $\kappa=\delta\Theta^{(0,1)}$, there is a corresponding element $\kappa\in[\kappa]$ so that the Yang-Mills condition is satisfied. Starting from the instanton connection, there is then an infinitesimal moduli space $T\M_\Theta$ of connections for which the equations of motion are satisfied.

As mentioned, for the supergravity action to be invariant under the supersymmetry transformations \eqref{eq:O1spinorsusy1}-\eqref{eq:O1spinorsusy3}, the choice of connection is reduced further. In particular, invariance of the first order action forces the connection to be the Hull connection $\nabla^-$ \cite{Bergshoeff1989439}. Under these supersymmetry transformations, we therefore cannot choose any element in $T\M_{\nabla}$ when deforming the Strominger system. Rather we have to choose the element corresponding to a deformation of the Hull connection. 



\subsection{Changing the Connection}
We could ask what happens if we deform the connection in the action? Firstly, such deformations do not correspond to physical fields. Indeed, we shall see that they are equivalent to field redefinitions \cite{Sen1986289}. 
Secondly, insisting upon changing this connection means that we need to change the supersymmetry transformations correspondingly. It will however turn out that the conditions for supersymmetric solutions can be assumed to remain the same. Moreover, the condition that the new connection allows for supersymmetric solutions to the theory forces the new connection precisely to satisfy the instanton condition.

Let us discuss what happens when we change the connection $\nabla$ used in the action. That is, we let
\begin{equation}
\label{eq:connchange}
\nabla=\nabla^-+t\theta\:,
\end{equation}
where $\theta=\theta(\Phi)$ is a function of all the other fields of the theory, which we collectively have denoted by $\Phi$, and $t$ is an initesimal parameter. In the next section, we will take $t=\OO(\a)$, but for now we just assume it corresponds to an infinitesimal deformation of the connection. We are interested in what happens to the theory under such a small deformation.

Under supersymmetry, the new connection one-forms ${\Theta_I}^{JK}$ together with the supercovariant curvature $\psi_{IJ}$ transform as
\begin{align*}
\delta{\Theta_I}^{JK}=&{(\delta\Theta^-+t\delta\theta)_I}^{JK}=\frac{1}{2}\bar\epsilon\Gamma_I\psi^{JK}+t{\delta\theta_I}^{JK}+\OO(\a)\\
\delta\psi_{IJ}=&-\frac{1}{4}R^+_{IJKL}\Gamma^{KL}\epsilon=-\frac{1}{4}R^-_{KLIJ}\Gamma^{KL}\epsilon+\OO(\a)\\
=&-\frac{1}{4}\Big(R_{KLIJ}-t(\d_{\Theta^-}\theta)_{KLIJ}\Big)\Gamma^{KL}\epsilon+\OO(\a)
\end{align*}
where we have used \eqref{eq:curvatureid}. As the connection $\nabla$ always appears with a factor of $\a$, the $\OO(\a)$-terms can be neglected to the order we are working at, but they will become important in the next section when we discuss the theory to higher orders in $\a$. We thus see that $(\Theta,\psi_{IJ})$ transforms as an $SO(9,1)$-Yang-Mills multiplet, modulo $\OO(t)$ and $\OO(\a)$-terms. As noted, the $\OO(\a)$-terms can be ignored for now, but the $\OO(t)$-terms will have to be dealt with. This is done by changing the supersymmetry transformations accordingly as we shall see below\footnote{That a change of the connection requires a change of the supersymmetry transformations in order to have a supersymmetry invariant action was also noted in \cite{Melnikov:2014ywa}.}.

A lemma of Bergshoeff and de Roo \cite{Bergshoeff1989439} (see also \cite{Andriot:2011iw}) states that the action deforms as
\begin{equation}
\label{eq:lemmabos}
\frac{\delta S}{\delta\nabla^-}\propto\a\textrm{B}_\OO(\a^2)\:,
\end{equation}
under an infinitesimal deformation of the Hull connection. Here $\textrm{B}_0$ denotes a combination of zeroth order bosonic equations of motion. As the correction to the action due to the change of connection \eqref{eq:connchange} is proportional to the equations of motion, the change of connection $t\theta$ may equivalently be viewed as an infinitesimal field redefinition of order $\OO(t,\a)$, and is therefore non-physical.\footnote{That deformations of the connection corresponds to a field redefinition has been noted in the literature before, see e.g. \cite{Hull1986187, Sen1986289, Becker:2009df, Melnikov:2014ywa}.} Moreover, {\it any} infinitesimal $\OO(\a)$ field redefinition can similarly be absorbed in a change of the Hull connection. These sets are therefore equivalent, modulo $\OO(\a^2)$ terms.



We want to consider what happens to the theory under these deformations of the connection. In particular, we are interested in the allowed deformations of the connection, or equivalently field redefinitions, for which supersymmetric solutions of the Strominger system exist. We expect this to be related to the moduli space of connections \eqref{eq:modulitheta} studied in the last chapter. We see that this is indeed the case.

From \eqref{eq:lemmabos} it follows that the change to the action due to the correction,
\begin{equation*}
{\delta_t(\delta\Theta)_I}^{JK}=t{\delta\theta_I}^{JK}\:,
\end{equation*}
of the transformation of $\Theta$ can be absorbed in a redefinition of the bosonic supersymmetry transformations by a similar procedure as is done in \cite{Bergshoeff1989439} for the $\OO(\a^2)$-corrections to the supersymmetry transformations. Similarly, we also have
\begin{equation}
\label{eq:lemmaferm}
\frac{\delta S}{\delta\psi_{IJ}}\propto\a\Psi^{IJ}_\OO(\a^2)\:,
\end{equation}
by \cite{Bergshoeff1989439}, where $\Psi_0$ is a combination of zeroth order fermionic equations of motion. It follows that the change in the action due to the correction,
\begin{equation*}
\delta_t(\delta\psi)_{IJ}=\frac{t}{4}(\d_{\Theta^-}\theta)_{KLIJ}\Gamma^{KL}\epsilon\:,
\end{equation*}
may be absorbed into a redefinition of the fermionic supersymmetry transformations. The corrected transformations read
\begin{align}
\delta\psi_M=&\Big(\nabla^+_M+\frac{t}{4}\mathcal{C}_M\Big)\epsilon+\OO(\a^2)\notag\\
\label{eq:2spinorsusy1def}
=&\Big(\nabla^{\hbox{\tiny{LC}}}_M+\frac{1}{8}(\H_M+2t\mathcal{C}_M)\Big)\epsilon+\OO(\a^2)\\
\label{eq:2spinorsusy2def}
\delta\lambda=&-\frac{1}{2\sqrt{2}}\Big(\slashed\nabla^{\hbox{\tiny{LC}}}\phi+\frac{1}{12}(\H+3t\mathcal{C})\Big)\epsilon+\OO(\a^2)\\
\label{eq:2spinorsusy3def}
\delta\chi=&-\frac{1}{2}F_{MN}\Gamma^{MN}\epsilon+\OO(\a)\:,
\end{align}
where
\begin{equation}
\label{eq:C}
C_{MAB}=\a12e^{2\phi}\nabla^{+L}e^{-2\phi}\Big((\d_{\Theta^-}\theta)_{ABLM}\Big)\:.
\end{equation}
Here $\mathcal{C}_M=C_{MAB}\Gamma^{AB}$ and $\mathcal{C}=C_{MAB}\Gamma^{MAB}$, and $A,B,..$ denote flat indices, that is, $\Gamma^M=e^M_A\Gamma^A$, where $\{e^M_A\}$ is a ten-dimensional viel-bein frame. We have written the corrections in this way to be able to compare with the higher order $\a$-corrections in the next section. With the new supersymmetry transformations \eqref{eq:2spinorsusy1def}-\eqref{eq:2spinorsusy3def}, the action with the new connection is again invariant.

We also note that, as we saw above, deforming the connection $\nabla^-\rightarrow\nabla^-+t\theta$ really just corresponds to an $\OO(\a)$-field redefinition. Hence, the supersymmetry algebra above, including the bosonic transformations which we did not write down for brevity, should just be the old algebra written in terms of the new fields. There are therefore no issues concerning closedness of the algebra.

\subsection{Supersymmetric Solutions}
\label{sec:supsol1}
Let us look for four-dimensional supersymmetric maximally symmetric compact solutions to the $t$-adjusted theory. This ammounts to setting the transformations \eqref{eq:2spinorsusy1def}-\eqref{eq:2spinorsusy3def} to zero. We consider solutions such that
\begin{equation}
\label{eq:Stromtcorr}
\nabla^+_m\,\eta=\OO(\a^2)\:.
\end{equation}
Given the redefined supersymmetry transformations, this might seem like a restriction of allowed supersymmetric solutions. However, this is not the case, at least for compact solutions. Indeed, we have the following proposition.

\begin{Proposition}
\label{prop:Strom}
Consider heterotic compactifications to four dimensions on a smooth compact space $X$ at $\OO(\a^{2n-1})$ or less. If
\begin{equation*}
\nabla^+\eta=\OO(\a^n)\:,
\end{equation*}
then we may without loss of generality assume that
\begin{equation*}
\nabla^+\eta=0\:,
\end{equation*}
i.e. the solutions are solutions of the Strominger system.
\end{Proposition}
\begin{proof}
First note that since $\nabla^+_m\:\eta=\OO(\a^n)$, it follows that
\begin{align*}
\OO(\a^{2n})=(\nabla^+_m\:\eta,\nabla^{+m}\:\eta)=\int_X(\nabla^+_m\:\eta)^\dagger\nabla^{+m}\:\eta=\int_X\eta^\dagger\Delta_+\eta=(\eta,\Delta_+\eta)\:,
\end{align*}
upon an integration by parts\footnote{As the Bismut connection has antisymmetric torsion, it follows that $\nabla^+_mv^m=\nabla^{\hbox{\tiny{LC}}}_mv^m$ for some vector field $v^i$. This allows the integration by parts to be carried out.}. Here 
\begin{equation*}
\Delta_+\eta=-\nabla_m^+\nabla^{+m}\:\eta\:,
\end{equation*}
and $\Delta_+$ is the Laplacian of the Bismut connection.

Next expand $\eta$ in eigen-modes of $\Delta_+$,
\begin{equation*}
\vert\eta\rangle=\sum_n\alpha_n\vert\psi_n\rangle\:,
\end{equation*}
where $\{\vert\psi_n\rangle\}$ is an orthonormal basis of eigenspinors of $\Delta_+$ with corresponding eigenvalues ${\lambda_n}$, and where we have gone to bra-ket notation for convenience. We can then compute
\begin{equation*}
(\eta,\Delta_+\eta)=\langle\eta\vert\Delta_+\vert\eta\rangle=\sum_n\lambda_n\vert\alpha_n\vert^2=\OO(\a^{2n})\:.
\end{equation*}
Note that $\lambda_n\ge0$ as $\Delta_+$ is positive semi-definite. From this it follows that each term in the above sum is of $\OO(\a^{2n})$\footnote{We treat $\a$ as a formal expansion parameter, so sums of terms of higher orders in $\a$ cannot decrease the order in $\a$.}. That is  
\begin{equation}
\label{eq:sumtermO4}
\lambda_n\vert\alpha_n\vert^2=\OO(\a^{2n})\;\;\;\;\forall\;\;\;\;n.
\end{equation}
Moreover, we know that $\vert\eta\rangle=\OO(1)$, which implies
\begin{equation*}
\vert\vert\eta\vert\vert^2=\sum_n\vert\alpha_n\vert^2=\OO(1)\:.
\end{equation*}
It follows that at least one $\alpha_k=\OO(1)$. Then, from \eqref{eq:sumtermO4}, the corresponding eigenvalue is $\lambda_k=\OO(\a^{2n})$. At the given order in $\a$, we may without loss of generality set $\lambda_k=0$. It follows that there is a spinor in the kernel of $\nabla^+$, which we may take to be $\eta$.
\end{proof}
We now use equation \eqref{eq:2spinorsusy1def}, and Proposition \ref{prop:Strom} with $n=1$, to get \eqref{eq:Stromtcorr}.
It also follows from equation \eqref{eq:2spinorsusy1def} that we need
\begin{equation}
\label{eq:vanishC}
\mathcal{C}_m\eta=\OO(\a^2)\:,
\end{equation}
for the solution to be supersymmetric. From Appendix \ref{app:proof} it then follows that the corrected connection $\nabla=\nabla^-+t\theta$ should be an instanton.

It is easy to see that \eqref{eq:vanishC} is satisfied, once we know that we are working with supersymmetric solutions of the Strominger system. Plugging the connection $\nabla$ into the instanton condition, and using that $\nabla^-$ is an instanton at this order, we find
\begin{equation}
(\d_{\Theta^-}\theta)_{mn}\Gamma^{mn}\eta=\OO(\a)\:,
\end{equation}
precicely the condition for the deformed connection to remain an instanton. From this, it also follows that
\begin{align*}
\mathcal{C}_m\eta=&\a12e^{2\phi}\nabla^{+n}e^{-2\phi}\Big((\d_{\Theta^-}\theta)_{ABnm}\Big)\Gamma^{AB}\eta\\
=&\a12e^{2\phi}\nabla^{+n}e^{-2\phi}\Big((\d_{\Theta^-}\theta)_{ABnm}\Gamma^{AB}\eta\Big)\\
=&\OO(\a^2)\:,
\end{align*}
as desired.
 
Finally, we remark that as noted in section \ref{sec:instcond}, there is an infinitesimal moduli space 
\begin{equation}
\label{eq:Mtheta}
T\M_{\nabla^-}=H^{(0,1)}_{\bp_{\nabla^-}}(X,\End(TX))
\end{equation}
of connections satisfying this condition, where the tangent space is taken {\it at the Hull connection}. Each element $\theta\in T\M_{\nabla^-}$ corresponds to an infinitesimal $\OO(\a)$ field redefinition of the supergravity with the corresponding change of the supersymmetry transformations, \eqref{eq:2spinorsusy1def}-\eqref{eq:2spinorsusy3def}. Compact supersymmetric solutions of these equations may by Proposition \ref{prop:Strom} be assumed to be solutions of the Strominger system, and they also solve the equations of motion provided $\theta\in T\M_{\nabla^-}$. From this perspective, the moduli space \eqref{eq:Mtheta} is unphysical. That is, the moduli space \eqref{eq:Mtheta} may then be viewed as the space of allowed infinitesimal $\OO(\a)$ field redefinitions for which the equations of motion and supersymmetry are compatible.\footnote{We suppress issues concerning the preservation of the Bianchi identity \eqref{eq:BIanomaly} in this chapter. That is, as we saw in the last chapter, to preserve the Bianchi identity we should require $\theta\in\ker(\H)$.}

\section{Higher Order Heterotic Supergravity}
\label{sec:higher}
Having discussed the first order theory, we now consider heterotic supergravity at higher orders in $\alpha'$. We continue our investigation from a ten-dimensional supergravity point of view, by a similar analysis as that of Bergshoeff and de Roo \cite{Bergshoeff1989439}, where the Hull connection was used at higher orders as well. We wish to generalize this analysis a bit, and allow for a more general connection choice in the action, as was done in the previous section. In order not to overcomplicate matters unnecessarily, we return to letting the $TX$-connection be the Hull connection at $\OO(\a)$, which is needed in order that the full action be invariant under the usual supersymmetry transformations \eqref{eq:O1spinorsusyConn1}-\eqref{eq:O1spinorsusyConn3} at $\OO(\a)$. We will however allow this connection to receive corrections at $\OO(\a^2)$. 


There are two important points which we wish to emphasise in this section. Firstly, as we saw in the last section, we may deform the tangent bundle connection away from the Hull connection provided we deform the supersymmetry transformations correspondingly. We take a similar approach in this section, where we deform away from the Hull connection by an $\a$-correction, $\nabla=\nabla^-+\theta$, where now $\theta=\OO(\a)$ depends on the fields of the theory in some way. Our findings from the previous section also persist in this section. That is, the deformation $\theta$ now corresponds to an $\OO(\a^2)$ field redefinition, and $\theta$ is therefore non-physical in this sense. Moreover, the supersymmetry transformations also change with $\theta$, in accordance with the deformed fields. However, as we shall see again, not all field choices allow for supersymmetric solutions of the type we consider.


Secondly, we note there is a symmetry between the tangent bundle connection $\nabla$ and gauge connection $A$ in the $\OO(\a)$ action. As a guiding principle, as is also done in \cite{Bergshoeff1989439}, we would like to keep this symmetry to higher orders. With this philosophy it seems natural to choose $\nabla$ so it satisfies its own equation of motion similar to that of $A$, whenever the equations of motion are satisfied. Note that this is true for the Hull connection at $\OO(\a)$, by equation \eqref{eq:lemmabos}. 


Moreover, this also seems to be the connection choice we need in order for the supersymmetry conditions to hold at the locus of equations of motion. Indeed, we find the following
\begin{Theorem}
\label{tm:inst}
Strominger system type supersymmetric solutions, where $\nabla^+\epsilon=0$ for heterotic compactifications on a compact six-dimensional manifold $X$, survive as solutions of heterotic supergravity at $\OO(\a^2)$ if and only if the connection $\nabla$ is an instanton, satisfying it's own ``supersymmetry condition"
\begin{equation}
R_{mn}\Gamma^{mn}\epsilon=0\:.
\end{equation}
Compact $\OO(\a^2)$-supersymmetric solutions can be assumed to be of this type without loss of generality. Moreover, $\nabla$ satisfies it's own equation of motion for these solutions.
\end{Theorem}


Note then that our choice of connection is as if the connection $\nabla$ was dynamical. We again stress that this is not the case. Indeed, $\nabla$ must depend on the other fields of the theory, as there are no more fields around. We only require the connection to satisfy an equation of motion as if it was dynamical, and this then relates to \it how\rm\;$\nabla$ depends on the other fields. 

With these observations, we make the following conjecture
\begin{Conjecture}
\label{con:higherorder}
At higher orders in $\a$, the correct connection choice/field choice is the choice which preserves the symmetry between $\nabla$ and $A$. That is, $\nabla$ should be chosen as if it was dynamical, satisfying it's own equation of motion. Moreover, for supersymmetric solutions, $\nabla$ should be chosen to satisfy it's own supersymmetry condition, similar to the one satisfied by $A$.
\end{Conjecture}


\subsection{The Second Order Theory}
According to Bergshoeff and de Roo \cite{Bergshoeff1989439} the bosonic part of the heterotic action does not receive corrections at this order, and is still given by \eqref{eq:actionRedef}. Bergshoeff and de Roo used the Hull connection when writing down the action, but we shall be more generic, choosing a connection $\nabla$ that only differs from the Hull connection by changes of $\OO(\a)$,
$\nabla=\nabla^-+\OO(\a)$.

The supersymmetry transformations do receive corrections. What these corrections are again depend crucially on which connection is chosen in the action as we will discuss in the next section. Using the Hull connection $\nabla=\nabla^-$, they are given in \cite{Bergshoeff1989439} and read\footnote{It should be noted that the specific form of these corrections, where there are no covariant derivatives of the spinor in the $\OO(\a^2)$-correction requires an addition of an extra term of $\OO(\a^2)$ to the fermionic sector action \cite{Bergshoeff1989439}.}
\begin{align}
\delta\psi_M=&\Big(\nabla^+_M+\frac{1}{4}\P_M\Big)\epsilon\notag\\
\label{eq:2spinorsusy1}
=&\Big(\nabla^{\hbox{\tiny{LC}}}_M+\frac{1}{8}(\H_M+2\P_M)\Big)\epsilon+\OO(\a^3)\\
\label{eq:2spinorsusy2}
\delta\lambda=&-\frac{1}{2\sqrt{2}}\Big(\slashed\nabla^{\hbox{\tiny{LC}}}\phi+\frac{1}{12}\H+\frac{3}{12}\P\Big)\epsilon+\OO(\a^3)\\
\label{eq:2spinorsusy3}
\delta\chi=&-\frac{1}{2}F_{MN}\Gamma^{MN}\epsilon+\OO(\a^2),
\end{align}
where
\begin{equation}
P_{MAB}=-\a6e^{2\phi}\nabla^{+L}(e^{-2\phi}\d H_{LMAB})\:.
\end{equation}
Here $\P_M=P_{MAB}\Gamma^{AB}$ and $\P=P_{MAB}\Gamma^{MAB}$. Here $A, B,\:..$ denote flat indices, while $I, J,\: ..$ denote space-time indices. Note again the reduction in $\a$ for the gauge-field transformation \eqref{eq:2spinorsusy3}.

\subsection{Second Order Equations of Motion}
We now derive the equations of motion of the action \eqref{eq:actionRedef}. As the action is the same as the first order action, one might guess that the equations of motion will be the same too. This is not quite correct, and we take a moment to explain why.

When deriving the first order equations of motion, one relies on the lemma of \cite{Bergshoeff1989439}, equation \eqref{eq:lemmabos}, from which it follows that the Hull connection satisfies an equation of motion of its own, whenever the other fields do. As a necessary condition to satisfying the first order equations of motion is that the zeroth order equations of motion are of $\OO(\a)$, the variation of the action with respect to $\nabla^-$ can be ignored as it is of $\OO(\a^2)$. This simplifies matters when deriving the first order equations of motion. At second order however, such terms will have to be included, potentially leading to a more complicated set of equations. 


We note that the $\OO(\a^2)$-corrections to the equations of motion come from the variation of the action with respect to $\nabla$. What they are will crutially depend on what connection $\nabla$ is used. Let us write the connection one-form of $\nabla$ as
\begin{equation*}
\Theta=\Theta^-+\theta,
\end{equation*}
where $\Theta^-$ are the connection one-forms of the Hull connection, and $\theta=\OO(\a)$, and depends on the other fields of the theory in some unspecified way. The action then takes the form
\begin{equation}
\label{eq:thetaaction}
S=S[\nabla^-]+\delta_\theta S+\OO(\a^3)\:.
\end{equation}
Let us compute $\delta_\theta S$. We find
\begin{equation*}
\delta_\theta S=\int_{M_{10}}e^{-2\phi}\Big[\delta_\theta H\wedge*H+\frac{\a}{2}\tr[\d_{\Theta^-}\theta\wedge*R^-]\Big]\:.
\end{equation*}
Now
\begin{equation*}
\delta_\theta H=-\frac{\a}{4}\delta_\theta\omega_{CS}^\nabla=\frac{\a}{2}\tr[\theta\wedge R^-]+\frac{\a}{4}\d\tr[\theta\wedge\Theta^-]\:.
\end{equation*}
Inserting this back into the action, we find
\begin{align*}
\delta_\theta S=\frac{\a}{2}\int_{M_{10}}e^{-2\phi}\tr\:\theta\wedge\Big[e^{2\phi}\d_{\Theta^-}(e^{-2\phi}*R^-)-R^-\wedge*H+\Theta^-\wedge e^{2\phi}\d(e^{-2\phi}*H)\Big]\:.
\end{align*}
We write this as
\begin{equation}
\label{eq:dstheta}
\delta_\theta S=\frac{\a}{2}\int_{M_{10}}e^{-2\phi}\tr\:\theta\wedge\textrm{B}_0 +\OO(\a^3)\:,
\end{equation}
since the expression in brackets is proportional to a combination of zeroth order bosonic equations of motion according to \eqref{eq:lemmabos}. The change of connection $\theta$ may be thought of as an $\OO(\a^2)$ field redefinition, as this is precisely how the action gets corrected when we perform an $\OO(\a^2)$ field redefinition. This is similar to the $\OO(t,\a)$ field redefinitions we described in the previous section. It follows that the change of the connection $\theta$ is unphysical in this sense. 

Let us next compute the variation of the action \eqref{eq:thetaaction} with respect to the connection $\nabla$, assuming that the first order equations of motion are satisfied. Recall that such a variation is of $\OO(\a^2)$, provided the zeroth order equations of motion are satisfied. Hence, it could be dropped when only considering the theory to $\OO(\a)$. Using the first order equations of motion, in particular
\begin{equation*}
\d(e^{-2\phi}*H)=\OO(\a^2)\:,
\end{equation*}
we find
\begin{align}
\delta_\nabla S\big\vert_{\delta S=\OO(\a^2)}=&\frac{\a}{2}\int_{M_{10}}e^{-2\phi}\tr\:\delta\Theta^-\wedge\Big[[\theta,*R^-]+e^{2\phi}\d_{\Theta^-}(e^{-2\phi}*\d_{\Theta^-}\theta)\notag\\
&-\d_{\Theta^-}\theta\wedge*H+e^{2\phi}\d_{\Theta^-}(e^{-2\phi}*R^-)-R^-\wedge*H\Big]+\OO(\a^3)\:.
\label{eq:deltaSnabla}
\end{align}
Note that any variations depending on $\delta\theta$ drop out of this expression. This is due to \eqref{eq:dstheta} and that $\theta$ is of order $\a$, which implies that variations of the action with respect to $\theta$ are of $\OO(\a^3)$ at the locus of the first order equations of motion. We therefore only need to worry about the $\delta\Theta^-$-part when varying the action with respect to $\nabla$.

Equation \eqref{eq:lemmabos} also guarantees that the expression in \eqref{eq:deltaSnabla} is of $\OO(\a^2)$. The change of the $\OO(\a^2)$ equations of motion depend on what the expression in the brackets is, which again depends on our connection choice. It should be stressed that even though $\theta$ corresponds to a field choice, this does not mean that any field choice will do. {\it We want to choose our fields so that supersymmetry, and in particular the Strominger system, is compatible with the equations of motion}. 

For instance, insisting upon the Hull connection $\nabla^-$ ($\theta=0$) would in general change the equations of motion. It also restricts the allowed supersymmetric solutions as we shall see below. The connection choice therefore appears to require corrections from the Hull connection at higher orders. Motivated by the symmetry between $A$ and $\nabla$ in the action, we try a more ``symmetric" connection choice.


Recall that the $\beta$-functions of the world-sheet sigma-model correspond to the heterotic supergravity equations of motion. In \cite{Foakes1988335} it was noted that the three-loop $\beta$-function of the gauge connection equal the two-loop $\beta$-function.\footnote{The $n$-loop $\beta$-functions of the sigma models correspond to the $\OO(\a^{n-1})$ supergravity}
That is, the $\beta$-function of the gauge field does not receive corrections at this order, so nor should the corresponding supergravity equation of motion. This is consistent with the supergravity point of view \cite{Bergshoeff1989439}. Motivated by this, and guided by the symmetry between $\nabla$ and the gauge connection in the action, it seems natural to choose $\nabla$ so that it satisfies it's own equation of motion
\begin{equation}
\label{eq:nablaeom}
e^{2\phi}\d_\Theta(e^{-2\phi}*R)-R\wedge*H=\OO(\a^2),
\end{equation}
at this order. This is exactly the equation one gets when varying the action with respect to $\nabla$, and which is indeed satisfied by the Hull connection at first order. It is easy to see that choosing this connection is in fact equivalent to choosing $\theta$ so that the expression in brackets in \eqref{eq:deltaSnabla} vanishes, modulo higher orders. This again implies that all the first order equations of motion remain the same to $\OO(\a^2)$.


Of course, changing the connection also requires that we change the supersymmetry variations appropriately, in order that the full action remains invariant under supersymmetry transformations at $\OO(\a^2)$. This also relates to how we correct the connection outside of the locus of equations of motion. We will return to this below, when we consider supersymmetric solutions. 

Note that supersymmetric solutions may be assumed to be solutions of the Strominger system without loss of generality. Indeed, an $\OO(\a)$ change to $\nabla^-$ implies an $\OO(\a^2)$ change to the action, and hence an $\OO(\a^2)$ change to teh supersymmetry transformations. Using Proposition \ref{prop:Strom} with $n=2$, we see that we may assume $\nabla^+\eta=0$. By appendix \ref{app:proof}, such solutions exist if and only if $\nabla$ is an instanton, and in particular \eqref{eq:nablaeom} is satisfied. This is in complete analogy with the gauge connection, as the supersymmetry condition for $A$ is that $F$ remains an instanton at $\OO(\a^2)$ as well.

Note further that insisting on the Hull connection in the second order theory will in general impose extra constraints. Consider supersymmetric solutions of the Strominger system at $\OO(\a^2)$. By appendix \ref{app:proof}, we require
\begin{equation*}
R^-_{mn}\Gamma^{mn}\eta=\OO(\a^2)\:.
\end{equation*}
However, by \eqref{eq:curvatureid}, it then follows that we need
\begin{equation*}
\d H=-2i\p\bp\omega=\OO(\a^2)\:,
\end{equation*}
an unnecessary extra constraint on the geometry. This was also argued in \cite{Ivanov:2009rh}, where the first order theory was taken to be exact, resulting in a Calabi-Yau geometry.

\subsection{Choosing Other Connections}
\label{sec:differentcon}
We now consider in detail what happens if a different connection, other than the Hull connection is chosen, that is $\theta\neq0$. We work at $\OO(\a^2)$ for the time being, and leave the cubic and higher order corrections for future work. 

As argued in \cite{Bergshoeff1989439}, the higher order corrections to the supersymmetry transformations come from the failure of $(\Theta,\psi_{IJ})$ to transform as a $SO(9,1)$ Yang-Mills multiplet, where $\psi_{IJ}$ is the supercovariant curvature given by \eqref{eq:scurv}. $(\Theta,\psi_{IJ})$ then transforms under supersymmetry as
\begin{align*}
\delta {\Theta_I}^{JK}=&{(\delta\Theta^-+\delta\theta)_I}^{JK}\\
=&\frac{1}{2}\bar\epsilon\Gamma_I\psi^{JK}+\OO(\a)\\
\delta\psi_{IJ}=&-\frac{1}{4}R^+_{IJKL}\Gamma^{KL}\epsilon\\
=&-\frac{1}{4}\Big({R}_{KLIJ}+\frac{1}{2}\d H_{KLIJ}-\d_{\Theta^-}\theta_{KLIJ}\Big)\Gamma^{KL}\epsilon\notag\\
=&-\frac{1}{4}{R}_{KLIJ}\Gamma^{KL}\epsilon+\OO(\a)\:,
\end{align*}
where \eqref{eq:curvatureid} has been used in the second equality of the expression for $\delta\psi_{IJ}$. Note that without the $\a$-effects, the multiplet transforms as a $SO(9,1)$ Yang-Mills multiplet. This is how the symmetry in the action between the gauge connection and tangent bundle connection arises at $\OO(\a)$.

The $\OO(\a)$ correction to the transformation of ${\Theta_I}^{JK}$ depends on how the correction $\theta$ of the connection is defined in terms of the other fields of the theory. This correction is what makes the action fail to be invariant under supersymmetry transformations. However, this failure of the action to be invariant may be absorbed into an $\OO(\a^2)$-redefinition of the bosonic supersymmetry transformations due to \eqref{eq:lemmabos}, as is done in \cite{Bergshoeff1989439} for the case of the Hull connection. 

The same holds for the $\OO(\a)$ correction to $\delta\psi_{IJ}$,
\begin{equation*}
\delta_{\a}(\delta\psi_{IJ})=-\frac{1}{8}\Big(\d H_{KLIJ}-2\d_{\Theta^-}\theta_{KLIJ}\Big)\Gamma^{KL}\epsilon\:.
\end{equation*}
This can be absorbed into a redefinition of the supersymmetry transformations of the fermions due to \eqref{eq:lemmaferm}. For the more general connection choice, it turns out that the correction we need only requires a change of the three-form $P$,
\begin{align}
\label{eq:changeP}
P_{MAB}\rightarrow\tilde P_{MAB}=-\a6e^{2\phi}\nabla^{+L}e^{-2\phi}\Big(\d H_{LMAB}-2(\d_{\Theta^-}\theta)_{ABLM}\Big)\:,
\end{align}
but otherwise the transformations \eqref{eq:2spinorsusy1}-\eqref{eq:2spinorsusy3} remain the same. Note also that as the deformation of the connection can again be viewed as an $\OO(\a^2)$ field redefinition, the new supersymmetry algebra is again closed.

We now compactify our theory on a complex three-fold $X$. As noted above, we can again assume without loss of generality that
\begin{equation*}
\nabla^+_m\eta=\OO(\a^3)\:.
\end{equation*}
By the rewriting of the bosonic action \eqref{eq:6daction}, which we again stress holds true at $\OO(\a^2)$, we find that for the equations of motion to hold we need $R=R^-+\d_{\Theta^-}\theta+\OO(\a^2)$ to satisfy the instanton condition, 
\begin{equation}
\label{eq:O2instanton}
R_{mn}\Gamma^{mn}\eta=\OO(\a^2)\:.
\end{equation}
Note the similarity between this condition and the supersymmetry condition for the gauge field \eqref{eq:2spinorsusy3}. 


Supersymmetry now also requires
\begin{equation*}
\tilde P_{mAB}\Gamma^{AB}\eta=\OO(\a^3)\;,
\end{equation*}
by \eqref{eq:changeP} and \eqref{eq:2spinorsusy1}. Here $A,\:B$ denote flat indices on $X$.  This equation is however trivial, once we know that $R$ is an instanton. Indeed, we have
\begin{align*}
\tilde P_{mAB}\Gamma^{AB}\eta=&-\a6e^{2\phi}\nabla^{+n}e^{-2\phi}\Big(\d H_{nmAB}-2(\d_{\Theta^-}\theta)_{ABnm}\Big)\Gamma^{AB}\eta\\
=&\a12e^{2\phi}\nabla^{+n}e^{-2\phi}\Big(R_{ABnm}\Gamma^{AB}\eta\Big)+\OO(\a^3)\\
=&\OO(\a^3)\;,
\end{align*}
where we again used \eqref{eq:curvatureid} in the second equality.

It should also be mentioned that the instanton connection also solves the $\nabla$-equation of motion, as shown in \cite{Gauntlett:2002sc, Gillard:2003jh}. Indeed in dimension six, by the supersymmetry conditions, it follows that
\begin{equation*}
e^{2\phi}\d_\Theta(e^{-2\phi}*R)-R\wedge*H=e^{2\phi}\d_\Theta e^{-2\phi}(*R+R\wedge\omega)\:.
\end{equation*}
As $R$ is both of type $(1,1)$, and primitive, we have the identity $*R=-\omega\wedge R$ by \eqref{eq:Weil}. It follows that the instanton connection satisfies the $\nabla$-equation of motion, and the first order equations of motion do not receive corrections. 



We have thus gone through the proof of the statements in Theorem \ref{tm:inst}. Next, we want to consider their interpretation and give a discussion of the results. In doing so we also give our reasons for proposing Conjecture \ref{con:higherorder}.

\section{Discussion}
\subsection{Summary of Results}
\label{sec:discConn}
In the first order theory, we saw that the connection $\nabla=\nabla^-+t\theta$, where $\theta$ depends on the fields of the theory in some way, should satisfy the instanton condition whenever the solution is supersymmetric. As shown in the last chapter, this condition has an infinitesimal moduli space of the form
\begin{equation}
\label{eq:Mdnabla}
T\M_{\nabla^-}\cong H^{(0,1)}_{\bp_\Theta}(X,\End(TX))\:,
\end{equation}
where the tangent space is taken {\it at the Hull connection} $\nabla^-$. At first order, the requirement that the full supergravity action should be invariant under the usual supersymmetry transformations reduces the choice to the Hull connection. Hence, the $t$-deformed theory requires changes to the supersymmetry transformations, and we found what these where. We also found the allowed deformation space of connections, for which supersymmetric solutions of the Strominger system exist, was given by \eqref{eq:Mdnabla}. Recall also that supersymmetric solutions may be assumed to be solutions of the Strominger system by Proposition \ref{prop:Strom}. Moreover, by the lemma of Bergshoeff and de Roo \cite{Bergshoeff1989439}, these deformations correspond to infinitesimal $\OO(\a)$ field redefinitions. 


Returning to the usual form of the first order supergravity, we saw that at second order the theory can again be corrected appropriately for any $\OO(\a)$-change $\theta$ of the Hull connection $\nabla^-$, corresponding to $\OO(\a^2)$ field redefinitions. Supersymmetric solutions could again be assumed to be solutions of the Strominger system, and the equations of motion are compatible with supersymmetry if and only if $\nabla=\nabla^-+\theta$ satisfies the instanton condition again.

\subsection{Higher orders}
Let us now take a moment to discuss higher orders in $\a$. Note that the condition we find for compatibility between supersymmetry and equations of motion, \eqref{eq:O2instanton}, is exactly the supersymmetry condition we would get from this ``connection sector" if $\nabla$ was part of a dynamical superfield, very much analogous to the gauge sector. Indeed, the fact that $(\nabla^-,\psi_{IJ})$ transforms as an $SO(9,1)$-Yang-Mills multiplet to $\OO(\a)$ is what motivated the construction of the action of \cite{Bergshoeff1989439} in the first place. From the discussion above, it appears that supersymmetric solutions behave as if this where the case, at least for compact solutions including $\OO(\a^2)$. The question then arises what happens at $\OO(\a^3)$ and higher?

It should first be noted that at higher orders, the form of the supergravity action is no longer unique, and undetermined $(\textrm{curvature})^4$-terms appear \cite{Bergshoeff1989439}. The form of these terms may however be determined through other means such as string amplitude calculations \cite{Gross198741, Cai1987279}, which was also used in \cite{Bergshoeff1989439}, and these terms indeed preserve the symmetry between the Lorentz and Yang-Mills sectors.

With this, it therefore seems natural to conjecture that the above structure also survives to higher orders. That is, the natural connection $\nabla$ used to calculate the curvatures should be chosen so that it satisfies an equation of motion similar to that of $A$, whenever the other equations of motion are satisfied. Moreover, for supersymmetric solutions, $\nabla$ should satisfy a supersymmetry condition similar to that of $A$. We also conjecture that, as seen to order $\OO(\a^2)$, the moduli of this ``supersymmetry condition" are equivalent to field redefinitions, and therefore do not correspond to physical lower energy fields in any sense.  

\subsection{Future directions}
Having reviewed our results, and discussed higher orders in $\a$, there are a few unanswered questions which we would like to look into in the future. Firstly, it would be interesting to check the proposed conjecture to the next order in $\a$. This should not be very difficult, as the cubic theory was laid out in general in \cite{Bergshoeff1989439}, and we only have to repeat their analysis using a more general connection. It should be noted that  at this order, the supersymmetry condition for the gauge field does receive corrections, 
and we expect this to be true for the tangent bundle connection as well.

Next, it would be interesting to return to the first order theory, and consider higher order deformations of the Hull connection. Indeed, in section \ref{sec:FirstOrder1} we only considered infinitesimal deformations away from the Hull connection of this theory. That is, we considered the tangent space of the moduli space of connections {\it at the Hull connection}, which we saw corresponded to infinitesimal $\OO(\a)$ field redefinitions. It would be interesting to perform higher order deformations of the connection, i.e. deformations of $\OO(t^2)$ and above, and to see how this relates to obstructions of the corresponding deformation theory. Moreover, do such ``finite" deformations still correspond to field redefinitions? 

It would also be interesting to consider our findings in relation to the sigma model. As was pointed out in \cite{Sen1986289} for the first order theory, and as we also find, changing the connection $\nabla$ corresponds to $\OO(\a)$ field redefinitions. Requiring world-sheet conformal invariance, i.e. the ten-dimensional equations of motion, in addition to space-time supersymmetry, puts conditions on the connection. As we have seen, and as was first noted in \cite{Hull1986357}, it is sufficient to use the Hull connection at first order. This connection was also necessary modulo field redefinitions. We found that the allowed field redefinitions correspond to the moduli space \eqref{eq:Mdnabla}, and it would be interesting to see if this moduli space can be retrieved from the sigma model point of view as well.

At next order, we found that the Hull connection was not a good field choice, provided we want supersymmetric solutions to the Strominger system. Still, we found the necessary and sufficient condition for compatibility was that $\nabla$ should satisfy the instanton condition. Moreover, $\nabla$ is related to the Hull connection by a corresponding $\OO(\a^2)$ field redefinition. But for supersymmetric solutions of the Strominger system, the Hull connection lead to too stringent constraints on the geometry. It would be interesting to investigate this further from a sigma-model point of view. In particular, it would be interesting to see what the more ``physical field choices", i.e. connections satisfying the instanton condition, look like in this picture.

\begin{subappendices}

\section{Proof of Instanton Condition}
\label{app:proof}
In this appendix we repeat the proof of \cite{delaOssa:2014cia}, showing that supersymmetric solutions of the Strominger system and the equations of motion are compatible if and only if $\nabla$ is an instanton. We consider the theory including $\OO(\a^2)$ terms.

Recall first that the second order bosonic action is the same as the first order action \cite{Bergshoeff1989439}. According to \cite{Cardoso:2003af}, the six-dimensional part of the action \eqref{eq:action} may be written in terms of $SU(3)$-structure forms as
\begin{align}
\label{eq:6daction}
S_6=&\frac{1}{2}\int_{X_6}e^{-2\phi}\Big[8\vert d\phi-W_1^\omega\vert^2+\omega\wedge\omega\wedge\mathcal{\hat R}-\vert H-e^{2\phi}*\d(e^{-2\phi}\omega)\vert^2\Big]\notag\\
&-\frac{1}{4}\int\d ^6y\sqrt{g_6}{N_{mn}}^pg^{mq}g^{nr}g_{ps}{N_{nq}}^s\notag \\
&-\frac{\alpha'}{2}\int\d ^6y\sqrt{g_6}e^{-2\phi}\Big[\tr\vert F^{(2,0)}\vert^2+\tr\vert F^{(0,2)}\vert^2+\frac{1}{4}\tr\vert F_{mn}\omega^{mn}\vert^2\Big]\notag\\
&+\frac{\alpha'}{2}\int\d ^6y\sqrt{g_6}e^{-2\phi}\Big[(\tr\vert{R}^{(2,0)}\vert^2+\tr\vert {R}^{(0,2)}\vert^2+\frac{1}{4}\tr\vert R_{mn}\omega^{mn}\vert^2\Big]+\OO(\a^3)\:,
\end{align}
where the Bianchi Identity has been applied. $\mathcal{\hat R}$ is now the Ricci-form of the unique connection $\hat\nabla$ with totally antisymmetric torsion, for which the complex structure is parallel. For supersymmetric solutions of the Strominger system $\nabla^+$ coincides with $\hat\nabla$, which is known as the Bismut connection in the mathematics literature. The Ricci-form is
\begin{equation*}
\mathcal{\hat R}=\frac{1}{4}\hat R_{pqmn}\omega^{mn}\d x^p\wedge\d x^q,
\end{equation*}
while ${N_{mn}}^p$ is the Nijenhaus tensor for this almost complex structure. Note that
\begin{equation*}
\mathcal{\hat R}=0
\end{equation*}
is an integrability condition for supersymmetry.

Performing a variation of the action at the supersymmetric locus, we find that most of the terms vanish. The only surviving terms are
\begin{align}
\label{eq:6dvary}
\delta S_6=&\frac{1}{2}\int_{X_6}e^{-2\phi}\omega\wedge\omega\wedge\delta\mathcal{\hat R}\notag\\
&+\frac{\alpha'}{2}\delta\int\d ^6y\sqrt{g_6}e^{-2\phi}\Big[(\tr\vert{R}^{(2,0)}\vert^2+\tr\vert {R}^{(0,2)}\vert^2+\frac{1}{4}\tr\vert R_{mn}\omega^{mn}\vert^2\Big]+\OO(\a^3).
\end{align}
In \cite{Cardoso:2003af} it is shown that $\delta\mathcal{\hat R}$ is exact, and therefore the first term vanishes using supersymmetry by an integration by parts. If the equations of motion are to be satisfied to the order we work at, we therefore find
\begin{equation*}
R_{mn}\Gamma^{mn}\eta=\OO(\a^2)\:,
\end{equation*}
i.e. the instanton condition. Note the reduction in orders of $\a$ due to the factor of $\a$ in front of the curvature terms in the action.

\section{The Hull connection}
\label{app:Hull}
For completeness, we also repeat the argument of \cite{Martelli:2010jx} that the Hull connection does indeed satisfy the instanton condition for the $\OO(\a)$-theory, whenever we have supersymmetry. 

It is easy to prove that
\begin{equation}
\label{eq:curvatureid}
R^+_{MNPQ}-R^-_{PQMN}=\frac{1}{2}\d H_{MNPQ}.
\end{equation}
At zeroth order we get
\begin{equation*}
R^+_{MNPQ}=R^-_{PQMN}+\OO(\a),
\end{equation*}
by the Bianchi Identity. Contracting both sides with $\Gamma^{PQ}\epsilon$, and using 
\begin{equation*}
R^+_{MNPQ}\Gamma^{PQ}\epsilon=\OO(\a^2)
\end{equation*}
at the supersymmetric locus, we find
\begin{equation*}
R^-_{PQ}\Gamma^{PQ}\epsilon=\OO(\a)\:,
\end{equation*}
as required.

\end{subappendices}

\part{Moduli in Domain Wall Compactifications}
\label{part:3d}

\chapter{Non-Maximally Symmetric Compactifications and Four-Dimensional Supergravity}
\label{ch:nonmaxcomp}


In this part of the thesis we change gear and discuss a different issue that appears in string-compactifications, namely that of moduli stabilisation. As we saw in chapter \ref{ch:SS}, string-compactifications generically produce a lot of moduli fields. These fields are massless to first approximation. However, as they are not observed in accelerators, they need to be given a mass in order to be ``lifted" away, to agree with experiment. In this part of the thesis, we discuss different ways of achieving this. We will focus on the moduli associated to the compact space $X_6$, referred to as geometric, or gravitational moduli. The bundle moduli are often associated to fields in the observable spectrum, and to include these would require a full treatment in the spirit of chapter \ref{ch:SS}. In particular, one would need to do the full dimensional reduction, including the bundle, which is beyond the scope of this thesis. We shall however do a couple of consistency checks when considering explicit models. E.g., we check by the Dirac index that we have the correct number of generations. 

We will discuss two different ingredients used to stabilise moduli, fluxes and torsion. Let's begin with torsion. It turns out that by allowing the internal geometry to be non-K\"ahler, the superpotential induces terms proportional to the torsion $\d\omega$,
\begin{equation*}
W\propto\int_{X_6}\Big(H+i\d\omega\Big)\wedge\Psi\:,
\end{equation*}
where now $H$ is given by \eqref{eq:anomalycancellation}, that is it includes the connections on the bundles in general. $W$ can be used to stabilize moduli, and we will consider an explicit example of this in chapter \ref{ch:HFNK}. Unfortunately, not all torsional configurations allow for a maximally symmetric vacuum. Indeed, as we saw in chapter \ref{ch:SS}, requiring Minkowski space reduces $X_6$ to be a {\it heterotic} $SU(3)$-structure manifold. We want to extend this by allowing $X_6$ to have a more general $SU(3)$-structure. Specifically, we consider compactifications on manifolds of \it half-flat\rm\:$SU(3)$-structure.

To achieve this, we must forgo the assumption of a Minkowski vacuum, and consider non-maximally symmetric perturbative vacua.\footnote{Perturbative here means in terms of the $\a$-expansion, not including non-perturbative effects like gaugino condensates or branes.} The simplest generalisation of this is a domain wall, which has half-flat structures as its most generic supersymmetric solution. In these solutions, the fields depend on a particular direction, which we denote the $y$-direction. From a four-dimensional perspective, these solutions only preserve half of the supercharges, and are therefore known as {\it half}-BPS.

In the next chapter, we will consider half-flat torsional geometries known as cosets. These have nice geometrical properties, as explained in the introduction. We will see that with a combination of torsion and $\a$-effects, it is possible to stabilise all geometric moduli, leaving a $y$-dependent dilaton. We will show this both from a ten-dimensional and four-dimensional perspective. Also, by addition of non-perturbative effects, we will see that it is possible to lift the vacuum to a maximally symmetric one in a heterotic KKLT scenario. 

In chapter \ref{ch:CY} we discuss K\"ahler Calabi-Yau solutions. The purpose of this chapter is to show that fluxes can be used as a viable tool for moduli stabilisation, even for Calabi-Yau compactifiations. This has both pragmatic and calculational advantages. Indeed, in Part~\ref{part:4d} of the thesis, we attempted to make modest progress in the study of {\it heterotic} $SU(3)$-structures. A full understanding of the moduli space of such geometries is however far from complete. For Calabi-Yau's on the other hand, far more is known, and keeping the internal space Calabi-Yau is therefore an advantage when doing phenomenology. 

 As noted, heterotic theory only has NS flux available, and this vanishes for supersymmetric Minkowski solutions with an internal Calabi-Yau. To use flux, we therefore need to sacrifice the maximally symmetric assumption. We again consider the example of a domain wall. We show that the ten-dimensional and four-dimensional solutions can be matched everywhere in moduli space, for generic fluxes. We work at zeroth order in $\a$ in chapter \ref{ch:CY}, sufficient for the point we wish to make.




\section{Heterotic Half--BPS Domain Wall Solutions}
We now briefly discuss the general setting of $N=1$ heterotic domain wall solutions. Half-flat manifolds and, in particular, the nearly K\"ahler manifolds that we shall be concerned with in the next chapter, form solutions to the heterotic equations provided they are combined with a four-dimensional domain-wall solution~\cite{Lukas:2010mf, Klaput:2011mz}. In this case, the variation of the half-flat manifold along the direction transverse to the domain wall is described by Hitchin flow equations, or generalisations thereof. 

Our 10-dimensional solutions consist of a six-dimensional space with $SU(3)$-structure 
and a four-dimensional domain wall, as described in Refs.~\cite{Lukas:2010mf, Klaput:2011mz, Gray:2012md, delaOssa:2014lma}. This amounts to choosing the $1 + 2$ dimensions along the domain wall to be maximally symmetric and the remaining seven dimensions to form a non-compact $G_2$-structure manifold. The space-time now takes the form of a product
\begin{equation*}
M_{10}=X_7\times M_3\:,
\end{equation*}
where $X_7=Y$ is a seven-dimensional non-compact space with $G_2$-structure, and $M_3$ is three-dimensional Minkowski space. The associated metric takes the form
\begin{align}\label{domainwallansatz}
\d s^2=\eta_{\alpha\beta}\d x^\alpha\d x ^\beta+\underbrace{\d y^2+\underbrace{g_{uv}(x^m)\d x^u\d x^v}_{\text{$X,\;SU(3)$-structure}}}_{\text{$Y,\;G_2$ structure}}\:.
\end{align}
Here $\alpha$, $\beta$, ... range from 0 to 2 and label the domain wall coordinates, $y=x^3$ is the remaining four-dimensional direction, transverse to the domain wall, and $u$, $v$, ... run from 4 to 9 and label coordinates of the internal compact manifold $X$. The indices $m$, $n$, ... run from 3 to 9 and label all seven directions of the $G_2$-structure manifold $Y$. 

As evident from the above equation, the seven-dimensional $G_2$-structure manifold $Y$ is a warped product of the $y$-direction and the $SU(3)$-structure manifold $X$. To describe this structure mathematically, it is most convenient to formulate the $G_2$ and $SU(3)$-structures in terms of differential forms, which we will do in the next section.

\section{$G_2$- and $SU(3)$-Structure from the Supersymmetry Conditions}
\label{sec:G2susy}

We now briefly review how the conditions for unbroken supersymmetry,  \eqref{eq:O1spinorsusy1}-\eqref{eq:O1spinorsusy3}, give rise to the $G_2$- and $SU(3)$-structures of the domain wall solution \eqref{domainwallansatz}, mainly following Ref.~\cite{Lukas:2010mf}.

The general ten dimensional Majorana-Weyl spinor $\varepsilon$ which appears in the supersymmetry conditions \eqref{eq:O1spinorsusy1}-\eqref{eq:O1spinorsusy3} is decomposed in accordance with our metric Ansatz \eqref{domainwallansatz} as
\begin{align}
\varepsilon(x^m)=\rho\otimes\eta(x^m)\otimes\theta\:.
\end{align}
Here $\theta$ is an eigenvector of the third Pauli matrix $\sigma^3$, whose eigenvalue determines the chirality of $\epsilon$, while $\eta(x^m)$ is a seven dimensional spinor, and $\rho$ is a constant Majorana spinor in 2+1 dimensions and represents the two preserved supercharges of the solution. Hence, from the viewpoint of four-dimensional $N=1$ supergravity, the solution is $\frac{1}{2}$--BPS.

The spinor $\eta(x^m)$ can be used to define a three-form
\begin{equation}
\varphi_{mnp}=-i\eta^\dagger\gamma_{mnp}\eta
\end{equation}
and a four-form
\begin{equation}
\Phi_{mnpq}=\eta^\dagger\gamma_{mnpq}\eta
\end{equation}
where $\gamma_{m\dots n}:=\gamma^m\dots\gamma^n$ is a product of seven dimensional Dirac matrices. The two forms $\varphi$ and $\Phi$ define a $G_2$-structure and are both Hodge dual to each other with respect to the metric $g_7=\d y^2+g_{uv}(x^m)\d x^u\d x^v$, that is, $\varphi=*_7 \Phi$. Therefore, this is the metric compatible with the so defined $G_2$-structure on $Y=\{y\} \ltimes X$ \cite{Joyce}.

Now, it can be shown that the first two supersymmetry conditions\footnote{Together with the requirement that the $H$-flux has only legs in the compact directions.} \eqref{eq:O1spinorsusy1} and \eqref{eq:O1spinorsusy2} are satisfied if and only if \cite{Lukas:2010mf, Gran:2005wf, Gauntlett:2002sc, 2003JGP....48....1F}
\begin{align}
\label{eq:killing1}
\d_7\varphi&=2\d_7\phi\wedge\varphi-*_7H 
\\
\d_7 *_7\varphi&=2\d_7\phi\wedge*_7\varphi
\\
\varphi\wedge H&=2*_7\d_7\phi\:,\\
\label{eq:killing4}
*_7\varphi\wedge H&=0 \:.
\end{align}
Here, $*_7$ is the seven-dimensional Hodge-dual with respect to the metric $g_7$ and and $\d_7=dx^m\p_m$ is the seven-dimensional  exterior derivative.

To focus on the compact space $X$, we will now decompose these equations by performing a $6+1$ split. The forms $\varphi$ and $\Phi$ can be written in terms of six dimensional forms~as
\begin{align}\label{phi:Hitchinflow}
\varphi&=-\d y\wedge\omega+\Psi_-\\
*_7\varphi&=\d y\wedge\Psi_++\frac{1}{2}\omega\wedge\omega\:,
\end{align}
where $\omega$ is a two-form and $\Psi=\Psi_++\I\,\Psi_-$ a complex three-form which, together, define an $SU(3)$-structure on $X$.
In terms of these forms, Eqs.~(\ref{eq:killing1})-(\ref{eq:killing4}), can be re-written~as
\begin{align}
\label{eq:killingspinor:1}
\d\Psi_-&=2\d\phi\wedge\Psi_-
\\\label{eq:killingspinor:2}
\d\omega&=2\p_y\phi\,\Psi_--\p_y\Psi_--2\d\phi\wedge\omega+*H
\\\label{eq:killingspinor:3}
\omega\wedge\d\omega&=\omega\wedge \omega\wedge\d\phi
\\\label{eq:killingspinor:4}
\d\Psi_+&=\omega\wedge\p_y\omega-\p_y\phi\, \omega\wedge\omega+2\d\phi\wedge\Psi_+
\\\label{eq:killingspinor:5}
\omega\wedge H&=*\d\phi
\\\label{eq:killingspinor:6}
\Psi_-\wedge H&=(2\p_y\phi)\;*1
\\\label{eq:killingspinor:7}
\Psi_+\wedge H&=0
\end{align}
where all symbols and forms are quantities on the six-dimensional compact internal space $X$. In particular, $*$ denotes the six-dimensional Hodge dual with respect to the metric $g_6=g_{uv}(x^m)\d x^u\d x^v$.

An $SU(3)$-structure can be characterised by the decomposition of the torsion tensor into irreducible $SU(3)$-representations, as reviewed in section \ref{subsec:Xgeom}. The structure decomposes into five torsion classes, which are related to the exterior derivatives of $\omega$ and $\Psi$ via \eqref{eq:tordom} and \eqref{eq:tordPsi}. Using these relations, it can be shown that the supersymmetry conditions \eqref{eq:killingspinor:1}-\eqref{eq:killingspinor:7} restrict the torsion classes to
\begin{align}\label{torsionclasses:genhf}
W_0^-=0 \qquad  W_2^-=0  \qquad W_1^\omega=\d\phi \qquad W_1^\Psi=2\d\phi\;
\end{align}
while the remaining classes are arbitrary. For the special case $H=0$, $d\phi=0$ this means that all but $W_0^+$ and $W_2^+$ vanish and such SU(3) structures are referred to as \emph{half-flat}. 
Without such a restriction, SU(3) structures satisfying \eqref{torsionclasses:genhf} are often referred to as {\it generalized half-flat}. In the next chapter, we shall also be concerned with geometries where all but one torsion class vanishes, namely $W_0$. Such geometries are referred to as {\it nearly K\"ahler}.

Recall that the Strominger system is characterized by the stronger conditions
\begin{align}\label{torsionclasses:strominger}
W_0=0 \qquad  W_2=0  \qquad W_1^\omega=\d\phi \qquad W_1^\Psi=2\d\phi\:.
\end{align}
Therefore, the Strominger system -- which results from a metric Ansatz with a maximally symmetric four-dimensional space-time -- is seen to be a special case of the more general Ansatz \eqref{domainwallansatz}, as one would have expected. Specializing \eqref{torsionclasses:strominger} further and setting $H=0$, $d\phi=0$ forces all torsion classes to vanish which corresponds to the case of Calabi-Yau manifolds times four-dimensional Minkowski space.

In addition to the above conditions which restrict the gravitational sector of the supergravity, there is also the instanton condition for the gauge field, coming from setting \eqref{eq:O1spinorsusy3} to zero. For a gauge field lying purely in the compact space $X$, this condition is equivalent to the conditions
\begin{align}
\label{eq:inst1}
F\wedge\Psi&=0\\
\label{eq:inst3}
\omega\lrcorner F&=0\:.
\end{align}
The last of these equations is referred to as the Yang-Mills condition. Solving these equations turns out to be a technical challenge in any heterotic compactification. For compactifications on Calabi-Yau manifolds, or Strominger-type compactifications, these are usually solved using the Donaldson-Uhlenbeck-Yau Theorem or its generalisation, Theorem \ref{tm:LiYau}, respectively. The geometries \eqref{torsionclasses:genhf} are in general not of Strominger type (and not even complex, since $W_0\not =0$ and $W_2\not = 0$) and, therefore, these theorems do not apply. However, explicit solutions to the instanton equations for Abelian gauge fields on homogeneous half-flat manifolds have been obtained in Ref.~\cite{Klaput:2011mz}. Taking into account the order $\alpha'$ backreaction of these gauge fields via the Bianchi identity is one of the main purposes of the next chapter.

Finally, we also need to make sure that we satisfy the heterotic Bianchi identity \eqref{eq:bianchi}, which we repeat here for convenience. At $\OO(\a)$, this reads
\begin{equation}
\label{eq:bianchi3d}
\d H=\frac{\a}{4}(\tr\:F\wedge F-\tr\:R^-\wedge R^-)+\OO(\a^2)\:,
\end{equation}
where $R^-$ is the curvature of the Hull connection, required under the supersymmetry transformations \eqref{eq:O1spinorsusy1}-\eqref{eq:O1spinorsusy3}. Moreover, we also need to satisfy the equations of motion \eqref{eq:eom1}-\eqref{eq:eom4}. This is however trivial, as the theorem by Ivanov \cite{Ivanov:2009rh} (see also \cite{Martelli:2010jx}) guarantees that the equations of motion are satisfied to the order we require, as long as the connection appearing in the Bianchi identity and action is an instanton
\begin{align}
\label{eq:instR3d}
R^-_{mn}\Gamma^{mn}\eta=\OO(\a)\:.
\end{align}
As is shown in appendix \ref{app:Hull}, the Hull connection is an instanton for supersymmetric solutions, to the order we are concerned with.

\section{The Four-Dimensional Theory}
\label{sec:4ddimred}
Having discussed supersymmetric compactifications of the heterotic string on domain-wall geometries, we would like to review some of the details of the dimensional reduction of this theory. As noted above, the geometries we encounter will in general have non-trivial torsion, generalising the usual Calabi-Yau compactifications. We therefore need a framework where this is taken into account. The framework we will employ is that of half-flat mirror manifolds \cite{Gurrieri:2002wz}, which is general enough for our purposes, and which reduces to the usual Calabi-Yau reduction for the torsion-free story. 

\subsection{Half-flat Mirror Geometry}
\label{sec:mirrorgeom}

We now take a moment to review a convenient language in which to formulate the more general compactifications discussed above, namely that of {\it half-flat mirror manifolds}. This will also be convenient when we later discuss the dimensional reduction of the theory. The language is analogous to Calabi-Yau manifolds, and indeed generalises this setting. This language also applies to the explicit examples of nearly K\"ahler coset considered in the next chapter \cite{Klaput:2011mz}, and is hence the most generic framework we need.

Half-flat mirror manifolds were introduced in Refs.~\cite{Gurrieri:2002wz, Gurrieri:2004dt, Gurrieri:2007jg} in the context of type II mirror symmetry with NS fluxes. Specifically, they are mirrors of Calabi-Yau compactifications of type II compactifications with electric NS flux, which under the mirror map turns into torsion of the compact manifold. These manifolds are equipped with a set, $\{\omega_i\}$, of two-forms, and a dual set, $\{\tilde\omega^i\}$, of four forms. They also have a symplectic set, $\{\alpha_A,\beta^B\}$, of three-forms, as in the Calabi-Yau case. These forms satisfy the following relations
\begin{eqnarray}\label{eq:mirrorgeombasisrelations1}
\int_X\omega_i\wedge\tilde\omega^j=\delta_i^j,\;\;\int_X\alpha_A\wedge\alpha_B=0,\;\;\int_X\beta^A\wedge\beta^B=0,\;\;\int_X\alpha_A\wedge\beta^B=\delta_A^B,
\end{eqnarray}
similar to the harmonic and symplectic basis forms on a Calabi-Yau manifold. Furthermore, we define intersection numbers $d_{ijk}$ analogous to the Calabi-Yau case by writing
\begin{equation}\label{eq:intersectionnumbers}
\omega_i\wedge\omega_j\equiv d_{ijk}\,\tilde\omega^k\,,
\end{equation}
for the half-flat mirror two-forms $\omega^i$. In contrast to Calabi-Yau manifolds, however, these forms are not harmonic anymore in general. Instead, the non-closed forms satisfy the differential relations
\begin{equation}\label{eq:mirrorgeombasisrelations2}
\d\omega_i =e_i\beta^0\:,\quad \d\alpha_0=e_i\tilde\omega^i\:.
\end{equation}
The coefficients $e_i$ are constants on $X$ and parametrize the intrinsic torsion of the manifolds. The $SU(3)$-structure forms $\omega$ and $\Psi$ can be expanded in this basis
\begin{equation}
\label{eq:exphf}
\omega=v^i\omega_i\; ,\quad\Psi =Z^A\alpha_A+\I\, \,\G_A\,\beta^A\; ,
\end{equation}
where the fields $v^i$ are analogous to the K\"ahler moduli, the $Z^A$ analogous to the complex structure moduli and $\G_A$ are derivatives of the pre-potential $\G(Z^A)$, $\G_A=\p_A\G$. Taking the exterior derivative we get
\begin{equation}\label{eq:mirrorderivative}
\d\omega=v^ie_i\beta^0\; ,\quad\d\Psi =Z^0e_i\tilde\omega^i\:.
\end{equation}
By comparing with Eqs.~\eqref{eq:tordom} and \eqref{eq:tordPsi}, these results can be used to read off the torsion classes of half-flat mirror manifolds. In particular, we see that the constants $e_i$ indeed measure the intrinsic torsion of the manifolds.

\subsection{Reductions on Half-Flat Mirror Manifolds}
\label{sec:4dred}
We would now like to perform the dimensional reduction of the heterotic supergravity on such half-flat mirrors. We shall use this in the following chapters when considering the corresponding four-dimensional theories. We will not go through the full reduction in detail, but summarise the results. The procedure is very similar to the usual case of Calabi-Yau reductions, and can be found in \cite{Gurrieri:2004dt, Gurrieri:2007jg}. As stated above, we omit the gauge-bundle moduli, as these overcomplicate the situation and are besides the points we wish to make in this part of the thesis. To include these, knowledge of the full reduction of $N=1$ heterotic supergravity is needed, which is beyond the scope of this thesis.

We begin by explaining the relation between four- and ten-dimensional quantities, following the conventions of four-dimensional supergravity laid out in \cite{wess1992supersymmetry}.
A set of fields, $v^i$, analogous to the K\"ahler moduli of CY manifolds, appear in the expansion
\begin{eqnarray}\label{eq:mirrorexpansion:J}
\omega=v^i\omega_i\:.
\end{eqnarray}
of the $SU(3)$-structure form $\omega$ with respect to the two-forms $\omega_i$ of the half-flat mirror basis introduced in section \ref{sec:mirrorgeom}. We also introduce the standard quantity
\begin{equation}
\label{eq:stVol}
 \V=\frac{1}{6}d_{ijk}v^iv^jv^k\; ,
\end{equation} 
proportional to the volume of the compact space. This allows us to define the four-dimensional dilaton $s$ in terms of its ten-dimensional counterpart $\phi$ as
\begin{eqnarray*}
s=e^{-2\phi}\frac{\V}{\vol}\:,
\end{eqnarray*}
where $\vol$ is some reference volume. Next, we expand the $B$-field as 
\begin{equation*}
B=b^i\omega_i+B_{(4)}\:,
\end{equation*}
where $B_{(4)}$ is a two-form in four dimensions. The field strength has the following expansion
\begin{equation}
\label{eq:10dfstrength}
H=\d B+\epsilon_A\beta^A+\mu^B\alpha_B=\d b^i\wedge\omega_i+b^ie_i\beta^0+\d B_{(4)}+\epsilon_A\beta^A+\mu^B\alpha_B\:,
\end{equation}
where we have also included ``explicit flux" terms, parameterised by $\{\epsilon_A,\mu^B\}$. These terms correspond to a $B$-field that is only locally well-defined, and are not included in the above expansion of $B$. $\d$ in the above expression refers to the ten-dimensional exterior derivative in general. Note that $\d H\neq 0$, by the presence of $\alpha_0$. This is needed to parameterise the non-closeness of $H$ at this order. We will use this in the next chapter to parametrise the non-trivial Bianchi identity at $\OO(\a)$.

Excluding bundle moduli, we are interested in a set of chiral fields, $\Phi^{X}=\{S, T^i, X^A\}$, where
\begin{equation*}
S=a+is\:,\;\;\;\;T^i=b^i+iv^i\:,\;\;\;\;X^A=c^A+iw^A\:,
\end{equation*}
where $X^A=Z^A/\F$ are the four-dimensional fields corresponding to the complex structure moduli, and $a$ is the axion, dual to $B_{(4)}$. $\F$ here is some appropriate normalisation factor of the projective coordinates $Z^A$, not to be confused with the Atiyah class of section \ref{subsec:defsV}, and appropriately chosen for the matching of the ten-dimensional and four-dimensional theories. Typically, we have $\F=Z^0$, but as we will see in chapter \ref{ch:CY}, this can be generalised where appropriate.

The K\"ahler potential $K=K(\Phi^X,\bar{\Phi}^{\bar{X}})$ is then given by
\begin{equation}
\label{eq:kahlerpotfull}
K=-\log\Big(i(\bar S-S)\Big)-\log(8\V)-\log\Big[\frac{i}{\F^2}\:\int_X\Psi\wedge\bar\Psi\Big]\:.
\end{equation}
The theory also has a corresponding Gukov-Vafa-Witten superpotential $W=W(\Phi^X)$, given by \cite{Gukov:1999ya,Gurrieri:2004dt, CyrilThesis}
\begin{eqnarray}
\label{eq:gukovvafa}
W=\frac{\sqrt{8}}{\F}\int_{X}(H+i\:\d\omega)\wedge\Psi\:,
\end{eqnarray}
normalised by $\F$, and where the factor $\sqrt{8}$ is conventional. Note the appearance of the torsion term $\d\omega$ in \eqref{eq:gukovvafa}, which can potentially be used to stabilise moduli.

The scalar potential is then given by the usual formula
\begin{eqnarray}
V=e^{K}\Big(K^{X\bar Y}F_X\bar{F}_{\bar{Y}}-3\vert W\vert^2\Big)+\frac{1}{2}\DD_a \DD^a\; ,
\end{eqnarray}
where the F-terms are defined as $F_X=D_XW=\p_XW+K_XW$, with $K_X=\p_XK$. Further, $K_{X\bar Y}\equiv\p_X\p_{\bar Y}K$ is the K\"ahler metric, $K^{X\bar{Y}}$ is its inverse, and $D_a$ are the D-terms, which can originate from internal line-bundles of the ten-dimensional theory.


Let $\chi^X$ be the superpartner of $\Phi^X$. We write the four-dimensional gravitino as $\psi_\mu$, where greek letters denote four-dimensional indices. The supersymmetry transformations of the four-dimensional supergravity theory then read
\begin{align}
\label{eq:4dsusyvar1}
\delta\chi^X&=i\sqrt{2}\sigma^\mu\bar\kappa\p_\mu\Phi^X-\sqrt{2}e^{K/2}K^{X\bar Y}\bar F_{\bar Y}\kappa\\
\label{eq:4dsusyvar2}
\delta\psi_\mu&=2\D_\mu\kappa+ie^{K/2}W\sigma_\mu\bar\kappa\:,
\end{align}
where the Weil spinor $\kappa$ parametrises four-dimensional supersymmetry, and $\sigma^\mu=\{I_2,\sigma^\alpha\}$, where $\sigma^\alpha$ are Pauli matrices. Here the covariant derivative $\D_\mu$ is given by
\begin{equation*}
\D_\mu=\p_\mu+\omega_\mu+\frac{1}{4}\Big(K_X\p_\mu\Phi^X-K_{\bar X}\p_\mu\bar\Phi^{\bar X}\Big)\:,
\end{equation*}
where $\omega_\mu$ is the spin-connection.

Supersymmetry requires that we set the variations \eqref{eq:4dsusyvar1}-\eqref{eq:4dsusyvar2} to zero. Maximally symmetric four-dimensional supersymmetric solutions then requires
\begin{equation*}
W=\p_XW=0\:,
\end{equation*}
from \eqref{eq:4dsusyvar1}-\eqref{eq:4dsusyvar2}. These are usually referred to as the F-term conditions. Performing a reduction on generic torsional half-flat spaces with fluxes, these equations appear hard to satisfy. Indeed, from what we have learned already, this should not be possible without including non-perturbative effects or reducing the system to that of the Strominger system and a heterotic $SU(3)$-structure. Domain-wall solutions are however allowed, and the most generic four-dimensional domain wall solution, with $M_3$ flat space, was given in \cite{Lukas:2010mf}, which we briefly recall here.

\subsection{Four-Dimensional Domain Walls}
\label{sec:4dDW}
The four-dimensional metric is assumed to have the form
\begin{equation*}
\d s_4^2=e^{-2B(y)}(\d y^2+\eta_{\alpha\beta}\d x^\alpha\d x^\beta)\:,
\end{equation*}
where $\{\alpha,\beta\}\in\{0,1,2\}$, $\eta_{\alpha\beta}$ is the three-dimensional Minkowski metric, and $B(y)$ is some warp-factor.

Under this assumption, the general equations for supersymmetric solutions then read
\begin{align}
\label{eq:4dsusyvarDM1}
\p_y\Phi^X&=-ie^{-B}e^{K/2}K^{X\bar Y}\bar F_{\bar Y}\\
\label{eq:4dsusyvarDM2}
\p_yB&=ie^{-B}e^{K/2}W\\
\label{eq:4dsusyvarDM3}
\Im(K_X\p_y\Phi^X)&=0\\
\label{eq:4dsusyvarDM4}
2\p_y\kappa&=-\p_y B\kappa\:.
\end{align}
In the case of a axio-dilaton independent superpotential, one finds that the second equation is exactly the flow equation of the four-dimensional dilaton $\phi_4$, where $s=e^{-2\phi_4}$, but with $B$ replaced with $\phi_4$. In this case we can identify the warp factor $B$ with $\phi_4$. 

Furthermore, the spinor $\kappa$ also satisfies the constraint
\begin{equation*}
\kappa=\sigma^2\bar\kappa\:,
\end{equation*}
reducing its number of independent components to two. These solutions are therefore known as half-BPS. A study of such solutions at zeroth order in $\a$ was performed in \cite{Lukas:2010mf}, where a matching was found between the ten-dimensional equations \eqref{eq:killingspinor:1}-\eqref{eq:killingspinor:7} and the four-dimensional equations \eqref{eq:4dsusyvarDM1}-\eqref{eq:4dsusyvarDM4} in the large complex structure limit.

In the next chapter, we will go one step further and include $\a$-effects. We will consider the explicit example of coset spaces, and show that by including these effects it is possible to stabilise all geometric moduli of the theory. The perturbative solution is still a domain wall, as the dilaton has a non-trivial profile in this case. We find that by including non-perturbative effects to the four-dimensional theory, like e.g. a gaugino condensate, it is possible to lift this vacuum to a maximally symmetric one.

In chapter \ref{ch:CY} we will also consider domain wall solutions on Calabi-Yau's with fluxes. We match the ten-dimensional and four-dimensional killing-spinor equations at a generic point in moduli space, extending the analysis of \cite{Lukas:2010mf}. We conclude that fluxes can be used to stabilise moduli, even in Calabi-Yau compactifications, provided one sacrifices a maximally symmetric space-time, at least perturbatively.

\chapter{Moduli Stabilisation in Half-Flat Compactifications}
\label{ch:HFNK}


In this chapter we will study heterotic domain wall compactifications on torsional half-flat manifolds, with particular emphasis on the inclusion of $\alpha'$ corrections and moduli stabilisation. As we shall see, the zeroth order solution is not sufficient to stabilise the geometric moduli of the internal space $X_6$, as has been pointed out before \cite{Klaput:2011mz, CyrilThesis}. This changes when including $\a$-effects, and we will see that flux induced through the heterotic Bianchi identity together with torsion of the internal space can indeed be used to stabilise all geometric moduli, leaving a $y$-dependent dilaton. Next, we will address the question as to whether the domain wall can be ``lifted" to a maximally symmetric vacuum via stabilization of all moduli. For the examples studied the answer is a cautious ``yes". We will see how a combination of $\alpha'$ and non-perturbative effects can indeed lift the runaway directions of the original, lowest-order perturbative potential and lead to a supersymmetric AdS vacuum in a heterotic KKLT-type scenario. The chapter is based on \cite{Klaput:2012vv}.

\section{Introduction}
An explicit study of $\alpha'$ corrections and the required construction of gauge fields requires an explicit and accessible set of half-flat manifolds. For this reason, we will focus on cosets which admit half-flat structures and, specifically, on $SU(3)/\U{1}^2$ which provides the greatest flexibility among the cosets for building gauge fields via the associated bundle construction. As noted in the introduction, cosets have nice geometrical descriptions in terms of $G$-invariant forms, where a lot of calculations can be done explicitly. The $G$-invariant construction also respects the half-flat mirror manifold description of section \ref{sec:mirrorgeom}, which is useful when we come to dimensionally reduce the theory.

Following Ref.~\cite{Klaput:2011mz}, we construct explicit gauge bundles consisting of sums of line bundles. The conditions for these gauge fields to be supersymmetric, the D-term conditions from a four-dimensional point of view, fix two of the three T-moduli, thereby restricting the half-flat structure to be nearly K\"ahler. We will see that the anomaly condition can be satisfied for appropriate bundle choices and we solve the Bianchi identity explicitly for such choices. This results in a non-harmonic H-flux, induced by the bundle flux, which leads to a correction to the metric and the dilaton profile at order $\alpha'$. These corrections preserve the nearly K\"ahler structure on the coset space. 

From a four-dimensional point of view, the bundle-induced H-flux leads to an additional, constant term in the superpotential. This term can stabilise the remaining T-modulus, but the dilaton is still left as a runaway direction. We will see that upon inclusion of non-perturbative effects, such as a gaugino condensate, all moduli can be stabilised in a supersymmetric AdS vacuum. For appropriate bundle choices this stabilisation does arise in a consistent part of moduli space, that is, at weak coupling and for moderately large internal volume. However, as we shall see there is a tension in that it is not possible, for the specific examples analysed, to make the volume very large (so that there is no doubt about the validity of the $\alpha'$ expansion) and keep the theory at weak coupling. 

The chapter is organised as follows. In section \ref{sec:GeoPre} we begin with a review of coset geometry, and we review the lowest order solution in section \ref{section:lowestorder}, following Ref.~\cite{Klaput:2011mz}. Next in section \ref{sec:alphaprimecorr} we consider what $\alpha'$-corrections do to the solution, and we compare this to the four-dimensional supergravity in section \ref{sec:NK4d}. We review and discuss future directions in section \ref{sec:Disc}.

\section{Geometric Preliminaries}
\label{sec:GeoPre}


Before we move on to construct explicit solutions, we review some of the coset geometry needed for this chapter. We have tried to make it a short but self-consistent review for the reader, leaving some of the technical details to appendix \ref{app:coset}. Extensive reviews of coset geometry and the $G$-invariant formalism we employ have appeared in the literature before, and the reader is referred to \cite{1998math.ph...7026O, castellani2001g, Kapetanakis19924, 0264-9381-5-1-011, Roberto19901, Castellani1984394, Lust:1985be, Lust:1986ix, KashaniPoor:2007tr, Castellani:1999fz} for more details.

Of the known nearly K\"ahler homogeneous spaces $SU(3)/\U{1}^2$, $\Sp{2}/\SU{2}\times \U{1}$, $\Gtwo/SU(3)$ and $\SU{2}\times\SU{2}$, only the first two spaces allow for line bundles using the construction method we employ. A study of the expected number of generations, using the index of the Dirac operator, 
shows that only $SU(3)/\U{1}^2$ admits bundles with three generations \cite{Lust:1986ix, Klaput:2011mz}. Hence, in our analysis we will focus on the case $SU(3)/\U{1}^2$, even though many of the results can be extended in a straightforward way to include all four spaces \cite{Klaput:2012vv}. 

We will start with a brief review of coset geometry, the construction of $SU(3)$-structures on cosets and the relation to half-flat mirror geometry. Then, we discuss the construction of vector bundles and, in particular, line bundles on cosets. In the following sections we will combine these ingredients with a four-dimensional domain wall, and we construct ten-dimensional solutions with two supercharges to the heterotic string.

\subsection{$SU(3)$-Structures on the Coset}\label{section:coset}
We begin with a review of coset space differential geometry and, in particular, the construction of the corresponding $G$-invariant half-flat $SU(3)$-structures. 
A coset space $G/H$ is obtained by identifying all elements of the Lie group manifold $G$ which are related by the action of the subgroup $H \subset G$. For the construction of bundles on $G/H$ later on, it will be useful to view $G$ as a principal fibre bundle over $G/H$ with fibre $H$, that is, $G=G(G/H,H)$. The base space $G/H$ admits a natural frame of vielbeins, which descend from the left-invariant Maurer-Cartan forms on $G$ and will be denoted by $e^u$ \cite{Castellani:1999fz}. These one-forms are, in general, no longer left-invariant under the action of $G$. However, in the cases of interest, there exist $G$-(left)-invariant two-, three- and four-forms. 

The space of $G$-invariant two- and three-forms for $SU(3)/\U{1}^2$ is spanned  by
\begin{align*}
\{\, e^{12}\, ,\; e^{34}\, ,\; e^{56}\,\}\:,\quad \{\, e^{136}-e^{145}+e^{235}+e^{246}\,,\;  e^{135}+e^{146}-e^{236}+e^{245}\,\}\:,
\end{align*}
where $e^{u_1\dots u_n}:=e^{u_1} \wedge \dots \wedge e^{u_n}$. The $G$-invariant four-forms which can be obtained from the above $G$-invariant two-forms via Hodge duality can be found in appendix \ref{app:coset}.

Requiring the $SU(3)$-structure to be compatible with the given group structure of the coset implies that the forms $\omega$ and $\Psi$ can be expressed in terms of the above forms. Indeed, one finds that the most general $G$-invariant $SU(3)$-structure forms for $SU(3)/\U{1}^2$ are given by
\begin{equation}
\label{eq:zerothsolutionsu3}
\begin{array}{lll}
\omega&=&R_1^2e^{12}-R_2^2e^{34}+R_3^2e^{56}\\
\Psi&=&R_1R_2R_3\Big[\left(e^{136}-e^{145}+e^{235}+e^{246})+\I\, (e^{135}+e^{146}-e^{236}+e^{245}\right)\Big]\:,
\end{array}
\end{equation}
with independent parameters $R_1$, $R_2$ and $R_3$. It should be noted that the form of $\omega$ given above is {\it negative definite}, and so differs from the usual hermitian form by a change of sign. With this convention, the compatibility relation \eqref{eq:volume} reads\footnote{In this, and the following chapter, we use the convenient normalisation $\vert\vert\Psi\vert\vert^2=8$.}
\begin{equation}
\label{eq:volume3d}
{\rm dvol}_X=-\frac{1}{6}\omega\wedge\omega\wedge\omega=\frac{i}{8}\Psi\wedge\bar\Psi\:.
\end{equation}
This change of convention is historical for coset geometry, and will be employed throughout this chapter.

From the above $SU(3)$-structure forms we can construct a unique compatible metric \cite{2001math......7101H}, which coincides with the most general $G$-invariant metric on $G/H$. The metric is given by
\begin{align}\label{def:ginvmetric}
\d s^2
&=
R_1^2\,(e^1\otimes e^1+e^2\otimes e^2)
+
R_2^2\,(e^3\otimes e^3+e^4\otimes e^4)
+
R_3^2\,(e^5\otimes e^5+e^6\otimes e^6)\;
\end{align}
We recognise the parameters $R_i$ as \quot{radii} of the coset, determining the volume and shape of the space. We shall later see that they are related to K\"ahler moduli.

Having introduced $G$-invariant geometry and $SU(3)$-structure on our coset, all required tools to solve the geometric sector of the heterotic string, that is, the Killing spinor equations \eqref{eq:killingspinor:1}-\eqref{eq:killingspinor:7}, are available. This has been known for some time and was first realised in Ref.~\cite{Lust:1986ix}. The additional technical difficulty of \emph{heterotic} string compactifications is the construction of vector bundles which satisfy the instanton equations \eqref{eq:inst1}, \eqref{eq:inst3}. In past works, this has usually been approached using an Ansatz similar to the standard embedding. We will adopt the bundle construction developed in Ref.~\cite{Klaput:2011mz} which contains the standard embedding Ansatz as special case. 

\subsection{Half-Flat Mirror Geometry of the Coset}\label{sec:mirrorgeom2}


We recall from section \ref{sec:mirrorgeom} that half-flat mirror geometry, in analogy with Calabi-Yau manifolds, is defined by a set of two-forms, $\{\omega_i\}$, a set of dual four-forms, $\{\tilde{\omega}^i\}$, and a set $\{\alpha_A,\beta^B\}$ of symplectic three-forms. Unlike in the Calabi-Yau case, these forms are, in general, no longer closed but instead satisfy a set of differential relations~\eqref{eq:mirrorgeombasisrelations2} which involve the torsion parameters $e_i$.

It turns out that for the coset under consideration, there is only a single pair, $\{\alpha_0,\beta^0\}$, of symplectic three-forms in addition to a certain number of two- and four-form pairs, $\{\omega_i,\tilde{\omega}^i\}$. A subset,$\{\omega_r\}$ of the two-forms which we label by indices $r,s,\ldots$ are, in fact, closed. For $SU(3)/U(1)^2$ these forms are given in appendix \ref{app:coset}, equation \eqref{eq:Ginvforms}. In particular, there are three pairs of two- and four-forms in this case. The exterior derivatives of $\omega_3$ and $\alpha_0$ are given by $d\omega_3=\beta^0$ and $d\alpha_0=\tilde{\omega}^3$, while all other forms are closed. This means the closed two-forms are $\omega_r$, where $r=1,2$. Comparing with the general differential relations~\eqref{eq:mirrorgeombasisrelations2} for half-flat mirror geometry this shows that the three torsion parameters are given by $(e_1,e_2,e_3)=(0,0,1)$.

It can be shown that the above forms indeed satisfy all the relations for half-flat mirror geometry given in section \ref{sec:mirrorgeom}. In particular, the $SU(3)$-structure forms on the coset given in the previous subsection can be re-written in half-flat mirror form as
\begin{equation}
\omega= v^i\omega_i\:,\quad \Psi = Z\,\alpha_0+\I\, G\,\beta^0\:, \label{JO}
\end{equation}
where $Z$ is the single ``complex structure" modulus and $G$ the derivative of the pre-potential. From \eqref{eq:ZG}, we see that these two quantities are related by
\begin{align*}
Z = \frac{\vol}{\pi^2} G\:,
\end{align*}
where $\vol$ is the coordinate volume, a specific number whose value differs depending on the coset. For $SU(3)/U(1)^2$ it is given in appendix~\ref{app:coset}, and reads $\V_0=4(2\pi)^3$. 

It is also easy to verify from the above expressions for the forms that
\begin{equation}
\label{eq:mirrorsu3}
 \omega_i\wedge\alpha_0=\omega_i\wedge\beta_0=0
\end{equation}
for all $i$, in analogy with the Calabi-Yau case. These relations are also expected from the absence of $G$-invariant 5-forms on our coset space. A further useful relation can be deduced from the $SU(3)$-structure compatibility relation~\eqref{eq:volume3d}. Inserting the expansions~\eqref{JO} for $\omega$ and $\Psi$ into this relation leads to
\begin{equation}
\label{compzeroth}
d_{ijk} v^i v^j v^k=-\frac{3}{2} ZG=-\frac{3\,\vol}{2 \pi^2}\, G^2\:.
\end{equation}
This shows that $Z$ is determined by the ``K\"ahler moduli'' $v^i$, and hence no independent ``complex structure" moduli exist in our coset model.

\subsection{Levi-Civita Connection}
The Levi-Civita connection is the unique torsion-free and metric compatible connection on the tangent bundle. On our space, with the most general $G$-invariant metric \eqref{def:ginvmetric}, the Levi-Civita connection one-form is
\begin{align}
\label{def:levicevita}
{{\omega^{LC}}_u}^v=\frac{1}{2} {f_{wu}}^v e^w+{f_{iu}}^v\varepsilon^i\:.
\end{align}
The $\varepsilon^i$ are the Maurer-Cartan left-invariant one-forms on $G$ along the directions of the generators $H_i$ of the sub-group $H$. On $G/H$ these can be written in terms of the forms $e^u$, but the explicit expressions are not required.

The Levi-Civita connection enters the Bianchi identity \eqref{eq:bianchi3d} as part of the connection one-form $\omega^-$ defined in \eqref{eq:Hullconn}. As we will see below, our spaces do not allow for H-flux at lowest order in $\alpha'$, and we can therefore set $R^- = R^{LC}$ to first order in the Bianchi identity. For $SU(3)/\U{1}^2$ this means that the contribution to the Bianchi identity is given as
\begin{align}
\label{eq:trRRLeviCevita}
\tr\, R^{LC}\wedge R^{LC} = -\frac{9}{4}\,\frac{\vol}{\pi}\;\tilde{\omega}^3\:.
\end{align}

\subsection{Vector Bundles on the coset}
We now turn to the problem of finding appropriate gauge bundles on the coset, which can satisfy the instanton equations \eqref{eq:inst1}, \eqref{eq:inst3}. Such bundles have been explicitly constructed in \cite{Klaput:2011mz}, based on the well-known relation between vector bundles and principal fibre bundle. The principal fibre bundle in our case is $G=G(G/H,H)$ and any representation $\rho\,:\,H\to\;\mathbb{C}^n$ uniquely defines a rank $n$ vector bundle which is referred to as an {\em associated} vector bundle. Moreover, any connection defined on $G$ uniquely defines a connection on every associated vector bundle. We shall require the structure of the bundle to be compatible with the group structure of $G/H$. This leads to a natural connection on $G=G(G/H,H)$, related to the reductive decomposition of the Lie algebra, given by
\begin{eqnarray*}
A=\varepsilon^i H_i\:,
\end{eqnarray*}
where $H_i$ are the generators of the Lie algebra of $H$ and the $\varepsilon^i$ are the Maurer-Cartan left-invariant one-forms on $G$ along the directions of the generators $H_i$. As before, their explicit form in terms of the vielbein $e^u$ will not be required.

On an associated vector bundle defined by the representation $\rho$, the connection associated to $A$ is then
\begin{eqnarray*}
A_\rho=\varepsilon^i\rho(H_i)\;
\end{eqnarray*}
with field strength
\begin{eqnarray}
\label{eq:random23523}
F=-\frac{1}{2}{f_{uv}}^i\rho(H_i)e^u\wedge e^v\:.
\end{eqnarray}
Note that the one-forms $\varepsilon^i$ have dropped out. This construction holds in general for every representation $\rho$ of $H$. 

We would like to add a few remarks on the \quot{standard embedding}, a choice of gauge connection frequently made in the literature. For this choice, the bundle curvature $F$ and the Riemann curvature $R^{LC}$ are set equal, which solves the Bianchi identity \eqref{eq:bianchi3d} for $H=0$.\footnote{It solves the Bianchi identity for $\d H=0$, but as we shall see by Theorem \ref{tm:noflux}, it follows that $H=0$.} However, in the present context, such a choice leads to a problem. Since our spaces are not Ricci-flat, the so-chosen field strength $F$ does not satisfy the instanton equations, so that the solution is not supersymmetric. If we instead embed the gauge connection in the tangent bundle connection given by
\begin{align}
\label{eq:somechoice}
{\rho(H_i)_u}^v={f_{iu}}^v\:,
\end{align}
then the curvature \eqref{eq:random23523} satisfies the instanton equations.
Note that \eqref{eq:somechoice} does not solve the Bianchi identity for $H=0$ anymore. However, since this connection only differs from the Levi-Civita connection \eqref{def:levicevita} by a torsion term, both choices yield the same cohomology class for $\tr F\wedge F$ and $\tr R^{LC} \wedge R^{LC}$. This means that the topological constraint arising from the Bianchi identity is satisfied, while the exact identity is only satisfied to lowest order in $\alpha'$. This has been the case for most heterotic bundle constructions in past works. In contrast, we will construct $\OO(\a)$ solutions to the the Bianchi identity, and to the supersymmetry constraints.

\subsection{Line Bundle Sums}
When constructing a solution to the $E_8 \times E_8$ heterotic string, the structure group of a vector bundle has to be embedded in $E_8$ and the resulting low energy gauge group will be given by the commutant of the structure group within $E_8$. Recently, it has been noted that vector bundles which consist of sums of line bundles provide a fertile class of models which can be studied systematically~\cite{Anderson:2011ns, Anderson:2012yf}. Such line bundle sums have been used for the half-flat compactifications in Ref.~\cite{Klaput:2011mz}, and will also be the focus of the present chapter. 

Let us first focus on a single line bundle, $L$, defined by a one-dimensional representation $\rho\,:\,H\to\;\mathbb{C}$. For $SU(3)/\U{1}^2$, such a representation is characterized by two integers, $p^r$, where $r=1,2$, representing the first Chern class of the bundle. Writing
\begin{align*}
\rho(H_7) = -\I\,(p^1+p^2/2)\qquad \rho(H_8)=-\I\, p^2/(2\sqrt{3})\;,
\end{align*}
and using equation~\eqref{eq:random23523} the first Chern class of such a line bundle becomes
\begin{equation*}
 c_1(L)=\frac{\I}{2\pi}[F]=p^r [\omega_r]\:.
\end{equation*}
Hence, the integers ${\bf p}=(p^r)$ indeed label the first Chern class of the line bundles and we can adopt the notation $L={\cal O}_X({\bf p})$. A sum of line bundles
\begin{eqnarray*}
V=\bigoplus_{a=1}^n {\cal O}_X({\bf p}_a)
\end{eqnarray*}
is, therefore, characterized by the set, $\{p_a^r\}$, of integers and its total first Chern class is given by
\begin{equation*}
 c_1(V)=\sum_{a=1}^np_a^r [\omega_r]\:.
\end{equation*}
Given that there are no $G$-invariant exact two-forms on our spaces, it follows that the field strength for the connection on $V$ is given by
\begin{eqnarray}
F=[F]=-2\pi i\sum_ap^r_a\omega_r\:.
\label{eq:F}
\end{eqnarray}

To ensure that the structure group of $V$ can be embedded into $E_8$, we impose the vanishing of the first Chern class, $c_1(V)=0$. This condition restricts the integers $p_a^r$ by
\begin{equation*}
 \sum_{a=1}^np_a^r=0\qquad \forall\,r\:.
\end{equation*}
Then, the structure group of $V$ is $S(U(1)^n)$ which is indeed a sub-group of $E_8$ for $1<n\leq 8$. Further, for $n=3,4,5$, the commutant of $S(U(1)^n)$ within $E_8$ is given by $S(U(1)^3)\times E_6$, $S(U(1)^4)\times SO(10)$, and $S(U(1)^5)\times SU(5)$, respectively. These are the phenomenologically interesting GUT gauge groups, and for the ``visible" $E_8$ we will therefore focus on line bundle sums of rank $3$, $4$ or $5$. 

Subsequently, we will require the vector bundle contribution to the Bianchi identity~\eqref{eq:bianchi3d}. We evaluate this contribution for a sum of line bundles on $SU(3)/\U{1}^2$. Writing $(p_a,q_a)=(p_a^1,p_a^2)$ for ease of notation, we find
\begin{align}
\tr\, F \wedge F &=- \frac{\vol}{8\pi}\Big[\sum_a(6 p_a^2 + q_a^2+6p_a q_a)\tilde{\omega}^1\notag\\
&+\sum_ap_a(3 p_a +2 q_a)\tilde{\omega}^2+\frac{4}{3}\sum_a(3p_a^2 + q_a^2 + 3p_a q_a)\tilde{\omega}^3\Big]
\label{eq:trFF}
\end{align}
Note that we will, of course, have two different bundles, one for each $E_8$ factor, corresponding to the visible and hidden sectors of the theory. Hence, the Bianchi identity has two contributions of the form~\eqref{eq:trFF}, each controlled by its own set of integers. 

The Bianchi identity also gives an additional integrability condition on the bundle. We require the Bianchi identity to be satisfied in cohomology, giving 
\begin{equation*}
p_1(V)=p_1(X)=0\:,
\end{equation*}
where we have used that $p_1(X)=0$, which is true for the coset under consideration. This leads to relations between the observable bundle parameters $\{p_a,q_a\}$ and hidden bundle parameters $\{\tilde p_a,\tilde q_a\}$, which can be written as
\begin{align}
\label{eq:constraintpq1}
\sum\limits_{a=1}^{n}(6p_a^2 + q_a^2+6p_a q_a)+\sum\limits_{a=1}^{\tilde{n}}(6\tilde p_a^2 + \tilde q_a^2+6\tilde p_a \tilde q_a)&=0
\\\label{eq:constraintpq2}
\sum\limits_{a=1}^{n}p_a(3p_a +2q_a)+\sum\limits_{a=1}^{\tilde{n}}\tilde{p}_a(3\tilde p_a +2\tilde q_a)&=0\:.
\end{align}
Solutions to these equations exist, and explicit examples will be considered later. Note that the presence of the hidden bundle is helpful in that is can be used to cancel the observable bundle contributions which may be somewhat constrained by model building considerations. 

Another basic phenomenological requirement on the visible vector bundle is the presence of three chiral generations. The number of generations is given by the chiral asymmetry, which is counted by the index of the Dirac operator. This can be computed using the Atiyah-Singer index theorem. For a sum of line bundles, this leads to \cite{Nash:1991pb, Kimura:2007xa, Klaput:2011mz}
\begin{equation*}
 {\rm ind}(V)=-\frac{1}{6}d_{rst}\sum_{a=1}^np_a^rp_a^sp_a^t\:, 
\end{equation*}
where $d_{ijk}$ are the intersection numbers. Specializing to $SU(3)/\U{1}^2$ gives
\begin{equation*}
{\rm ind}(V)=-\sum\limits_{a=1}^{n}\left(p_a^3+\frac{1}{2}p_aq_a(q_a+3p_a)\right)\:. 
\end{equation*}

\section{Solutions to Lowest Order in $\alpha'$}\label{section:lowestorder}
We have now collected all ingredients to solve the heterotic string on our coset. In this section we will review the solution at lowest order in $\alpha'$ which has been found in Ref.~\cite{Klaput:2011mz}.

As discussed in Chapter \ref{ch:nonmaxcomp}, finding a supersymmetric vacuum of the heterotic string is equivalent to finding fields which satisfy the Bianchi identity \eqref{eq:bianchi3d}, the Killing spinor equations \eqref{eq:killingspinor:1}-\eqref{eq:killingspinor:7}, the instanton equations \eqref{eq:inst1}, \eqref{eq:inst3} and the integrability condition \eqref{eq:instR3d}. The discussion in appendix \ref{app:Hull} shows that the integrability condition is satisfied to first order. This means solving the Killing spinor equations and the Bianchi identity implies that the equations of motion are satisfied to lowest order as well. For clarity, we will label the lowest order solution by $(0)$, except for the bundle\footnote{We will see later that the solution at first order requires all fields to change apart from the gauge field strength.} which we will still denote by $F$. The relevant objects are then $H^{(0)}$, $\phi^{(0)}$, $\omega^{(0)}$, $\Psi^{(0)}$, $g^{(0)}$ and $F$.

\subsection{Bianchi Identity}
Let us consider the Bianchi identity first. At lowest order in $\alpha'$ it is
\begin{eqnarray*}
\d H^{(0)}=0\:.
\end{eqnarray*}
We then have the following no-go theorem
\begin{Theorem}[No-Go Theorem]
\label{tm:noflux}
If the compact space has vanishing third Betti number, $b_3=0$, then a solution to the Killing-Spinor equations \eqref{eq:killingspinor:1}-\eqref{eq:killingspinor:7} with $\d H=0$ requires $H=0$.
\end{Theorem}
\begin{proof}
Take a look at the first two Killing spinor equations \eqref{eq:killingspinor:1}, \eqref{eq:killingspinor:2}
\begin{align*}
\d(e^{-2\phi}\Psi_-)&=0\\
\d(e^{-2\phi}\omega)&=-\p_y(e^{-2\phi}\Psi_-)+*He^{-2\phi}.
\end{align*}
Since $H^3(X)=0$, these equations show that $*He^{-2\phi}$ is the sum of two exact forms and, hence, an exact form itself. Using this, we have
\begin{eqnarray*}
\vert\vert He^{-\phi}\vert\vert^2=\int_{X}H\wedge*e^{-2\phi}H=0\:,
\end{eqnarray*}
after partial integration. It follows that
\begin{align*}
H=0\:.
\end{align*}
\end{proof}
\noindent In particular, as our coset has $H^3_d(X)=0$, we see that the flux vanishes at zeroth order in $\a$. It follows that no nontrivial $H$-flux can be present in geometries where $b_3=0$ at lowest order. 
Similar no-go theorems have appeared in other contexts in the literature before \cite{Gauntlett:2002sc, Kimura:2006af, Gray:2012md}.

\subsection{Instanton Equations}\label{section:HYM}
The gauge bundle has to satisfy the instanton conditions $\omega\lrcorner F=0$ and $F\wedge\Psi=0$. The second of these conditions is automatically satisfied for the holomorphic three-form \eqref{eq:zerothsolutionsu3}, and field strength~\eqref{eq:random23523}. The first condition, however, leads to an additional constraint on the parameters appearing in the $SU(3)$-structure \cite{Klaput:2011mz}. To see this, note that $\omega\lrcorner F=0$ is equivalent to
\begin{eqnarray}
F\wedge \omega\wedge\omega=0\:.
\label{eq:dual-inst}
\end{eqnarray}
Inserting $\omega=v^i \omega_i$, with the $G$-invariant two-forms $\omega_i$ and the field strength (\ref{eq:F}) into equation~\eqref{eq:dual-inst} gives
\begin{align}
\label{eq:HYM:slope}
d_{rjk}\, p^r_av^jv^k=0\;\;\mbox{for all }\;a\:.
\end{align}
Here $d_{rjk}$ are the intersection numbers (see appendix \ref{app:coset}). 
The solution to Eqs.~\eqref{eq:HYM:slope}, for generic values of the integers $p_a^r$, is to set all $v^r$ to zero. For $SU(3)/\U{1}^2$ this leaves us with one remaining non-zero modulus $v^3$ corresponding to the non-harmonic two-form $\d\omega_3\neq0$. Therefore, from the relations \eqref{eq:moduliRSU3}-\eqref{eq:moduliRSU3last} between the K\"ahler moduli $v^i$ and the radii $R_i$ we see that the Yang-Mills equations are solved if
\begin{align*}
R_1^2&=R_2^2=R_3^2\equiv R^2&\text{for }&SU(3)/\U{1}^2\:.
\end{align*}
It then follows, using the relations \eqref{eq:tordom}, \eqref{eq:tordPsi} between the torsion classes and the $SU(3)$-structure forms, that the only non-vanishing torsion class is the real part of the first class $W_0^+=1/R$ \cite{1126-6708-2009-09-077, Klaput:2011mz}. This means that the $SU(3)$-structure of $X$ is nearly K\"ahler at this locus, that is $X$ has a half-flat $SU(3)$-structure where all torsion classes but $W_0$ vanish.

There is a subtlety in the case $SU(3)/\U{1}^2$. First note that \eqref{eq:HYM:slope} can be written, in the case of a line bundle, as
\begin{equation*}
\omega^{uv}{f_{uv}}^i\rho_a(H_i)=\Big(\frac{2}{R_1^2}-\frac{1}{R_2^2}-\frac{1}{R_3^2}\Big)\rho_a(H_7)+\sqrt{3}\Big(\frac{1}{R_3^2}-\frac{1}{R_2^2}\Big)\rho_a(H_8)=0\:,
\end{equation*}
where we have employed \eqref{eq:intersectionnums} and \eqref{eq:moduliRSU3}-\eqref{eq:moduliRSU3last}. If then $q_a=-2 p_a$ or $q_a=0$ for all $a$, the parameters $R_i$ do not have to be all equal. From now on we exclude these special cases, unless otherwise stated and we will return to this possibility when we discuss the four-dimensional effective supergravity in section \ref{sec:NK4d}.

\subsection{Killing Spinor Equations}
\label{section:hitchinflowzero}
Having solved the integrability condition, the instanton conditions and the Bianchi identity, we now turn to solving the Killing spinor equations. To lowest order the two Killing spinor equations \eqref{eq:killingspinor:5} and \eqref{eq:killingspinor:6} read
\begin{align*}
0&=*_{(0)}\,\d\phi^{(0)}
\\
0&=(2\p_y\phi^{(0)})\;*_{(0)}1 \:,
\end{align*}
or equivalently $\d\phi^{(0)}=\p_y\phi^{(0)}=0$. This means that, in addition to vanishing H-flux, the dilaton is constant. The Killing spinor equations \eqref{eq:killingspinor:1}-\eqref{eq:killingspinor:7} then reduce to the Hitchin flow equations\cite{2001math......7101H}
\begin{align}\label{eq:killing:simplified}
\d\Psi_-^{(0)}&=0
\\\label{eq:killing:simplified2}
\d\omega^{(0)}&=-\p_y\Psi_-^{(0)}
\\\label{eq:killing:simplified3}
\omega^{(0)}\wedge\d\omega^{(0)}&=0
\\\label{eq:killing:simplified4}
\d\Psi_+^{(0)}&=\omega^{(0)}\wedge\p_y \omega^{(0)}
\:.
\end{align}

As can be explicitly checked, the Hitchin flow equations \eqref{eq:killing:simplified}-\eqref{eq:killing:simplified4} are solved by the $G$-invariant $SU(3)$-structure \eqref{eq:zerothsolutionsu3}, provided the parameters $R_i$ assume a certain $y$-dependence.
To work this out, we insert the half-flat mirror geometry expansion (which we introduced in sections \ref{sec:mirrorgeom} and \ref{sec:mirrorgeom2}) into these flow equations. The two equations \eqref{eq:killing:simplified} and \eqref{eq:killing:simplified3} are automatically satisfied. 
The other two equations become
\begin{align}
\label{eq:hitch2zeroth}
Z\,e_i\,\tilde\omega^i&=d_{ijk}\,v^i(\p_yv^j)\,\tilde\omega^k\:.
\\
\label{eq:hitch1zeroth}
v^ie_i\,\beta^0&=-\p_yG\;\beta^0
\end{align}
Multiplying with $\wedge (v^l \omega_l)$ on both sides of \eqref{eq:hitch2zeroth} and integrating gives
\begin{align*}
Ze_k v^k=d_{ijk} v^i (\p_y v^j) v^k\:.
\end{align*}
Now, using the compatibility relation \eqref{compzeroth} we can express this as
\begin{align}
\label{eq:Rydependence}
e_k v^k=-\p_y G\:,
\end{align}
which shows that equations \eqref{eq:hitch1zeroth} and \eqref{eq:hitch2zeroth} are, in fact, equivalent. 

We have seen previously, that the presence of the gauge fields force all radii to be equal. The $y$-dependence should therefore reside in this overall modulus $R=R(y)$ and we write the $SU(3)$-structure forms as
\begin{align*}
\omega^{(0)} = R^2\,\tilde{v}^k \omega_k\:,\qquad \Psi^{(0)}=R^3 \left(\tilde{Z}\, \alpha_0 + \I\, \tilde{G} \,\beta^0 \right)\:,
\end{align*}
with  $(\tilde{v}^k)=(0,0,\tilde{v})$ for $SU(3)/\U{1}^2$ and a constant $\tilde{v}$. The values of $\tilde{Z}$ and $\tilde{G}$ follow from this choice via equation~\eqref{compzeroth}. From \eqref{eq:Rydependence}, the $y$-dependence of $R$ is determined by 
\begin{align}
\label{eq:ydependzerothorder}
\p_y R = - \frac{\tilde{v}}{3\,\tilde{G}}\:.
\end{align}
Since the right-hand side of this equation is a non-zero constant, the solutions for $R$ are linear in $y$ and diverging as $y\to\pm\infty$. We will see later that the $\alpha'$ corrections can remove this divergent behaviour. 

\subsection{The Four-Dimensional Perspective}
The divergent behaviour observed in the previous section can also be seen from a four-dimensional perspective. Following section \eqref{sec:4ddimred}, the NS field-strength is now assumed to have the following expansion
\begin{eqnarray*}
H=b^ie_i\beta^0+\d b^i\wedge\omega_i+\d B_4\:,
\end{eqnarray*}
leading to the following superpotential
\begin{equation*}
W=\sqrt{8}e_iT^i=\sqrt{8}T\:,
\end{equation*}
where we have set the scale-factor $\F=Z$, and $T^3=T$. Clearly, this does not allow for maximally symmetric supersymmetric solutions. It does however allow for domain wall solutions. In this case, the four-dimensional equations \eqref{eq:4dsusyvar1}-\eqref{eq:4dsusyvar2} can be shown to match the above ten-dimensional flow, as was done in \cite{Lukas:2010mf, CyrilThesis}.

In section \ref{sec:alphaprimecorr} we show that $\a$-corrections can be used to stabilise the remaining geometric modulus $v^3$ of the theory. The price we pay is that the dilaton becomes $y$-dependent, and so the domain wall persists perturbatively. In section \ref{sec:NK4d} we will consider the corresponding four-dimensional theory by performing the dimensional reduction, following section \ref{sec:4ddimred}. By adding non-perturbative effects to the theory, like e.g. a gaugino condensate, we will see that it is possible to lift this vacuum to a maximally symmetric one. 

\subsection{Side Issues: Kaluza-Klein Gauge Group and Wilson Lines}
Before we discuss $\a$-corrections, we take a moment to consider a couple of side issues that might worry the reader. An obvious question is whether the symmetries of our coset $G/H$ lead to a Kaluza-Klein gauge group in four dimensions, in addition to the remnants of the $E_8\times E_8$ gauge group. It turns out that Kaluza-Klein gauge fields from such spaces take values in the quotient $N(H)/H$ where $N(H)$ is the normaliser of $H$ in $G$ \cite{Coquereaux:1984mj}. For our coset, this quotient is merely a discrete group. Indeed, with $H=\U{1}^2$, one finds that $N(H)/H\cong S_3$, the permutation group of three elements. Hence, a Kaluza-Klein gauge group in four dimensions does not arise.

The standard method to break GUT gauge groups in heterotic constructions is to include a Wilson line in the gauge bundle. This requires a non-trivial first fundamental group of the underlying space. However, all coset studied here are simply connected and, hence, do not admit any Wilson lines. Alternatively, if the space admits a freely-acting symmetry a closely related compactification can be defined on the quotient manifold which has a non-trivial first fundamental group and, hence, allows for the inclusion of Wilson lines. However, for our coset it has been shown that only torsion-free discrete groups can have a free action on $G/H$ \cite{kobayashi1}. That is, groups which do not posses any cyclic elements. In particular, this excludes all finite groups. The mathematical literature provides an existence theorem for a freely acting infinite but finitely generated discrete freely-acting group on every coset of compact groups $G$, $H$. However, we have not been able to find such a group explicitly for one of our coset. For this reason, Wilson line breaking of the GUT symmetry is not currently an option. Instead, flux in the standard hypercharge direction might be used. Such details of particle physics model building are not the primary concern of the present chapter and will not be discussed further here.

\section{Solutions on Homogeneous Spaces Including $\alpha'$-Corrections}
\label{sec:alphaprimecorr}

In the previous sections we have seen how to construct lowest order solutions to the heterotic string on homogeneous spaces, using the associated vector bundle construction on cosets. It turns out that the four-dimensional space-time is a domain wall and that the radius, $R$, of the internal space varies linearly with $y$, the coordinate transverse to the domain wall.

How do we expect this to change if we include first order $\alpha'$ corrections? In our discussion before, we saw that the Bianchi identity \eqref{eq:bianchi3d} at lowest order requires the three-form flux $H$ to be closed, which forces $H$ to vanish at lowest order by theorem \ref{tm:noflux}. Now, at the next order the Bianchi identity is
\begin{align}
\label{eq:bianchi3dfull}
\d H=\frac{\alpha'}{4}\left(\tr F\wedge F-\tr R^{LC}\wedge R^{LC}\right)+\OO(\a^2)
\end{align}
and we expect a non-zero $H$ which is not closed. From a four-dimensional point of view, flux will contribute to the (super)-potential and we therefore expect some effect on moduli. Of course, the non-zero $H$ also feeds into the gravitino and dilatino Killing spinor equations and will change the gravitational background. 

In order to work this out, we first need to find solutions to the Bianchi identity \eqref{eq:bianchi3dfull} and then solve the Killing spinor equations \eqref{eq:killingspinor:1}-\eqref{eq:killingspinor:7} and the instanton equations \eqref{eq:inst1}-\eqref{eq:inst3}. The equations of motion are guaranteed by the integrability condition \eqref{eq:instR3d}, using the Hull connection, which at this order in $\a$ may be taken to be the Levi-Civita connection. 

\subsection{Full Solution Ansatz}
Note first that the instanton equations \eqref{eq:inst1}-\eqref{eq:inst3} again forces all the radii to be equal, as in the zeroth order case. We are therefore left with an overall radius $R(y)$, and potentially a $y$-dependent dilaton $\phi(y)$. We then make the following Ansatz for our solution, which we require to be $G$-invariant\footnote{This is the most general $G$-invariant ansatz for the solution, as $*H$ is required to be exact by equation \eqref{eq:killingspinor:2}.}:
\begin{equation}
\label{eq:Ansatz:hflux}
\begin{array}{lll}
\omega&=&R(y)^2\, \tilde{v}^i\omega_i\\
\Psi&=&R(y)^3 \left(\tilde{Z}\alpha_0+\I\,\tilde{G}\,\beta^0\right)\\
H &=& \mathcal{C}(R, \alpha')\,\frac{\vol}{\pi}\,\alpha_0\\
\phi &=&\phi(y)
\end{array}
\end{equation}
for $\{(\omega,\Psi), H, \phi\}$. The bundle is defined to be the same as at lowest order. The function $\mathcal{C}(R,\alpha')$ in the Ansatz for $H$ also depends on the bundle parameters and, along with $R(y)$, it has to be determined for a full solution. The tilded parameters have been defined in Section~\ref{section:hitchinflowzero}. In the following, we present explicit expressions for the space $SU(3)/\U{1}^2$. 

We consider our Ansatz for the Bianchi identity \eqref{eq:bianchi3dfull}. Recall that $\tr\:R^{LC}\wedge R^{LC}$ is given by \eqref{eq:trRRLeviCevita}. For $\tr F\wedge F$ we get the same result~\eqref{eq:trFF} as before. Including observable and hidden sector and assuming that the integrability conditions \eqref{eq:constraintpq1} and \eqref{eq:constraintpq2} are satisfied it can be written as
\begin{align*}
\tr F\wedge F = \mathcal{A}(\mathbf{p},\mathbf{q}, \tilde{\mathbf{ p}}, \mathbf{\tilde q})\,\frac{\vol}{\pi}\,\d\alpha_0
\end{align*}
where 
\begin{align}
\label{eq:Aparameter2}
\mathcal{A}(\mathbf{p},\mathbf{q}, \tilde{\mathbf{ p}}, \mathbf{\tilde q})
=-\frac{1}{12}\left[\sum\limits_{a=1}^{n}q_a^2+\sum\limits_{a=1}^{\tilde{n}}\tilde q_a^2\right] \:.
\end{align}
It may seem that this only depends on the bundle parameters $q_a,\,\tilde q_a$, but not on $p_a,\,\tilde p_a$. However, note that this result only hold for consistent bundles satisfying the integrability conditions \eqref{eq:constraintpq1} and \eqref{eq:constraintpq2}, which relate $p_a$, $\tilde{p}_a$ with $q_a$, $\tilde q_a$.

With these results, the Bianchi identity shows that $\mathcal{C}$ is given by
\begin{align}
\label{eq:Cexact}
\mathcal{C} = \a\B\:+\OO(\a^2)\:,
\end{align}
where 
\begin{equation}
\B=\frac{4\mathcal{A}+9}{16}\label{Bdef}
\end{equation}
is determined by the bundle parameters $\mathcal{A}$, equation~\eqref{eq:Aparameter2}. Note that the dependence of $\mathcal{C}$ on the radius $R$ is a higher order in $\a$ effect. We can therefore assume that $\mathcal{C}$ is a moduli-independent constant to the order we are working.

\subsection{Hitchin Flow Revisited} 

Apart from a non-vanishing $H$ and $y$-dependence of $R$, our Ansatz~\eqref{eq:Ansatz:hflux} remains unchanged from its lowest order form. This means that all equations \eqref{eq:killingspinor:1} - \eqref{eq:killingspinor:7} which do not contain $y$ derivatives or $H$ are automatically satisfied. 

The remaining three equations,  \eqref{eq:killingspinor:2}, \eqref{eq:killingspinor:4} and \eqref{eq:killingspinor:6}, lead to differential equations for the $y$-dependence of $R(y)$ and $\phi(y)$ and inserting the Ansatz~\eqref{eq:Ansatz:hflux} into these gives
\begin{eqnarray}
R^2 \tilde v^i e_i\,\beta^0 &=&\left( 2 \p_y\phi\, R^3\tilde G-3 R^2\p_y R\,\tilde G\right)\beta^0+\mathcal{C} \pi \beta^0\label{eq:hitchinrev1}\\
R^3 \tilde Z e_i \tilde{\omega}^i &=&d_{ijk} \tilde{v}^i \tilde{v}^j\tilde{\omega}^k\left(2 R^3 \p_y R -\p_y\phi R^4\right)\label{eq:hitchinrev2}\\
-\frac{\vol}{\pi}\,\tilde{G} \mathcal{C} R^3 \alpha_0\wedge \beta^0&=&2\,\p_y\phi\,(*1)_0\label{eq:hitchinrev3}\:.
\end{eqnarray}
A direct evaluation yields $\omega\wedge \omega\wedge\omega= -6\,R^6 (*1)_0$, where $ (*1)_0=e^{123456}$. Equation~\eqref{compzeroth} then yields the relation $(*1)_0 = \frac{\vol}{4\pi^2}\tilde{G}^2 \alpha_0\wedge \beta^0$. If we insert this last relation into the third flow equation \eqref{eq:hitchinrev3} and then use the result in equation~\eqref{eq:hitchinrev1}, we obtain 
\begin{align}
\p_y \phi &= -\frac{\mathcal{C}}{R^3}\label{fephi}\\
\p_y R &=-\frac{1}{6 \pi}\left[\tilde v + \frac{3\pi\,\mathcal{C}}{R^2}\right]\:.\label{feR}
\end{align}
Here, we have set $\tilde{G}=2\pi$, the value appropriate for $SU(3)/\U{1}^2$. These two equations already fully determine $R(y)$ and $\phi(y)$ and equation~\eqref{eq:hitchinrev2} yields no additional information. This can be seen after multiplying it with $\wedge (\tilde v^k \omega_k)$ and making use of the compatibility relation \eqref{compzeroth}, in complete analogy with the lowest order analysis in Section~\ref{section:hitchinflowzero}.

\subsection{Solving the Flow Equations}\label{sec:solvehitchin}
We now solve the above differential equations~\eqref{fephi} and \eqref{feR} for the $y$-dependence of the radius $R$ and the dilaton $\phi$ to order $\alpha'$. Inserting the leading term in $\mathcal{C}$ from equation~\eqref{eq:Cexact} into Eqs.~\eqref{fephi} and \eqref{feR} leads to
\begin{align}
\label{eq:alphaprimeflowexp1}
\p_y \phi
&=
-\frac{\B}{R^3}\alpha'
\\\label{eq:alphaprimeflowexp2}
\p_y R 
&= 
-\frac{1}{6 \pi}\left[\tilde v + \frac{3\pi\, \B\,\alpha'}{R^2}\right]
\:,
\end{align}
where $\B=(4\mathcal{A}+9)/16$, not to be confused with the Atiyah map of section \ref{subsec:defsTX}. The structure of the solutions to these equations depends crucially on the sign of $\B$ and we distinguish the three cases
\begin{equation*}
\begin{array}{ll}
 \mbox{Case 1:}&\B=0\\
 \mbox{Case 2:}&\B<0\\
 \mbox{Case 3:}&\B>0\:.
\end{array}
\end{equation*} 
Note from equation~\eqref{eq:Aparameter2}, that $\B$ is a function of the bundle parameters and that, for $SU(3)/\U{1}^2$,  all three cases can indeed be realized for appropriate bundle choices. Let us now discuss the solution for each of these cases in turn. 

\subsubsection*{Case 1, $\B=0$}
In this case, $H=0$, and Eqs.~\eqref{fephi} and \eqref{feR} revert to their zeroth order counterparts discussed in Section~\ref{section:hitchinflowzero}. This means that, due to a special choice of bundle, the $\alpha'$ corrections vanish and we remain with a constant dilaton and a linearly diverging radius $R$.

\subsubsection*{Case 2, $\B<0$}\label{case2}
In this case, equation~\eqref{eq:alphaprimeflowexp2} allows for a special $y$-independent solution where $R$ assumes the constant value
\begin{align}
\label{eq:staticR10d}
R_0^2 = \frac{3\pi\, |\B|\,\alpha'}{\tilde v}
\:.
\end{align}
For this static solution, where all geometric moduli are stabilised, the $\phi$ equation can then be easily integrated and we obtain a linear dilaton
\begin{align}
\label{eq:staticR10dphi}
\phi(y) = \frac{|\B|}{R_0^3}\alpha'\,y
\:.
\end{align}
The behaviour of this solution is radically different from what we have seen at zeroth order. There, the radius $R$ was linearly divergent and the dilaton constant. For the above solution, this situation is reversed with $R$ constant and the dilaton linearly diverging. Note that a linear dilaton will give rise to regions of strong coupling, where the solution cannot be trusted. Indeed, as $y$ approaches zero, the coupling constant $e^{-2\phi}$ becomes of $\OO(1)$. We can interpret this as the solution approaching the domain wall.

We can integrate equation~\eqref{eq:alphaprimeflowexp2} in general, to obtain the implicit solution
\begin{eqnarray*}
y-y_0=6\pi\,\Big[-\frac{R}{\tilde v}+\frac{\sqrt{3|\B|\alpha'}\;}{\tilde v^{3/2}}\operatorname{arctanh}\left(\sqrt{\frac{\tilde v}{3|\B|\alpha'}}R\right)\Big].
\end{eqnarray*}
Here $y_0$ is an arbitrary integration constant which corresponds to the position of the domain wall and will be set to zero for convenience. This solution has the generic form displayed in Fig.~\ref{fig:Rstab} (solid line) and exhibits a kink at $y=y_0=0$, further indicating the position of the domain wall. It approaches the above constant solution \eqref{eq:staticR10d} for $R$ as $|y|\to\infty$, that is, far away from the domain wall. In this limit, the dilaton asymptotes the linearly divergent behaviour~\eqref{eq:staticR10dphi}.

\subsubsection*{Case 3, $\B>0$}
No constant solution for $R$ exists in this case and integrating equation~\eqref{eq:alphaprimeflowexp2} gives
\begin{align*}
y-y_0=6\pi\,\left[-\frac{R}{\tilde v}+\frac{\sqrt{3\B\alpha'}\;}{\tilde v^{3/2}}\operatorname{arctan}\left(\sqrt{\frac{\tilde v}{3\B\alpha'}}R\right)\right]
\,.
\end{align*}
This solution is plotted in Fig.~\ref{fig:Rstab} (dashed line) for $y_0=0$. For $\vert y\vert \to \infty$, $R$ diverges linearly and in fact approaches the zeroth order solution \eqref{eq:ydependzerothorder}, while the dilaton becomes constant. Hence, we see that, far away from the domain wall, we recover the zeroth order solution, with a constant dilaton and a linearly divergent radius $R$.

\begin{figure}[h!]
\begin{center}
\includegraphics[width=90mm]{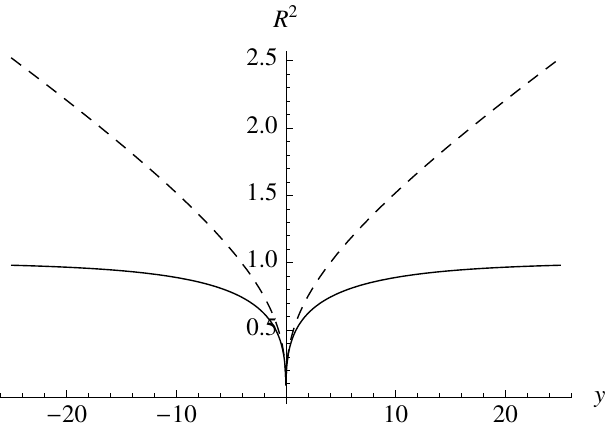}
\caption{{\it Plot of the radial modulus $R^2$ as a function of the distance, $y$, from the domain wall at $y=0$ for $\B<0$ (solid line), and $\B>0$ (dashed line). For convenience, we have set $R_0^2=1$.}}
\label{fig:Rstab}
\end{center}
\end{figure}

To summarize, we have seen that the qualitative behaviour of the moduli on $y$, the coordinate transverse to the domain wall, is controlled by the gauge bundle via the quantity $\B=(4\mathcal{A}+9)/16$ for the case of $SU(3)/\U{1}^2$, where $\mathcal{A}$ is defined in equation~\eqref{eq:Aparameter2}. For $\B=0$ the solution is, in fact, unchanged from the zeroth order one which has a constant dilaton and a linearly divergent radius $R$. For $\B>0$ the solution is modified due to $\alpha'$ effects close to the domain wall but approaches the zeroth order solution far away from the domain wall. The behaviour is quite different for $\B<0$ which, asymptotically, leads to a constant radius $R$ and a linearly diverging dilaton.

We see that $\alpha'$ effect can have a significant effect on moduli and their stabilisation. Indeed, by a suitable bundle choice, it is possible to stabilise all geometric moduli. From a four-dimensional viewpoint this should be encoded in a (super-) potential which appears at order $\alpha'$. We will now discuss this in detail by considering the four-dimenional $N=1$ supergravity associated to our solutions.

\section{The Four-Dimensional Effective Theory}
\label{sec:NK4d}
Above, we have found $\mathcal{O}(\alpha')$ corrected solutions to the ten-dimensional heterotic string. In this section, we will examine the corresponding four-dimensional effective supergravity theories and their vacua. In particular, we would like to verify that our ten-dimensional results can be reproduced from this perspective.

Before we begin we make a couple of comments about consistently reducing the theory to four dimensions. By performing the compactification, we should integrate out all the higher massive string modes and massive Kaluza-Klein degrees of freedom. These are the modes with masses above the compactification scale. We do not integrate out all the massive modes however. Indeed, as we shall see, for consistency with the ten-dimensional theory, it is essential that we include fields with masses of $\OO(\a)$ and below. As the potential generated is of $\OO(\a)$, this is precisely the mass-scale for the stabilised modulus of the previous section. For consistency, we therefore need to include these in the four-dimensional theory. We now turn to the explicit four-dimensional solution.

\subsection{Four-Dimensional Supergravity and Fields}

We now follow the reduction procedure laid out in \ref{sec:4dred} for our coset. The relevant moduli superfields are $(\Phi^X)=(S,T^i)$ with the dilaton $S$ and T-moduli $T^i$. There are no independent complex structure moduli. 

We assume the following expansion of the ten-dimensional three-form field strength
\begin{eqnarray}
\label{eq:hfluxmostgen}
H=-b^ie_i\beta^0-\d b^i\wedge\omega_i+\frac{\pi^2}{\vol}\,\mu\,\alpha_0+\d B_4\:,
\end{eqnarray}
where the minus is due to the opposite convention for the hermitian form. The factor in front of the flux parameter $\mu$ is conventional in order to simplify later expressions. The first term in this expansion is due to the non-vanishing torsion of the internal space and $e_i$ are the torsion parameters. We recall that they are given by $(e_1,e_2,e_3)=(0,0,1)$ for $SU(3)/\U{1}^2$. The third term in equation~\eqref{eq:hfluxmostgen} is a result of the non-vanishing H-flux induced via the Bianchi-identity. Its coefficient, $\mu$, can be read off from Eqs.~\eqref{eq:Ansatz:hflux}, \eqref{eq:Cexact} and is explicitly given by
\begin{align}
\label{eq:mudefinition}
\mu=\pi\alpha'\B\:,
\end{align}
where the quantity $\B=(4\mathcal{A}+9)/16$ depends on parameters of the gauge bundle as in equation~\eqref{eq:Aparameter2}. Given these preparations, we can identify the (bosonic parts of the) four-dimensional superfields as
\begin{equation*}
 S=a+\I  s\:,\quad T^i=b^i+\I\, v^i\:
\end{equation*}
As there are no complex structure moduli, the K\"ahler potential now reads
\begin{equation*}
 K=-\log\left(i(\bar{S}-S)\right)-\log(8\V)\:,
\end{equation*}
where $\V\rightarrow-\V$ in \eqref{eq:stVol} due to the opposite convention for $\omega$, while the superpotential \eqref{eq:gukovvafa}, with $\omega$ replaced by $-\omega$, is given by
\begin{equation}
\label{eq:suppot2}
 W=\sqrt{8}(e_iT^i+i\mu)\:,
\end{equation}
where we have set $\F=Z$. The superpotential \eqref{eq:suppot2} is obtained by inserting the various forms from equation~\eqref{JO} and \eqref{eq:hfluxmostgen} and using equation~\eqref{compzeroth} as well as the properties of the half-flat mirror basis given in section \ref{sec:mirrorgeom}. The first term arises from the non-vanishing torsion of the internal space and the second term is due to the non-vanishing H-flux induced by the gauge bundle and Bianchi identity.

\subsection{D-Term Conditions}
The ${\rm S}(\U{1}^n)$ and ${\rm S}(\U{1}^{\tilde{n}})$ structure groups of our observable and hidden line bundle sums also appear as gauge symmetries in the four-dimensional theory. Their associated D-terms have a Fayet-Illiopoulos (FI) terms \cite{Dine1987589}, and in general matter field terms which involve gauge bundle moduli. Switching on these moduli deforms the gauge bundle to a one with non-Abelian structure group, a possibility which we will not consider here. Focusing on the FI terms, one finds that for the observable sector
\begin{equation}
\label{eq:Dterm}
 D_a\sim \frac{d_{rij}p_a^rv^iv^j}{\V}\:,
\end{equation}
and similarly for the hidden sector. The D-flat conditions, $D_a=0$, hence implement the conditions~\eqref{eq:HYM:slope} from a four-dimensional viewpoint. Therefore, generically the D-flat conditions imply that all but the last modulus, $v=e_iv^i$, vanish as we have seen in section \ref{section:HYM}. The associated axions are absorbed by the gauge fields so we remain with a single T-modulus superfield $T=e_iT^i=b+iv$ and, of course, the dilaton $S$. In terms of these ``effective" fields the K\"ahler potential and superpotential read
\begin{equation}
 K=-\log(S+\bar{S})-\log 8(T+\bar{T})^3\:,\quad W=\sqrt{8}(T+ \mu)\:, \label{Weff}
\end{equation}
where we have switched to the ``phenomenological" definition $S=s+ia$ and $T=v+ib$ of the superfields, obtained from the previous one by multiplying the superfields by $-i$ and changing the signs of the axions.

It is worth noting that the above D-terms receive a dilaton-dependent correction at one loop~\cite{Lukas:1999nh,Blumenhagen:2005ga}. This correction is small in the relevant part of moduli space and will not change our conclusions, qualitatively. For simplicity, we will therefore neglect this correction.

Moreover, recall that for specific choices of the bundle parameters it is possible to satisfy \eqref{eq:Dterm} and leave more than just one of the moduli non-zero, as we pointed out at the end of section \ref{section:HYM}. However, for supersymmetric solutions, the corresponding F-terms
\begin{align*}
F_{T^s}\propto \frac{1}{\V}W\p_{T^s}\V\propto d_{sij}v^i v^j\:.
\end{align*}
for these moduli drive the model back to the nearly-K\"ahler locus where only the last $v^i$ is non-zero. Therefore, starting from this locus covers already the most general supersymmetric case.

\subsection{F-Term Conditions}
The superpotential~\eqref{Weff} is $S$-independent, and it is therefore expected that the dilaton cannot be stabilised. Below we will add a gaugino condensation term to $W$ in order to improve on this. However, it is still instructive at this stage to consider the F-term equations which follow from~\eqref{Weff}. For the $T$ modulus we have
\begin{equation*}
 F_T\propto\Big(1+\frac{3\mu}{v}-\frac{3ib}{v}\Big)\:.
\end{equation*} 
Hence, $F_T=0$ implies a vanishing T-axion, $b=0$, and
\begin{equation*}
 v=-3\mu\:. \label{v4d}
\end{equation*} 
Since $v>0$ this solution is only physical provided that $\B<0$ and we have seen that this can be achieved for appropriate bundle choices. Indeed, this is precisely the case discussed in Section~\ref{case2} which led to a domain solution with an asymptotically constant volume given by equation~\eqref{eq:staticR10d}. This asymptotic value is, in fact, identical to our four-dimensional result~\eqref{v4d}, as one would expect. Of course, $F_S\sim W\neq 0$ for this value of $v$ so that we do not have a full solution to the F-term conditions but, rather, a runaway in the dilaton direction. The ``simplest" solution for this type of potential is a domain wall, which is precisely what we found previously from a ten-dimensional viewpoint.

\subsection{Including a Gaugino Condensate}
We will now attempt to lift the dilaton runaway by adding a gaugino condensate term to the superpotential, employing a scenario similar to the KKLT scenario of type IIB \cite{Kachru:2003aw}. $W$ in equation~\eqref{Weff} is then replaced by
\begin{equation*}
 W=\sqrt{8}(T+ \mu+ke^{-cS})\:.
\end{equation*}
Here, $\mu$ is defined in equation~\eqref{eq:mudefinition}, $k$ is a constant of $\OO(1)$ and $c$ is a constant depending on the condensing gauge group, with typical values  $c_{\SU{5}}=2\pi/5$, $c_{{\rm E}_6}=2\pi/12$, $c_{{\rm E}_7}=2\pi/18$ and $c_{{\rm E}_8}=2\pi/30$. In the following, it will be useful to introduce the re-scaled components
\begin{equation*}
 x=cs\:,\quad y=ca
\end{equation*}
of the dilaton superfield. With those variables, the dilaton F-term equations, $F_S=0$, then read
\begin{align*}
v+\mu+(1+2x)ke^{-x}\textrm{cos}(y)&=0\\
b-(1+2x)ke^{-x}\textrm{sin}(y)&=0\, ,
\end{align*}
while $F_T=0$ leads to
\begin{align*}
v+3\mu+3ke^{-x}\textrm{cos}(y)&=0\\
b-ke^{-x}\textrm{sin}(y)&=0\, .
\end{align*}
The vanishing of the superpotential, $W=0$, is equivalent to the conditions
\begin{align*}
v+\mu+ke^{-x}\textrm{cos}(y)&=0\\
b-ke^{-x}\textrm{sin}(y)&=0\, .
\end{align*}
The simplest type of vacuum is a supersymmetric Minkowski vacuum, that is a solution of $F_S=F_T=W=0$. It is easy to see that this can only be achieved for $s=0$, which corresponds to the limit of infinite gauge coupling, and is therefore discarded. 

Next, we should consider supersymmetric AdS vacua, which are stable by the Breitenlohner-Freedman criterion \cite{Breitenlohner:1982jf, Breitenlohner:1982bm}. These are solutions of $F_S=F_T=0$. It follows immediately that the axions are fixed by $\cos (y)=-{\rm sign}(k)$ and $b=0$, while $x$ and $v$ are determined by
\begin{equation}
 f(x)\equiv (1-x)e^{-x}=\frac{\mu}{k}\:,\qquad v=\frac{3x}{1-x}\mu\:.\label{min}
\end{equation} 
Normally, we require a solution with $x>1$ in order to be at sufficiently weak coupling, and we will focus on this case. Then, for a positive $v$ we need the flux parameter $\mu$ to be negative and, hence, the constant $k$ to be positive. A negative value for $\mu$ is indeed possible. Provided this choice of signs, the equations~\eqref{min} have two solutions, one with a value of $x$ satisfying $1<x<2$ which is an AdS saddle and another one with $x>2$ which is an AdS minimum. The cosmological constant at those vacua is given by
\begin{equation*}
 \Lambda=-\frac{3c\mu^2}{4v^3x}\left(\frac{1+x}{1-x}\right)^2\:.
\end{equation*} 
We note that $v$ is stabilized perturbatively while stabilization of the dilaton involves the gaugino condensation term. 

It has of course been observed some time ago~\cite{Dine:1985rz} that the dilaton in heterotic CY compactifications can be stabilized by a combination of a constant, arising from H-flux, and gaugino condensation in the superpotential. The situation here is different from these early considerations in two ways.
\begin{itemize}
 \item There is an additional T-dependent term in the superpotential which arises from the non-vanishing torsion of the internal space. 
 \item The flux term in the superpotential does not arise from harmonic H-flux but from bundle flux through the Bianchi identity.
\end{itemize} 
It is important to check that the above vacuum can be in a acceptable region of field space where all consistency conditions are satisfied. To discuss this, we set $\alpha'=1$ from hereon. We need that $x>1$ to be at weak coupling, $v\gg1$ so that the $\alpha'$ expansion is sensible, $k\exp (-x)<1$ so that the condensate is small and $\vert\Lambda\vert<<1$ for a small vacuum energy. Eqs.~\eqref{min} immediately point to a tension in satisfying the first two of these constraints. While $v$ is proportional to the bundle flux $\mu$ and, hence, prefers a large value of $\mu$, a large value of the dilaton requires $\mu$ to be small.

Let us consider this in more detail. For concreteness we use a minimum value of $v=9$, a sufficiently large value for the $\alpha'$ expansion to be sensible. This implies the constraint
\begin{equation}
\label{eq:condmu1}
\vert\mu\vert\ge\frac{3(x-1)}{x}
\end{equation}
on the flux $\mu$. We also require the non-perturbative effects to be weak, that is $k\exp (-x)<1$, which leads to the condition
\begin{equation}
\label{eq:condmu2}
\vert\mu\vert\le x-1.
\end{equation}
Combining both conditions, it follows that $x\ge3$. We also find that $\Lambda\leq\OO(0.1)$, consistent with a small vacuum energy. Hence, the two conditions \eqref{eq:condmu1} and \eqref{eq:condmu2} are necessary and sufficient to guarantee a consistent vacuum.
\begin{figure}[h!]
\begin{center}
\includegraphics[width=85mm]{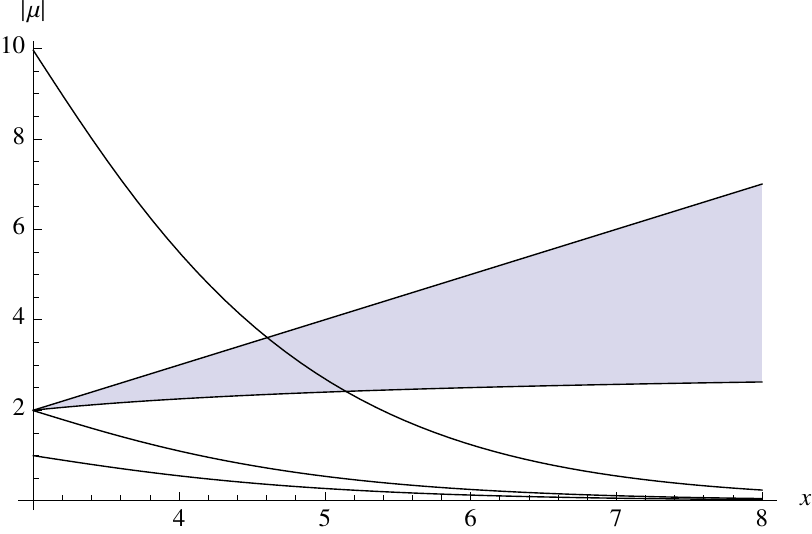}
{\caption{\it Plot of the consistent values for $\vert\mu\vert$. The shaded part is defined by the conditions \eqref{eq:condmu1} and \eqref{eq:condmu2}. The other three lines represent the condition~\eqref{eq:condk} for values $k_{\rm max}=10$ (bottom line), $k_{\rm max}=20$ (middle line) and $k_{\rm max}=100$ (top line). Consistent values for the flux $|\mu|$ are, hence, defined by the shaded part located below the line for the value of $k_{\rm max}$ under consideration.}}
\label{fig:consistent}
\end{center}
\end{figure}

There is a further condition, concerning the constant $k$ in the gaugino condensation term, whose value for a given vacuum is given by 
\begin{equation*}
k=\frac{\vert\mu\vert e^x}{x-1}.
\end{equation*}
The general expectation is for $k$ to be of around $\OO(1)$, so requiring it to be less than some maximum value $k_{\textrm{max}}$ implies
\begin{equation}
\label{eq:condk}
\vert\mu\vert\le k_{\textrm{max}}(x-1)e^x\,.
\end{equation}
Fig.\ \ref{fig:consistent} shows the restriction on $|\mu|$ for different values of $k_{\textrm{max}}$. We see that simultaneous solutions to \eqref{eq:condmu1}, \eqref{eq:condmu2} and \eqref{eq:condk} only exist if $k_{\textrm{max}}\ge20$. For $k_{\rm max}={\cal O}(100)$ the consistent flux values are in the range $2\leq |\mu|\leq 4$. 

\subsection{Supersymmetric AdS Example}
We would now like to show that the required values for the flux can indeed be obtained for appropriate choices of the gauge bundle. 
On the coset $\SU3/\U1^2$ we choose observable and hidden line bundle sums defined by the parameters
\begin{align*}
(p_i)&=(-2,0,0,0,2) & (q_i)&=(1,-2,1,2,-2)\\
(\tilde{p}_i)&=(2,2,0,-2,-2) & (\tilde{q}_i)&=(-3,-4,-1,4,4)\:.
\end{align*}
For this choice, the anomaly constraints~\eqref{eq:constraintpq1} and \eqref{eq:constraintpq2} are satisfied and the chiral asymmetry in the observable sector is three. Since both line bundle sums have rank five the gauge group in both sectors is $\mathrm{S}(\U{1}^5)\times\SU{5}$.
Computing the flux $\mu=\pi{\mathcal B}$ from equation~\eqref{Bdef} for this bundle choice leads to 
\begin{align*}
\mu=-\frac{15\pi}{16}\approx-2.95.
\end{align*}
This value is negative, as required, and indeed within the consistent range for $|\mu|$. Both the AdS saddle and the AdS minimum can be realized for this value of $\mu$, as can also be seen from Fig.~\ref{fig:susypot}. Many more consistent examples can be found for the coset $\SU3/\U1^2$. See \cite{Klaput:2012vv} for more details.

\begin{figure}[h!]
\begin{center}
\includegraphics[width=100mm]{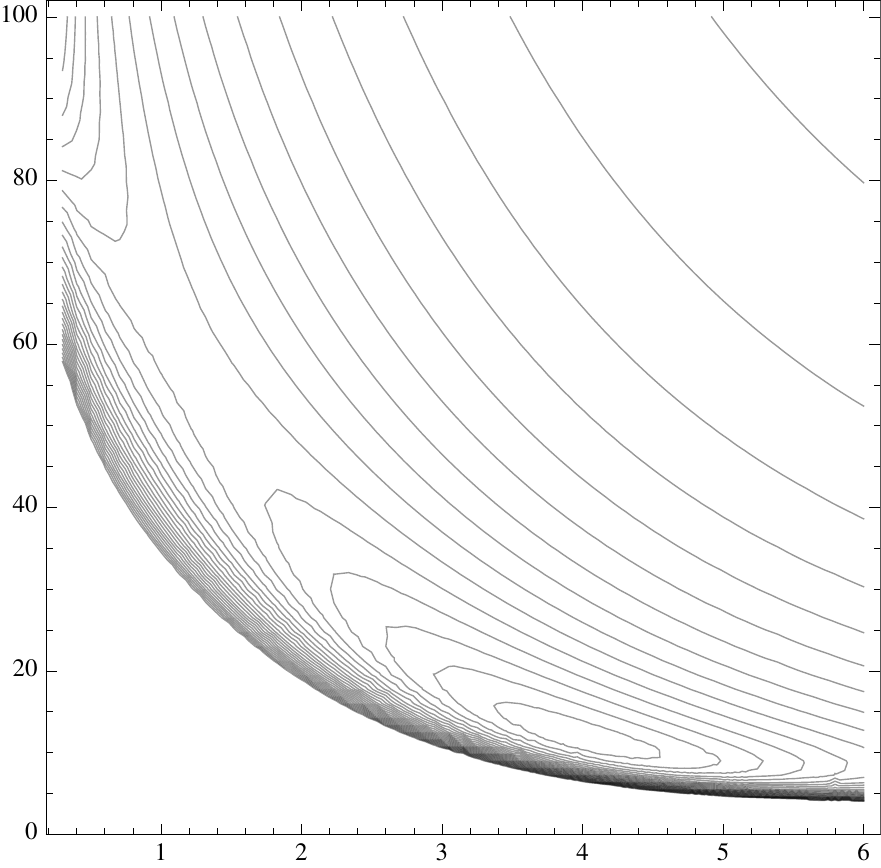}
{\caption{\it Contour plot of the potential with $\textrm{cos}(y)=-1$, for $k=53.4$ and $\mu=-15\pi /16$. This potential has a supersymmetric AdS minimum at $(x,v)\simeq (4,11.8)$, and also a supersymmetric AdS saddle at $(x,v)\simeq (1.18,58)$.}\label{fig:susypot}}
\end{center}
\end{figure}

\section{Discussion and Outlook}
\label{sec:Disc}
We have thus seen that a combination of $\alpha'$ and non-perturbative effects can indeed lift the runaway directions of the original, lowest-order perturbative potential and lead to a supersymmetric AdS vacuum. For appropriate bundle choices this stabilisation does arise in a consistent part of moduli space, that is, at weak coupling and for moderately large internal volume. However, there is a tension in that it is not possible, for the specific examples analysed, to make the volume very large (so that there is no doubt about the validity of the $\alpha'$ expansion) and keep the theory at weak coupling. 



These results provide the first concrete indication that maximal symmetry at lowest order in a string solution might not be a necessary condition for a physically acceptable vacuum. This, in turn, would mean that much larger classes of internal manifolds, such as half-flat manifolds and their generalizations, are relevant to string phenomenology.  A central question in this context is, of course, how the domain wall tension, essentially set by the torsion of the manifold, can be made sufficiently small so that other effects can compete and lift the vacuum. In our examples, this can be arranged  -- at a marginal level -- by a choice of gauge bundles, although it is not possible to stabilise the theory at parametrically large volume. However, it should be kept in mind that the spaces under consideration, i.e. cosets, have a rather limited pattern of torsion and flux parameters available. It remains to be seen whether other half-flat manifolds offer more flexibility in this regard.

Note also that in this chapter we have only discussed supersymmetric solutions to the corresponding four-dimensional supergravity. One might wonder if non-supersymmetric vacua exist. Indeed, supersymmetry should be broken from a phenomenological perspective. Moreover, once supersymmetry is established, it is notoriously difficult to break it, especially in models obtained by string compactification. Looking for string vacua where supersymmetry is broken at the compactification scale is therefore an intriguing possibility, even if it means that supersymmetry can no longer be used as a solution to the hierarchy problem. Specifically, it would be interesting to look for non-supersymmetric vacua for supereravities of the type described here. A study of this kind is underway.

\begin{subappendices}


\section{The Coset $SU(3)/U(1)^2$}\label{app:coset}
This appendix provides a short summary of all relevant data for the coset considered in this chapter, namely $SU(3)/\U{1}^2$. More details and derivations can be found in Ref.~\cite{Klaput:2011mz, Klaput:2012vv} and references therein. The data given here includes the generators of the Lie-group, relevant topological data and the half-flat mirror structure defined by the two-forms $\{\omega_i\}$, their four-form duals $\{\tilde\omega_i\}$ and the symplectic set $\{\alpha_0, \beta^0\}$. In accordance with our index convention, the reductive decomposition of the Lie algebra of $G$ is given by $\{T_A\}=\{K_u,H_i\}$, where the $K_u$, $u=1,\dots,6$ denote the coset generators and $H_i$ the generators of the sub-group $H$.


A possible choice of $SU(3)$ generators is provided by the Gell-Mann matrices
\begin{eqnarray*}
&\lambda_1
=
-\frac{i}{2}\left(
\begin{array}{ccc}
 0 & 1 & 0 \\
 1 & 0 & 0 \\
 0 & 0 & 0
\end{array}
\right)
,\;
\lambda_2
=
\frac{1}{2}\left(
\begin{array}{ccc}
 0 & -1 & 0 \\
 1 & 0 & 0 \\
 0 & 0 & 0
\end{array}
\right)
,\;
\lambda_3
=
-\frac{i}{2}\left(
\begin{array}{ccc}
 1 & 0 & 0 \\
 0 & -1 & 0 \\
 0 & 0 & 0
\end{array}
\right)
,
\\
&\lambda_4
=
-\frac{i}{2}\left(
\begin{array}{ccc}
 0 & 0 & 1 \\
 0 & 0 & 0 \\
 1 & 0 & 0
\end{array}
\right)
,\;
\lambda_5
=\frac{1}{2}\left(
\begin{array}{ccc}
0 & 0 & -1 \\
 0 & 0 & 0 \\
 1 & 0 & 0

\end{array}
\right)
,\;
\lambda_6
=-\frac{i}{2}\left(
\begin{array}{ccc}
 0 & 0 & 0 \\
 0 & 0 & 1 \\
 0 & 1 & 0
\end{array}
\right)
,
\\
&
\lambda_7
=\frac{1}{2}\left(
\begin{array}{ccc}
  0 & 0 & 0 \\
 0 & 0 & -1 \\
 0 & 1 & 0
\end{array}
\right)
,\;
\lambda_8
=-\frac{i}{2\sqrt{3}}
\left(
\begin{array}{ccc}
 1 & 0 & 0 \\
 0 & 1 & 0 \\
 0 & 0 & -2
\end{array}
\right)
\:.
\end{eqnarray*}
The two $\U{1}$ sub-groups are generated by $\lambda_3$ and $\lambda_8$. Hence, we choose as generators the re-labelled Gell-Mann matrices
\begin{align}
\begin{aligned}
K_1&=\lambda_1\qquad K_2=\lambda_2\qquad K_3=\lambda_4\qquad K_4=\lambda_5\\
K_5&=\lambda_6\qquad K_6=\lambda_7\qquad H_7=\lambda_3\qquad H_8=\lambda_8\:.
\end{aligned}
\end{align}
The geometry of the homogeneous space $SU(3)/\U{1}^2$ is determined by the structure constants which, relative to the basis $\{K_u, H_i\}$, are given by
\begin{eqnarray}\nonumber
&&f_{12}^{\phantom{12}7}=1\\
&&f_{13}^{\phantom{13}6}=
-f_{14}^{\phantom{14}5}=
f_{23}^{\phantom{23}5}=
f_{24}^{\phantom{24}6}=
f_{73}^{\phantom{73}4}=
-f_{75}^{\phantom{75}6}=1/2\\\nonumber
&&f_{34}^{\phantom{34}8}=
f_{56}^{\phantom{56}8}=\sqrt{3}/2
\:.
\end{eqnarray}
A basis of $G$-invariant two-, three- and four-forms is given by
\begin{equation}
\label{eq:Ginvforms}
\begin{array}{lllllll}
  \omega_1&=&-\frac{1}{2\pi}\Big(e^{12}+\frac{1}{2}e^{34}-\frac{1}{2}e^{56}\Big)&&\tilde\omega^1&=&\frac{4\pi}{3\vol}\Big(2e^{1234}+e^{1256}-e^{3456}\Big) \\
 \omega_2&=&-\frac{1}{4\pi}\Big(e^{12}+e^{34}\Big)&&\tilde\omega^2&=&-\frac{4\pi}{\vol}\Big(e^{1234}+e^{1256}\Big) \\
\omega_3&=&\frac{1}{3\pi}\Big(e^{12}-e^{34}+e^{56}\Big)&&\tilde\omega^3&=&\frac{\pi}{\vol}\Big(e^{1234}-e^{1256}+e^{3456}\Big) \\
\alpha_0&=&\frac{\pi}{2\vol}\Big(e^{136}-e^{145}+e^{235}+e^{246}\Big)&&\beta^0&=&\frac{1}{2\pi}\Big(e^{135}+e^{146}-e^{236}+e^{245}\Big)
\end{array}
\end{equation}
where $e^{i_1\dots i_n}:= e^{i_1} \wedge\dots\wedge e^{i_n}$ and the dimensionless volume $\vol$ is given by
\begin{eqnarray}\label{eq:VolX:su3}
\vol=\int_Xe^{123456}=\int_X(*1)_0=4(2\pi)^3\:.
\end{eqnarray}
This $G$-invariant basis forms fulfil the half-flat mirror relations in section \ref{sec:mirrorgeom} with torsion parameters $(e_1,e_2,e_3)=(0,0,1)$ and intersection numbers
\begin{align}\label{eq:intersectionnums}
\begin{aligned}
d_{111}&=6\qquad d_{112}=3\qquad d_{113}=4\qquad d_{122}=1\qquad d_{123}=2\qquad d_{133}=0\\
d_{222}&=0\qquad d_{223}=\frac{4}{3}\qquad d_{233}=0\qquad d_{333}=-\frac{64}{9}\:.
\end{aligned}
\end{align}
The only non-zero Betti numbers are $b_0=1$, $b_2=2$, $b_4=2$ and $b_6=1$ so that the Euler number is $\chi=6$.
The most general $G$-invariant $SU(3)$-structure forms are given by
\begin{align}
\begin{aligned}
\omega&=R_1^2e^{12}-R_2^2e^{34}+R_3^2e^{56}=v^i\omega_i,\\
\Psi&=R_1R_2R_3\Big((e^{136}-e^{145}+e^{235}+e^{246})+\I\, (e^{135}+e^{146}-e^{236}+e^{245})\Big)=Z\,\alpha_0+\I\, G\,\beta^0\;
\end{aligned}
\end{align}
with associated $G$-invariant metrics
\begin{align}
\d s_0^2
&=
R_1^2\,(e^1\otimes e^1+e^2\otimes e^2)
+
R_2^2\,(e^3\otimes e^3+e^4\otimes e^4)
+
R_3^2\,(e^5\otimes e^5+e^6\otimes e^6)
\:.
\end{align}
In these relations, the $R_i$ are three arbitrary ``radii" of the coset space which are related to the moduli $v^i$ by
\begin{align}\label{eq:moduliRSU3}
v^1&=-\frac{4\pi}{3}(R_1^2+R_2^2-2R_3^2)\\
v^2&=4\pi(R_2^2-R_3^2)\\
v^3&=\pi(R_1^2+R_2^2+R_3^2)\label{eq:moduliRSU3last}
\end{align}
and to $(Z,G)$ by
\begin{align}
\label{eq:ZG}
Z=\frac{2\vol}{\pi}R_1R_2R_3\:,\qquad
G=2\pi R_1R_2R_3\:.
\end{align}

\end{subappendices}

\chapter{Calabi-Yau Compactifications with Flux}
\label{ch:CY}

Having discussed non-maximally symmetric compactifications on more general non-complex torsional spaces, we now return to considering compactifications of a simpler type. Namely Calabi-Yau compactifications. We are again interested in the geometric moduli sector of the Calabi-Yau $X$. As we noted in the introduction, moduli-stabilisation is notoriously difficult in heterotic Calabi-Yau compactifications. It is the purpose of this chapter to argue that NS flux,  which naively one would think is not present in such compactifications, can be added to the toolbox provided we give up a maximally symmetric (Minkowski) space-time. This allows for more flexibility in heterotic model building. We will mostly work at zeroth order in $\a$ in this chapter, as we do not need to go to higher orders for the point we want to make. The chapter is based on \cite{Klaput:2013nla}.

\section{Introduction}
We show that an internal Calabi-Yau manifold is consistent with the presence of NS flux provided four-dimensional space-time is taken to be a domain wall. These Calabi-Yau domain wall solutions can still be associated with a covariant four-dimensional N=1 supergravity. In this four-dimensional context, the domain wall arises as the ``simplest" solution to the effective supergravity due to the presence of a flux potential with a runaway direction. The main message is that {\it NS flux is a legitimate ingredient for moduli stabilisation in heterotic Calabi-Yau models}. Ultimately, the success of such models depends on the ability to stabilise the runaway direction and thereby ``lift" the domain wall to a maximally supersymmetric vacuum, as was done in the last chapter with use of non-perturbative effects. 

The chapter is organised as follows. In section \ref{sec:Ricci-flat} we give a short recap of why Ricci-flat maximally symmetric compactifications do not allow for flux, before we argue how non-maximally symmetric compactifications avoid this no-go result. In section \ref{sec:cydomain10d} we specialise to domain wall compactifications and show that for every Calabi-Yau space there exists a domain wall solution with given harmonic NS flux. In section \ref{sec:lowenergy} we argue that the four-dimensional effective theories of regular Calabi-Yau vacua and Calabi-Yau domain wall vacua, differ essentially by the presence of a non-vanishing superpotential for the complex structure moduli. The proof that our constructions are valid away from the large complex structure limit in moduli space is given in the appendix. We discuss our results in section \ref{sec:conclCY}.

\section{Maximally Symmetric Space-Time}\label{sec:Ricci-flat}
As a warm-up, we begin by reviewing the standard arguments for why NS flux is inconsistent with an internal Calabi-Yau manifolds, provided the four-dimensional space is maximally symmetric. It is then shown that these arguments break down if we allow the four-dimensional space-time to be a domain wall. The full ten-dimensional Calabi-Yau domain wall solution is presented in the next section. 

We begin with the standard assumption that ten-dimensional space is a (possibly warped) product of a compact six-dimensional space $X_6$ and four-dimensional maximally symmetric space-time $M_4$ with metric
\begin{equation}
\label{eq:ansatzmaxsymm}
ds^2_{10} = e^{2A(x^m)}\left(g_{\mu\nu}(x^\mu)\,\d x^\nu \otimes \d x^\nu + g_{mn}(x^m)\,\d x^m \otimes \d x^n\right)\; .
\end{equation}
Here $A$ is a warp factor, $g_{\mu\nu}$ with indices $\mu,\nu,\dots =0,1,2,3$ is a maximally symmetric metric on $M_4$ and $g_{mn}$ with indices $m,n,\dots =4,\ldots ,9$ an unspecified metric on $X_6$. As usual, we demand that some supersymmetry is unbroken by the compactification. Recall the corresponding conditions, from the supersymmetry transformations of the gravitino and the dilatino
\begin{align}
\label{eq:killing1CY}
\Big(\nabla_M+\frac{1}{8}\H_M\Big)\epsilon=0\\
\label{eq:killing2CY}
\Big(\slashed\nabla\phi+\frac{1}{12}\H\Big)\epsilon=0\; .
\end{align}

The standard course of action is to set $H=0$ in those equations and study the resulting implications, leading to the well-known conclusion that $X_6$ must be a Calabi-Yau manifold with a Ricci-flat metric $g_{mn}$. Here, we are interested in the converse, namely assuming that $X_6$ is a Calabi-Yau manifold and analysing the implications for $H$. Any components of $H$ with four-dimensional indices must, of course, vanish due to four-dimensional maximal symmetry so we can focus on the purely internal components $H_{mnp}$. 


\begin{Theorem}[No-Go Theorem]
\label{tm:nogo}
Maximally symmetric heterotic Calabi-Yau compactifications do not allow for NS-flux.
\end{Theorem}
\begin{proof}
We begin by contracting eq.~(\ref{eq:killing1CY}) with $\Gamma^M$ and using eq.~(\ref{eq:killing2CY})  to get
\begin{equation*}
\Big(\slashed\nabla-\frac{3}{2}\slashed\nabla\phi\Big)\epsilon=0\:,
\end{equation*}
where the contractions are now over indices on the internal space $X_6$. For the re-scaled spinor $\tilde\epsilon=e^{-\frac{3}{2}\phi}\epsilon$ this implies $\slashed\nabla\tilde\epsilon=0$. On compact Ricci-flat, and in particular Calabi-Yau spaces, this implies that $\nabla_m\tilde\epsilon=0$. We conclude that $\tilde{\epsilon}$ is a covariantly constant spinor under the Levi-Civita connection. After a suitable $SO(6)$ redefinition of the gamma matrices we may assume that $\Gamma^a\tilde\epsilon=0$, where $a,b,\dots$ are holomorphic internal indices. Then, eq.~(\ref{eq:killing1CY}) leads to
\begin{equation*}
\Big(3\nabla_m\phi+\frac{1}{4}\H_m\Big)\tilde\epsilon=0.
\end{equation*}
Expanding in holomorphic and anti-holomorphic indices, and using $\{\Gamma^a,\Gamma^{\bar b}\}=2g^{a\bar b}$, this becomes
\begin{equation*}
\Big(3\nabla_m\phi+\frac{1}{2}H_{ma\bar b}g^{a\bar b}+\frac{1}{4}H_{m\bar a\bar b}\Gamma^{\bar a\bar b}\Big)\tilde\epsilon=0.
\end{equation*}
The last term implies $H_{m\bar a\bar b}=0$ and, since $H$ is  a real form, this leads to $H=0$. Then, it follows from the first term that $\nabla_m\phi=0$. Hence, we conclude that solving the supersymmetry conditions for a maximally symmetric four-dimensional space-time and an internal Calabi-Yau space, requires us to set $H=0$. This is the standard no-go theorem for flux on Calabi-Yau manifolds.
\end{proof}
This Theorem also makes sense from the perspective of chapter \ref{ch:SS}. Recall that for supersymmetric Minkowksi compactifications we have
\begin{equation}
\label{eq:Bismut2}
H=i(\p-\bp)\omega\:.
\end{equation}
It follows that if $X$ is K\"ahler, which is required by Calabi-Yau, then $H=i(\p-\bp)\omega=0$. Theorem \ref{tm:nogo} can also be deduced from a four-dimensional perspective. Recall the superpotential \eqref{eq:gukovvafa},
\begin{equation*}
W\propto\int_{X_6}(H+i\:\d\omega)\wedge\Psi=\int_{X_6}(\CS+\d T)\wedge\Psi\:,
\end{equation*}
of the four-dimensional theory. Here we have used \eqref{eq:anomalycancellation}, and we have set $T=B+i\omega$. We have also included the $\a$-correction Chern-Simons term $\CS$ for the time being. 

Recall that maximal symmetry and supersymmetry requires
\begin{equation*}
W=\p_X W=0\:.
\end{equation*}
Varying $W$ with respect to $T$, we find
\begin{equation*}
0=\delta_TW=-\int_{X_6}\delta T\wedge\d\Psi\:.
\end{equation*}
It follows that we need
\begin{equation*}
\d\Psi=0\:,
\end{equation*}
i.e. $\Psi$ is holomorphic. Similarly, a variation of $W$ with respect to the complex structure $J$, noting that $\delta_J\Psi=K\Psi+\chi$ where $\chi$ is of type $(2,1)$, we get
\begin{equation*}
0=\int_X(H+i\d\omega)\wedge(K\Psi+\chi)\:.
\end{equation*}
For generic $K$ and $\chi$, it follows that
\begin{align*}
(H+i\d\omega)^{(0,3)}&=H^{(0,3)}=0\\
(H+i\d\omega)^{(1,2)}&=0\:.
\end{align*}
Using that $H$ is real, it follows from these equations that
\begin{equation*}
H=i(\p-\bp)\omega\:,
\end{equation*}
in agreement with \eqref{eq:Bismut2}. It follows that $H$ can only be non-zero for Minkowski solutions if the internal space is non-K\"ahler.

Is it possible to avoid this conclusion by relaxing the condition of unbroken supersymmetry? There is a simple argument \cite{Gauntlett:2003cy}, which shows that this does not change the situation, at least at zeroth order in $\alpha'$. To see this, let us recall that the dilaton equation of motion reads to zeroth order in $\alpha'$
\begin{equation}
\label{eq:dilatoneom}
\nabla^2\,\e ^{-2\phi}=\e^{-2\phi} *(H\wedge * H)\:,
\end{equation}
where $\nabla_M$ is the covariant derivative of the  Levi-Civita connection on $M_{10}$. With the ansatz \eqref{eq:ansatzmaxsymm} it becomes
\begin{equation*}
-\d\left(\e^{4A}* \d \e ^{-2\phi}\right)=\e^{4A-2\phi} H\wedge * H\:,
\end{equation*}
Integrating over $X_6$ we obtain
\begin{equation}
\label{eq:normHzero}
-\int_{X_6}\d\left(\e^{4A}*\d\e ^{-2\phi}\right)=\int_{X_6}\e^{4A}\e^{-2\phi} (H\wedge *H)=\|\e^{2A}\e^{-\phi} H\|^2\:.
\end{equation}
However, since $X_6$ is compact the integral on the left-hand side must vanish, which implies that $H=0$. This is a special case of a more general theorem, which says that there are no smooth maximally symmetric flux compactifications in supergravity \cite{Maldacena:2000mw}.

\section{Calabi-Yau Domain Walls and Flux}\label{sec:cydomain10d}
In the previous section we saw that Calabi-Yau compactifications of the heterotic string with maximally symmetric four-dimensional space-time are inconsistent with the presence of flux. If we would like to add flux, we have to relax one of the underlying conditions. We will relax the condition of four-dimensional maximal symmetry. Instead, we assume that four-dimensional space-time has the structure of a domain wall, $M_4=M_3\times Y$, with the associated ten-dimensional Ansatz \eqref{domainwallansatz} for the metric, and a non-constant dilaton and non-zero flux $H_{uvw}$ on the internal space $X_6$ only. 

\subsection{Basic Equations}
We would now like to ask if the system of equations \eqref{eq:killingspinor:1}-\eqref{eq:killingspinor:7} can be solved for non-zero $H$, provided that $X_6$ is a Calabi-Yau manifold and $(\omega,\Psi)$ is the integrable $SU(3)$-structure with $\d\omega=\d\Psi=0$. Then, the second eq.~\eqref{eq:killingspinor:2} implies that the dilaton is a constant on $X_6$, $\d\phi=0$ and, as a result, the first three equations~\eqref{eq:killingspinor:1}--\eqref{eq:killingspinor:3} are satisfied. The remaining four equations specialize to the flow equations
\begin{align}
\label{eq:killingspinorCY:1}
{\Psi'}_-&=2\phi '\Psi_-+*H
\\\label{eq:killingspinorCY:2}
\omega\wedge\omega'&=\phi' \omega\wedge\omega\;
\\\label{eq:killingspinorCY:3}
\Psi_-\wedge H&=2\phi'\;*1\:,
\end{align}
and the constraint
\begin{align}
\label{eq:10dconstraints}
\Psi_+\wedge H =0 \:.
\end{align}

The equations~\eqref{eq:killingspinorCY:1}--\eqref{eq:killingspinorCY:3} are first-order differential equations which describe the variation of the $SU(3)$-structure $(\omega,\Psi)$ and the dilaton $\phi$ along the $y$-direction. By expanding $\omega$ and $\Psi$ into a basis of harmonic two- and three-forms they can be broken up into a set of first-order differential equations whose solutions exist locally from general theorems. Flux quantisation \cite{Rohm1986454}, also requires that $H$ is quantised, and therefore a constant along $y$. Eq.~\eqref{eq:10dconstraints} represents an additional constraint on the complex structure. Let us now analyse this in more detail.

\subsection{Existence of Solutions}
\label{sec:exist}
To analyse eqs.\ \eqref{eq:killingspinorCY:1}--\eqref{eq:10dconstraints} in detail, we introduce a symplectic basis $\{\alpha_A,\beta^B\}$ of harmonic three forms and a basis $\{\omega_i\}$ of harmonic two forms on $X_6$. As usual, the $SU(3)$-structure forms $(\omega,\Psi)$ can then be expanded as
\begin{equation}\label{eq:expOmZ}
 \omega=v^i\omega_i\; ,\qquad \Psi=Z^A\alpha_A-\mathcal{G}_{B}\beta^B\; ,
\end{equation} 
where $v^i$ and $Z^A$ are the K\"ahler and complex structure moduli, respectively, and the functions $\mathcal{G}_{B}$ are the first derivatives of the prepotential ${\mathcal G}={\mathcal G}(Z)$. Note that the pre-potential is rescaled by a factor of $i$ from \eqref{eq:exphf}. This is to get the form of $\mathcal{G}$ in conventional Calabi-Yau compactifications.

We also introduce the volume
\begin{equation}
\V=\frac{1}{6}\int_{X_6} \omega\wedge\omega \wedge\omega=\frac{1}{6}d_{ijk}v^iv^jv^k\; , \label{Vdef}
\end{equation} 
with the triple intersection numbers $d_{ijk}$. For more details on the description of the Calabi-Yau moduli space, see Appendix \ref{app:intro}. Likewise, the expansion of the flux in terms of the symplectic basis can be written as
\begin{equation}
\label{eq:fluxexp}
H =\mu^A\alpha_A+\epsilon_B\beta^B\; ,
\end{equation}
where $\mu^A$ and $\epsilon_A$ are the flux parameters. It is useful to introduce the re-scaled complex structure moduli $X^A=e^{-2\phi}Z^A$. Since the functions ${\mathcal G}_A$ are homogeneous of degree one it follows that ${\mathcal G}_A(X)=e^{-2\phi}{\mathcal G}_A(Z)$. We also define a new coordinate $z$ by 
\begin{equation}
\label{eq:variablechange}
 \frac{dy}{dz}=e^{2\phi}\; .
\end{equation} 

Using the above expansions and definitions, and assuming that the Hodge-dual and $\p_z$ anti-commute when acting on $\tilde\Psi=e^{-2\phi}\Psi$, which is shown to be a valid assumption in appendix \ref{app:commute}, the flow equations~\eqref{eq:killingspinorCY:1}--\eqref{eq:killingspinorCY:3} can be re-written in the form
\begin{eqnarray}
\label{eq:killingmod1}
 \partial_z\,{\rm Re}(X^A)&=&-\mu^A\\
\label{eq:killingmod2}
 \partial_z\,{\rm Re}({\mathcal G}_A)&=&\epsilon_A\\
 \partial_z\, v^i&=&\partial_z\phi\\
 \partial_z\phi&=&\frac{e^{4\phi}}{2\V}\textrm{Im}\left(\epsilon_AZ^A + \mu^A\mathcal{G}_A \right)\label{dileq}
\end{eqnarray} 
while the constraint \eqref{eq:10dconstraints} takes the form
\begin{equation}
 \textrm{Re}\left(\epsilon_AX^A +  \mu^A\mathcal{G}_A\right)=0\; . \label{cons}
\end{equation} 
Here and in the following $\mathcal{G}_A$ should be interpreted as functions of the re-scaled complex structure moduli $X^A$. The first three of these equations are easily integrated leading to
\begin{equation}
 \textrm{Re}(X^A)=-\mu^Az-\gamma^A\:,\quad
 \textrm{Re}(\mathcal{G}_B)=\epsilon_Bz+\eta_B\:,\quad
 v^i=e^\phi v^i_0\; , \label{vsol}
\end{equation} 
where $\gamma^A$, $\eta_B$ and $v_0^i$ are integration constants.  

For a given Calabi-Yau three-fold, the $\mathcal{G}_A$ are known (although complicated) functions of the complex structure moduli $X^A$.  Hence, the above equations implicitly determine the $z$-dependence of $X^A$. With these solutions eq.~\eqref{cons} turns into
\begin{equation}
\label{eq:cons}
-\gamma^A\epsilon_A+\eta_B\mu^B=0\:,
\end{equation}
that is, a condition on the integration constant which can be satisfied by a suitable choice of these constants\footnote{In ref.~\cite{Lukas:2010mf}, further constraints for the existence of a solution, in addition to \eqref{eq:cons}, are given. These arise due to the assumption that the flux components $\{\epsilon_0,\mu^0\}$ vanish, which is required for the half-flat compactifications discussed in ref.~\cite{Lukas:2010mf} but can be avoided for the Calabi-Yau compactifications discussed here.}. Finally, we need to discuss the dilaton equation~\eqref{dileq}. First, we note that, from eqs.~\eqref{vsol} and \eqref{Vdef}, the volume is given by $\V=\V_0e^{3\phi}$, where $\V_0$ is a constant explicitly given by $\V_0=d_{ijk}v_0^iv_0^jv_0^k/6$. Inserting this into the dilaton flow equation~\eqref{dileq} we obtain
\begin{equation}
\label{eq:eqdilaton}
\partial_ze^{-\phi}=\frac{1}{2\V_0}(\epsilon_A\,\textrm{Im}\;X^A+\mu^B\,\textrm{Im}\;\mathcal{G}_B)\; .
\end{equation}
With the explicit solutions for $X^A$ this leads to an explicit, although complicated first order differential equation for the $z$-dependence of the dilaton which can, at least in principle, be integrated.

In summary, we have established the existence of supersymmetric domain wall solutions for any choice of Calabi-Yau manifold and any harmonic flux on it, under the constraint \eqref{eq:10dconstraints}.

\subsection{Asymptotic Behaviour of Solutions}
Existence of solutions to the flow equations \eqref{eq:killingspinorCY:1}--\eqref{eq:killingspinorCY:3} is always guaranteed as we have demonstrated above. However, explicit integration requires detailed knowledge of the pre-potential and can only be done on a case-by-case basis. Still, we can deduce the properties of the solution in the limit of large $y$, that is, the behaviour of the fields $\{\phi,X^A, v^i\}$ far away from the domain wall.

To do this, we return to the flow equations \eqref{eq:killingspinorCY:1}--\eqref{eq:killingspinorCY:3} for a moment. Equation \eqref{eq:killingspinorCY:1} is equivalent to
\begin{equation}
\label{eq:2flow2}
\p_y(e^{-2\phi}\Psi_-)=e^{-2\phi}*H\:.
\end{equation}
Multiplying \eqref{eq:killingspinorCY:3} with $e^{-2\phi}$ and applying $\p_y$, we get using \eqref{eq:2flow2}
\begin{equation}
\label{eq:ddphi}
\p_y^2(e^\phi)=-\frac{1}{2e^{2\phi}\V_0}\int_XH\wedge*H=-\frac{1}{2e^{2\phi}\V_0}\vert\vert H\vert\vert^2\:,
\end{equation}
where we also have integrated over $X$. Note that \eqref{eq:ddphi} implies that the strictly positive function $e^\phi$ has a negative second order derivative. It then follows from \eqref{eq:ddphi} that $e^\phi$ must reach zero at least once along the flow, at the position of a domain wall. From the domain wall, this solution can either rise to a maximum before it drops again to zero, at the position of a new domain wall, or the solution approaches a linear function from below. 

The latter case is most interesting from a phenomenological point of view, describing an infinite universe bounded by a domain wall. It is the solution we will focus on. Here, the derivative $\p_ye^\phi$ approaches a constant from above. For non-vanishing flux this constant cannot be zero since eq.~\eqref{eq:ddphi} would then imply
\begin{equation*}
\lim_{y\rightarrow\infty}\vert\vert H\vert\vert^2=0\:.
\end{equation*}
As $H$ is constant, this can only be true if $H=0$. Hence, with non-vanishing flux on $X$, the dilaton $e^\phi$ approaches a linear increasing function as $y\rightarrow\infty$. The coupling $e^{-2\phi}$ therefore goes to zero, and we get the weak coupling limit at infinity. The generic $y$-dependence of $e^\phi$ and its derivative has been plotted below.
\begin{figure}[h!]
\begin{center}
\includegraphics[height=40mm]{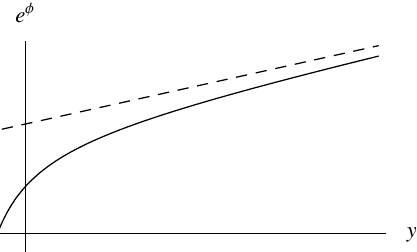}\;\;\;\;\;\;\;\;\;\;\;\;
\includegraphics[height=40mm]{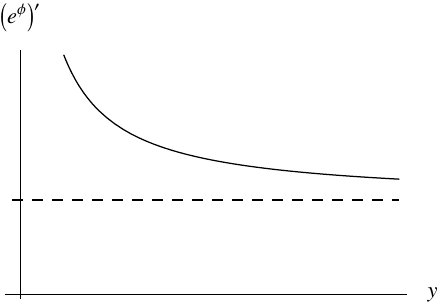}
\label{fig:asymptote}
\caption{\it Plot of the generic asymptotic behaviour of $e^\phi$ and its derivative $(e^\phi)'$ as $y\rightarrow\infty$.}
\end{center}
\end{figure}
Furthermore, from the definition of $z$ in eq.~\eqref{eq:variablechange} we see that in the limit $y\rightarrow\infty$ the behaviour of $e ^\phi$ implies that $z$ approaches a constant as $y\rightarrow\infty$. Accordingly, it follows that the rescaled fields $X^A$ approach constant values, while the original moduli $Z^A$ diverge. This means that the solution approaches the large complex structure limit far away form the domain wall, where the pre-potential can generically be approximated by
\begin{equation*}
\mathcal{G}(Z)=\K_{ABC}Z^AZ^BZ^C\:,
\end{equation*}
with intersection numbers $\K_{ABC}$. This observation allows us to check consistency with the results obtained in ref.~\cite{Lukas:2010mf}. Indeed, inserting the above form of the pre-potential into eqs.~\eqref{eq:killingmod1}--\eqref{cons} and, in addition, setting $\epsilon_0=\mu^0=0$, yields precisely the solution given in ref.~\cite{Lukas:2010mf}.

\section{Four-Dimensional Low Energy Theory}\label{sec:lowenergy}
We will now discuss the effective four-dimensional theories associated to Calabi-Yau domain wall solutions. These four-dimensional theories are covariant, $N=1$ supergravities, identical to the ones obtained from compactification on Calabi-Yau manifolds without flux, apart from the presence of a non-vanishing superpotential for the complex structure moduli. The domain wall can be recovered as a BPS-solution of the four-dimensional theory which couples to this superpotential. It should be mentioned that similar constructions have been found in the literature before, also from the $N=2$ supergravity and type II strings point of view. See in particular ref.~\cite{Curio:2000sc, Behrndt:2001qa, Behrndt:2001mx}. In particular, ref.~\cite{Behrndt:2001mx} has a solution to the BPS-equations which is along the same lines as the one we found in section \ref{sec:exist}.

We begin by reviewing the structure of the four-dimensional effective theory and its domain wall solution. Then we discuss the matching of the ten-dimensional Calabi-Yau domain wall solution, introduced in section 3, with the four-dimensional domain wall solution. We leave any technicalities for Appendix \ref{app:comparison} as they would distract from the main point of the section. Our discussion extends the results of ref.~\cite{Lukas:2010mf} where matching has been shown in the large complex structure limit. Here we find that these results can be extended to the entire moduli space. We also comment on the asymptotic behaviour found in the previous section, now from a four-dimensional perspective. 

\subsection{Four-Dimensional Effective Theory}\label{sec:4deff}
Upon dimensional reduction of the heterotic supergravity with a Calabi-Yau internal space $X_6$ one obtains a four-dimensional, $\mathcal{N}=1$ supergravity theory. It contains a set of chiral superfields $\Phi^X=(S,T^i,X^A)$ which correspond to the axio-dilaton $S=a+ i\: \,e^{-2\phi_4}$, the K\"ahler moduli $T^i$ and complex structure moduli $X^A$. The K\"ahler potential now reads
\begin{align}
\label{eq:kaehler4d}
K(\Phi^X,{\bar \Phi} ^ {\bar X})
=
K^{S}+K^{T}+K^{X} = - \log  i\: ({\bar S}-S)- \log(8\V) - \log i\:\Big[\mathcal{G}_B\bar X^B-X^A\mathcal{\bar G}_A\Big]\:.
\end{align}
Here $\V$ corresponds to the volume of $X_6$ and can be expressed in terms of it's intersection numbers $d_{ijk}$, i.e. $\V=\frac{1}{6}d_{ijk}t^i t^j t^k$ with $t^i = {\rm Im}\, T^i$. 

The superpotential of the theory is now given by
\begin{equation}
\label{eq:superpotential}
W=\sqrt{8}(\epsilon_A X^A+\mu^A\mathcal{G}_A),
\end{equation}
where $\mu^A, \epsilon_A$ are the flux parameters as defined in eq.~\eqref{eq:fluxexp}, $\mathcal{G}=\mathcal{G}(X^A)$ is the prepotential for the complex structure moduli and $\mathcal{G}_A=\frac{\partial}{\partial X^A}\mathcal{G}$ its derivatives. Recall that $\mathcal{G}$ is a homogeneous function of degree two. 

\subsection{Domain Wall Solution}
As we saw in section \ref{sec:4dDW}, and as also was shown in \cite{Lukas:2010mf}, the four-dimensional theories just described have $1/2$-BPS domain wall solutions with metric
\begin{align}
\label{eq:4ddomainwallmetric}
\mathrm{d}s^2=\e^{-2\phi(y)}\left(\eta_{\alpha\beta}\mathrm{d}x^\alpha\mathrm{d}x^\beta + \mathrm{d}y^2\right)
\end{align}
where $\eta_{\alpha\beta}\mathrm{d}x^\alpha\mathrm{d}x^\beta $ is the 1+2 dimensional Minkowski metric and $y=x^3$. With this metric, the Killing-spinor equations given by setting \eqref{eq:4dsusyvar1}-\eqref{eq:4dsusyvar2} to zero, reduce to
\begin{align}
\label{eq:4dKS}
\p_{y} \Phi^X=-ie^{-\phi_4}e^{K/2}K^{X\bar Y}D_{\bar Y}{\bar W}\:,
\end{align}
together with the constraint that the superpotential $W$ has to be purely imaginary and the axionic components of all fields are constant. Here, $D_{ Y}W = \partial_Y W + K_Y W$ as usual.

Furthermore, it was shown in ref.~\cite{Lukas:2010mf} that such four-dimensional domain wall solutions match their ten-dimensional counterparts, discussed in section~\ref{sec:cydomain10d}. The matching in ref.~\cite{Lukas:2010mf} was carried out only in the large complex structure limit, as is appropriate for the half-flat compactifications discussed there. For Calabi-Yau manifolds a restriction to large large complex structure is unnecessary. Fortunately, it turns out that this requirement is merely technical and that the matching can be shown to hold everywhere in complex structure moduli space. This is proven in Appendix \ref{app:comparison}, but we briefly review the procedure and results here.

\subsection{Comparing Four- and Ten-Dimensional Solutions}

The matching between the four- and ten-dimensional domain-wall vacua is carried out by showing that the respective Killing spinor equations are equivalent under appropriate field redefinitions. We merely present the results here, while the full proof is given in Appendix \ref{app:comparison}.

It turns out that for the matching to work we need to relate four- and ten-dimensional fields as
\begin{align}
\label{eq:10d4d:redef:1}
\e^{2\phi}&=\e^{2\phi_4}\V/\V_0
\\\label{eq:10d4d:redef:2}
Z^A&=\e^{2\phi}X^A
\\\label{eq:10d4d:redef:3}
v^i&=t^i
\:.
\end{align}
Note here the normalisation factor is generalised to $\F=e^{2\phi}$, which is proportional to $Z^0$ in the case of vanishing flux parameters, $\alpha_A=\beta^B=0$. This demonstrates that the low energy description of heterotic domain walls are given by the domain wall solutions of the $N=1$ four-dimensional supergravity theories discussed in the previous section. The matching holds everywhere in complex structure moduli space and for general harmonic flux, modulo the constraint \eqref{eq:10dconstraints}.

We also comment briefly on the large $y$-behaviour we found in the previous section. We saw that the fields $X^A$ stabilise at constant values as $y\rightarrow\infty$ far away from the domain-wall, where its influence is negligible. This is expected from a four-dimensional point of view, as we have introduced a superpotential for these fields. No such superpotential has been introduced for the dilaton or the K\"ahler moduli, which remain unstabilised.

As for the four-dimensional dilaton, we saw in the last section that the ten-dimensional dilaton diverges as $y\rightarrow\infty$. From \eqref{eq:10d4d:redef:1} and the fact that $V=e^{3\phi}V_0$, we see that
\begin{equation*}
\p_y\phi=-2\p_y\phi_4\; .
\end{equation*}
Hence, the four-dimensional dilaton goes to negative infinity, and we thus approach the weak coupling regime far away from the domain wall.

\section{Discussion}
\label{sec:conclCY}
In this chapter, we have shown that heterotic Calabi-Yau compactifications with flux exist, provided that we relax the condition of having a maximally symmetric four-dimensional space-time. Using a four-dimensional domain-wall ansats instead, we have found Calabi-Yau domain wall solutions for any harmonic flux and throughout complex structure moduli space. This extends previous results obtained in the large complex structure limit.

The main message is that harmonic NS flux is a legitimate ingredient in heterotic Calabi-Yau compactifications and can be added to the model without deforming the Calabi-Yau to a non-Kahler manifold. This means that the powerful set of model-building tools on Calabi-Yau manifolds is available, while NS flux can be added as a useful ingredient for moduli stabilisation.

Ultimately, the success of these models depends on the ability to lift these domain wall vacua to maximally symmetric ones which amounts to stabilising the remaining moduli, that is, the dilaton and the T-moduli. In the previous chapter, we saw that this can indeed be achieved in certain half-flat domain wall compactifications based on group coset spaces. Whether these results carry over to the present Calabi-Yau domain wall solutions is a subject of future study.

Another obvious generalization is to search for ten-dimensional heterotic solutions based on Calabi-Yau manifolds, harmonic flux and more general four-dimensional BPS-solutions, including, for example, four-dimensional cosmic string and black hole solutions. Especially, Calabi-Yau black hole solutions might be interesting in this context, as they might turn out to be consistent with the present universe without the need to ``lift" to a maximally symmetric four-dimensional space-time. Work in this direction is currently underway.

\begin{subappendices}
\section{Calabi-Yau Symplectic Geometry}\label{app:intro}
Here, we briefly summarize some useful facts about the symplectic geometry of Calabi-Yau moduli spaces~\cite{candelas1991moduli}. Overviews may be found, for example, in refs.~\cite{CyrilThesis, MatthiasThesis, Ceresole:1995ca}.

\subsection{Harmonic Expansion}
The K\"ahler form $\omega$ is expanded in a basis of harmonic $(1,1)$-forms
\begin{equation*}
J=v^i\omega_i,
\end{equation*}
where $\omega_i\in H^{(1,1)}(X)$. Likewise the harmonic (3,0)-form $\Psi$ is expanded as
\begin{equation*}
\Psi=Z^A\alpha_A-\mathcal{G}_B\beta^B,
\end{equation*}
where $\{\alpha_A,\beta^B\}\in H^3(X)$ is a real symplectic basis such that
\begin{equation*}
\int_X\alpha_A\wedge\beta^B=\delta_A^B,
\qquad
\int_X\alpha_A\wedge\alpha_B=\int_X\beta^A\wedge\beta^B=0.
\end{equation*}
The complex structure moduli space is a K\"ahler manifold, described by a holomorphic pre-potential $\mathcal{G}=\mathcal{G}(Z)$ which is a homogeneous function of degree two. Its derivatives are denoted by $\mathcal{G}_A=\p_A\mathcal{G}=\frac{\p\mathcal{G}}{\p Z^A}$.

Note that, in the context of the Calabi-Yau domain walls we discuss, the $SU(3)$-structure forms $\omega$ and $\Psi$ will depend on $y$, the direction normal to the domain wall. However, the basis forms $\{\omega_i\}$ and $\{\alpha_A,\beta^B\}$ are related to cycles of the Calabi-Yau manifold and are, hence, independent of $y$. Consequently, the $y$-dependence entirely resides in the moduli-fields $\{v^i,Z^A\}$.

\subsection{Some Symplectic Geometry and the Hodge-Dual}
As $\Psi$ is a $(3,0)$-form, we have
\begin{equation*}
*\Psi=-i\Psi.
\end{equation*}
The Hodge stars of the symplectic basis $\{\alpha_A,\beta^B\}$ are given by
\begin{align*}
*\alpha_A&={A_A}^B\alpha_B+B_{AB}\beta^B\\
*\beta^A&=C^{AB}\alpha_B+{D^A}_B\beta^B,
\end{align*}
where ${A_A}^B=-{D^A}_B$. These matrices may be written in terms of the matrix $N_{AB}$ given by
\begin{equation*}
N_{AB}=\mathcal{\bar G}_{AB}+2i\frac{\textrm{Im}(\mathcal{G}_{AC})Z^C\textrm{Im}(\mathcal{G}_{BD})Z^D}{\textrm{Im}(\mathcal{G}_{CD})Z^CZ^D}.
\end{equation*}
The corresponding expressions are
\begin{align}
\label{eq:A}
A&=(\textrm{Re}N)(\textrm{Im}N)^{-1}\\
\label{eq:B}
B&=-(\textrm{Im}N)-(\textrm{Re}N)(\textrm{Im}N)^{-1}(\textrm{Re}N)\\
\label{eq:C2}
C&=(\textrm{Im}N)^{-1}.
\end{align}

Next, we give some identities which will be useful in the next sections. We first define the complex structure K\"ahler potential
\begin{equation}
\label{eq:cplxK}
\K=\log\Big(\frac{i}{\F^2}\int\Psi\wedge\bar{\Psi}\Big).
\end{equation}
Next we define the parameters
\begin{equation}
\label{eq:fs}
f_A^B=\p_AZ^B+K_AZ^B=D_AZ^B.
\end{equation}
It may then be shown that the matrix $N$, the K\"ahler potential, the parameters $f_A^B$ and the pre-potential satisfy the following identities\footnote{Note that these identities do not depend on the rescaling $\F$ in \ref{eq:cplxK}.}
\begin{align}
\label{eq:usefulid1}
\K_{\bar BC}&=-\frac{1}{4\V}(\textrm{Im}N)_{DE}\bar f^{\bar D}_{\bar B}f^E_C\\
\label{eq:usefulid3}
\bar N_{\bar A\bar B}f^B_C&=\mathcal{G}_{AB}f^B_C\\
\label{eq:usefulid2}
(\textrm{Im}N_{AB})\bar f^{\bar A}_{\bar C}\bar Z^{\bar B}&=0.
\end{align}

\subsection{Hodge-Dual and $y$-Derivatives}\label{app:commute}
In this Appendix, we wish to show that we can assume
\begin{align}
\label{eq:commute1}
\p_z*\tilde\Psi=-*\p_z\tilde\Psi\:,
\end{align}
by an apropriate choice of ten-dimensional fields and coordinates. Here $z$ is defined by \eqref{eq:variablechange}, and $\tilde\Psi\propto\Psi$ with proportionality factor to be defined below. We also let $*$ denote the six-dimensional Hodge dual on the Calabi-Yau manifold. This relation will be useful for proving the results in Appendix \ref{app:comparison}.

Note first that
\begin{equation}
\label{eq:Covpsi}
D_y\Psi=\p_y\Psi+\K_y\Psi
\end{equation}
is a primitive $(2,1)$-form. Here $\K_y=\p_y\K$. Hence, by \eqref{eq:Weil}, we have
\begin{equation*}
*D_y\Psi=iD_y\Psi\:.
\end{equation*}
We may rewrite \eqref{eq:Covpsi} as
\begin{equation*}
D_y\Psi=e^{-\K}\p_y(e^{\K}\Psi)=\p_{\tilde y}\tilde\Psi\:,
\end{equation*}
where $\tilde\Psi=e^{\K}\Psi$, and the new coordinate $\tilde y$ is defined by
\begin{equation*}
\frac{\p\tilde y}{\p y}=e^{-\K}\:.
\end{equation*}
With appropriate choice of coordinates and fields, we may take $\tilde y=z$ and $\tilde\Psi=e^{-2\phi}\Psi$.

\section{Matching Ten- and Four-Dimensional Equations}
\label{app:comparison}
In this Appendix, we would like to show that the Killing-Spinor equations in 10 and four dimensions match, under a suitable field redefinition. 

Let us start by clearly stating the field redefinitions which will be necessary to relate both solutions. The dilaton $\phi$, K\"ahler moduli $v^i$ and complex structure moduli $Z^A$ of the ten-dimensional theory are related to the respective fields, $\phi_4$, $t^i$, $X^A$ of the four-dimensional theory via
\begin{align}
\label{eq:10d4d:redefinitions:1}
\e^{2\phi}&=\e^{2\phi_4}\V/\V_0
\\\label{eq:10d4d:redefinitions:2}
Z^A&=\e^{2\phi}X^A
\\\label{eq:10d4d:redefinitions:3}
v^i&=t^i
\:,
\end{align}
where again $V$ is the volume of the Calabi-Yau manifold $X_6$ and $V_0$ is some fixed reference volume. From now on, we set $\V_0=1$ for convenience. With these identifications, the equations for the K\"ahler moduli \eqref{eq:killingspinorCY:2} and \eqref{eq:4dKS} can be easily confirmed to match, in complete analogy to the proof in ref.~\cite{Lukas:2010mf}. 

Let us now demonstrate the matching of the Killing spinor equations for the dilaton whose four-dimensional version \eqref{eq:4dKS} becomes
\begin{align}
\label{eq:dilatonflow4d:generic1}
\p_y\phi_4=\frac{ i\:\,\e^{2\phi_4}}{4}W\; .
\end{align}
Here, we have used the K\"ahler potential and superpotential from eqs.~\eqref{eq:kaehler4d} and \eqref{eq:superpotential}. From the relation \eqref{eq:10d4d:redefinitions:1} between the 10- and four-dimensional dilaton, and the $y$-dependence of the volume 
\begin{equation*}
\p_y\V=3\p_y\phi \V\; ,
\end{equation*}
implied by eq.~\eqref{eq:killingspinorCY:2}, it follows that $\partial_y\phi=-2\partial_y\phi_4$. With the last relation it can be easily seen that \eqref{eq:dilatonflow4d:generic1} matches the ten-dimensional  dilaton equation \eqref{eq:killingspinorCY:3}, upon integrating the latter equation over $X$.

Next, let us show the matching of the Killing spinor equations for the complex structure moduli. To see this, we start with the ten-dimensional equation~\eqref{eq:killingspinorCY:1}, which can be written as
\begin{align}
\label{eq:OmegaH}
\p_y\Psi=2(\p_y\phi)\,\Psi - (H - i\, *H)\:,
\end{align}
where we have used \eqref{eq:commute1}. If we expand $\Psi$ and $H$ with respect to a symplectic basis $(\alpha_A,\beta^A)$ as before, that is,
\begin{align}
\label{eq:omegaHexpansions}
\Psi=Z^A\left(\alpha_A - \mathcal{G}_{AB}\beta^B\right)\; ,\qquad
H=\mu^A\alpha_A + \epsilon_A\beta^A
\; ,
\end{align}
we can turn \eqref{eq:OmegaH} into an the equation
\begin{equation}
\p_y(e^{-2\phi}Z^A)=-e^{-2\phi}(\mu^A-i\tilde\mu^A).
\label{eq:10dflow}
\end{equation}
for the complex structure moduli $Z^A$.  Here $\tilde\mu^A=C^{AB}\epsilon_B+A_B^A\mu^B$. With the complex structure K\"ahler potential \eqref{eq:cplxK}, equation \eqref{eq:10dflow} can be written in terms of complex structure moduli space geometry as 
\begin{align*}
\p_y(e^{-2\phi}Z^A)
&=
-e^{-2\phi}\K^{A\bar B}\K_{\bar B C}C^{CB}\Big((C^{-1})_{BD}\mu^D
-i\epsilon_B-i(\textrm{Re}N)_{BD}\mu^D\Big)
\\
&=
-i\frac{e^{-2\phi}}{4\V}\K^{A\bar B}(\textrm{Im}N)_{DE}
\bar f^{\bar D}_{\bar B}
f^E_CC^{CB}
\Big(
i(C^{-1})_{BD}\mu^D+\epsilon_B+(\textrm{Re}N)_{BD}\mu^D
\Big),
\end{align*}
where the first equality follows from ${A_C}^A=C^{AB}(\textrm{Re}N)_{BC}$, and for the second equality we have used equation of \eqref{eq:usefulid1}. Using \eqref{eq:usefulid2}, and the fact that $C=(\textrm{Im}N)^{-1}$, we see that
\begin{align}
\label{eq:Killing10}
\p_y(e^{-2\phi}Z^A)&=-\frac{ie^{-2\phi}}{4\V}\K^{A\bar B}\bar f^{\bar C}_{\bar B}\Big(\epsilon_C+N_{CD}\mu^D\Big)\notag\\
&=-\frac{ie^{-2\phi}}{4\V}\K^{A\bar B}\bar f^{\bar C}_{\bar B}\Big(\epsilon_C+\mathcal{\bar G}_{\bar C\bar D}\mu^D\Big),
\end{align}
where in the last equality we have used \eqref{eq:usefulid3}.

We want to compare this the to the $4d$ Killing spinor equation \eqref{eq:4dKS} for the moduli $X^A$, which reads 
\begin{align}
\label{eq:4dKSeq}
\p_{y}X^A=-\frac{i}{4} \e^{2\phi_4} K^{A\bar B}D_{\bar{B}}\bar{W}
=
-\frac{i}{4} \e^{2\phi_4} K^{A\bar B}D_{\bar{B}}X^{\bar{C}}(\epsilon_{\bar C}+\mathcal{\bar G}_{\bar C\bar D}\mu^{\bar D})
\:.
\end{align}
If we now use \eqref{eq:10d4d:redefinitions:2} to re-express all $\partial/\partial X^A$ derivatives into $\partial/\partial Z^A$ and the fact that $f_A^B=\p_A Z^B+\p_A \K Z^B=D_A Z^B$, then we see that in fact \eqref{eq:Killing10} and \eqref{eq:4dKSeq} are equal.

Note also that the constraint \eqref{eq:10dconstraints} gives rise to a purely imaginary superpotential by the Gukov-Vafa-Witten formula
\begin{align*}
W\propto\int_{X_6}H\wedge\Psi\:,
\end{align*}
as required by the four-dimensional theory.

\end{subappendices}

\part{Conclusions}
\label{part:concl}


\chapter{Conclusion}
In this thesis, we have considered heterotic supergravity at $\OO(\a)$ and above. We have mainly been concerned with moduli associated to this theory, in relation to compactifications to a four-dimensional non-compact space-time, which is not necessarily maximally symmetric. We devided the thesis in two parts: Maximally symmetric (Part~\ref{part:4d}) and non-maxiamlly symmetric compactifications (Part~\ref{part:3d}).

In part Part~\ref{part:4d} we considered compactifications of the theory to Minkowski space-time at $\OO(\a)$, commonly referred to as the Strominger system. Specifically, in chapter \ref{ch:SS}, we where interested in the moduli space, or spectrum, of such compactifications. As noted in the introduction, knowledge of the full spectrum of the Strominger system has for a long time been lacking. Although the spectrum has been known for the zeroth order Calabi-Yau solutions since the late 80's/early 90's, how to include $\a$-effects have remained an elusive problem up until now. Specifically, the non-K\"ahlarity of the compact space and the non-trivial Bianchi identity complicates matters greatly.

We made progress in this direction by showing that the system can be rephrased in terms of a holomorphic structure $\bD$ over some generalised bundle $\Q$ over the compact space $X$. We saw that the issues concerning the non-K\"ahlerness and non-trivial Bianchi identity where naturally included as part of the construction of $\bD$. In particular, it is a holomorphic structure, i.e. $\bD^2=0$, iff the Bianchi identities are satisfied. Using this, we where able to compute the infinitesimal moduli space as the first cohomology of this bundle,
\begin{equation*}
T\M=H^{(0,1)}_{\bD}(\Q)\:.
\end{equation*}
The bundle $\Q$ was constructed by a series of extensions, which the holomorphic structure $\bD$ respects. $H^{(0,1)}_{\bD}(\Q)$ could then be computed by the machinery of long exact sequences in cohomology, giving it as a subset of the usual cohomologies as expected. In computing this cohomology, we encountered extra ``moduli", $H^{(0,1)}(\End(TX))$, which should not correspond to physical fields in the low energy theory. The purpose of chapter \ref{ch:connredef} was to give these the correct interpretation, and we saw that they correspond to $\OO(\a)$ field redefinitions of the theory.

Knowing what the spectrum of the theory is is the first step in understanding the corresponding four-dimensional $N=1$ supergravity, the goal of heterotic phenomenology. As noted in the discussion section of chapter \ref{ch:SS}, in order to achieve this, knowledge of the K\"ahler potential, Yukawa couplings and superpotential is also needed. To obtain these, a dimensional reduction of heterotic supergravity on general heterotic $SU(3)$-structure manifolds needs to be performed. To do this properly, a lot more knowledge is needed about the moduli space of such compactifications in general. What are the obstructions, K\"ahler metric on the moduli space, its geometrical structure, etc? These are all interesting directions to pursue, in particular, in terms of viewing the system as a holomorphic structure $\bar D$. Indeed, a lot of knowledge of holomorphic structures on bundles and their moduli spaces already exists in the mathematics literature, which would be interesting to generalise to the Strominger system.

We also made the curious observation that the Strominger system can be rewritten as a hermitian Yang-Mills connection $\D$ on the extension bundle $\Q$. We proposed that this condition can be derived from an instanton type condition
\begin{equation*}
F_{\D AB}\Gamma^{AP}\epsilon=0\:.
\end{equation*}
Interestingly, such a condition also makes sense {\it outside the large volume limit}, and it therefore has a shot of describing the correct supersymmetric solutions in this regime. It would be interesting to see if an action with such corresponding supersymmetric solutions can be written down for the connection $\D$.

Note also that throughout this thesis we have omitted loop corrections from the world-sheet perspective, i.e. $g_s$-corrections. To include these, more knowledge of the full heterotic world-sheet model is needed. The generic heterotic string theory has $(0,1)$-supersymmetry, which allows for more flexibility then the corresponding type II theories, making it much harder to compute quantum corrections. Putting the target space in terms of a single hermitian Yang-Mills connection on $\Q$ might then prove useful in this direction, potentially simplifying the problem. Indeed, world-sheet models with such holomorphic bundles as target spaces have been studied in the literature before. It would be interesting to extend these results to $\bD$, and work in this direction is underway.

In Part~\ref{part:3d} of the thesis, we turned to discussing moduli stabilisation. As has been established, moduli stabilisation in heterotic theory is a hard problem, due to the lack of RR fluxes, and the non-K\"ahler spaces that arise when including NS flux. We saw that we could remedy this problem, by allowing for compactifications to a non-maxiamlly symmetric space-time
\begin{equation*}
X\rightarrow M_{10}\rightarrow M_4\:,
\end{equation*}
where $M_{4}$ is four-dimensional non-compact space-time, and $X$ is the compact space, non-trivially fibered over $M_4$. Specifically, we considered heterotic compactifications on space-times known as domain walls, and showed that such compactifications allowed for the internal geometry to be non-complex and torsional. $X$ then had an $SU(3)$-structure known as generalised half-flat. The torsion induced K\"ahler moduli dependent terms in the superpotential, similar to the flux terms induced for complex structure moduli. We saw that these could be used to stabilise geometric moduli.

Specifically, in chapter \ref{ch:HFNK} we considered compactifications on half-flat mirror manifolds, focusing on coset spaces $G/H$, where the geometry is easily expressed in terms of $G$-invariant forms. We showed that, although not possible at zeroth order, the first order geometry allows for all geometric moduli to be stabilised, leaving only a coordinate dependent dilaton. We further studied the corresponding four-dimensional theory, showing that by inclusion of non-perturbative effects, it is possible to lift the vacuum to a maximally symmetric supersymmetric AdS vacuum in a consistent region of moduli space. In light of this, it seems that much more general $SU(3)$-structure compactifications can be considered when discussing heterotic compactifications, more so than what is obtained by the usual Minkowski compactifications mentioned above.

Next, we considered Calabi-Yau compactifications with flux in chapter \ref{ch:CY}. Though we made some progress in Part~\ref{part:4d} of the thesis, mathematically, a lot more is still known about Calabi-Yau's then the heterotic $SU(3)$-structures described above. Compactifications on Calabi-Yau manifolds is therefore more desirable from a mere pragmatic point of view, and it has been the most useful setting for model-building this far. However, as we pointed out, such maximally symmetric compactifications suffer from no-go theorems when including fluxes. We showed that it is possible to evade these theorems, provided again that the maximal symmetric space-time assumption is sacrificed. It follows that fluxes can be used as an extra ingredient when it comes to moduli stabilisation in Calabi-Yau compactifications and model building. This is of course provided that the corresponding solutions can be lifted somehow by use of e.g. non-perturbative effects, as was done in chapter \ref{ch:HFNK} for the torsional compactifications. It would be interesting to investigate if this can be done consistently, and a viable maximally symmetric vacuum with all moduli stabilised can be found. Work in this direction is currently underway.

The main point of Part~\ref{part:3d} of the thesis was then that more general compactifications then the traditional zero flux Calabi-Yau, or it's $\a$-generalisation to heterotic $SU(3)$-structure manifolds, can be employed for heterotic model building, allowing for far more flexibility in moduli stabilisation. This also prompts a lot of future directions to be explored. In particular, as this thesis was only concerned with domain wall solutions, it would be interesting to investigate more generic solutions, such as cosmic strings and black holes, potentially leading to even more general compactifications. Furthermore, the four-dimensional supergravities, appearing from compactifications on half-flat mirror manifolds as described in section \ref{sec:4ddimred}, are far from being fully explored. It would be interesting to extend this analysis and consider more generic half-flat mirror manifolds, allowing for a broader set of torsion parameters and fluxes. It would also be interesting to extend the search, and look for non-supersymmetric vacua of these theories.



{\small

\addcontentsline{toc}{chapter}{Bibliography}
\bibliographystyle{JHEP}
\bibliography{BibliographyXD}        
}

\end{document}